\documentclass{article}

\usepackage{amsmath,amssymb,amsthm,graphicx,url}
\usepackage{vmargin}

\setpapersize{A4}

\newcommand{\g}[1]{\boldsymbol{#1}}

\newcommand{\R}[0]{\mathbb{R}} 

\newtheorem{theorem}{Theorem}
\newtheorem{lemma}{Lemma}

\newtheorem{definition}{Definition}
\newtheorem{proposition}{Proposition}
\newtheorem{corollary}{Corollary}
\newtheorem{assumption}{Assumption}

\title{\bf Finding sparse solutions of systems of polynomial equations via group-sparsity optimization}

\author{F. Lauer$^1$ and H. Ohlsson$^{2,3}$
\medskip\\
$^1$ \small LORIA, Universit\'e de Lorraine, CNRS, Inria, France\\
$^2$ \small Dept. of Electrical Engineering and Computer Sciences,
University of California, Berkeley, USA\\
$^3$ \small Dept. of Electrical Engineering, Link\"oping University, Sweden
}

\begin{document}

\maketitle

\begin{abstract}
The paper deals with the problem of finding sparse solutions to systems of polynomial equations possibly perturbed by noise. In particular, we show how these solutions can be recovered from group-sparse solutions of a derived system of linear equations. Then, two approaches are considered to find these group-sparse solutions. The first one is based on a convex relaxation resulting in a second-order cone programming formulation which can benefit from efficient reweighting techniques for sparsity enhancement. For this approach, sufficient conditions for the exact recovery of the sparsest solution to the polynomial system are derived in the noiseless setting, while stable recovery results are obtained for the noisy case. Though lacking a similar analysis, the second approach provides a more computationally efficient algorithm based on a greedy strategy adding the groups one-by-one. With respect to previous work, the proposed methods recover the sparsest solution in a very short computing time while remaining at least as accurate in terms of the probability of success. This probability is empirically analyzed to emphasize the relationship between the ability of the methods to solve the polynomial system and the sparsity of the solution. 
\end{abstract}

\section{Introduction}

When faced with an underdetermined system of equations, one typically applies a regularization strategy in order to recover well-posedness. The choice of regularization depends on the particular application at hand and should be made to drive the solution towards desired properties. In the absence of precise goals, the most popular choice favors solutions with minimum $\ell_2$-norm. However, an alternative becoming more and more popular is to search for sparse solutions (which often have non-minimal $\ell_2$-norm). 
In the case of a system of linear equations, this alternative has been investigated in numerous works, see, e.g., \cite{Bruckstein09} for a review, and entails many applications of great importance, particularly for signal processing where it goes under the name of compressed sensing/sampling \cite{Donoho06,Candes06}. Formally, finding the sparsest solutions of linear systems can be written as the minimization, under the constraints of the linear system, of the number of nonzero variables, which is a nonsmooth, nonconvex and NP-hard problem. Two main approaches can be distinguished to tackle such problems. The first one, known as Basis Pursuit (BP), relies on a convex relaxation based on the minimization of an $\ell_1$-norm, while the second one applies a greedy strategy to add nonzero variables one-by-one.  Many results on the convergence of these methods to the sparsest solution are available in the literature \cite{Bruckstein09,Eldar12,Foucart13}. 

More recently, a few works introduced extensions of this problem to nonlinear equations. In particular, the first greedy approaches appeared in \cite{Blumensath08,Blumensath13,Beck13,Ehler13}, while BP algorithms were developed in \cite{Ohlsson13quadratic} to find sparse solutions of systems of quadratic equations 
and in \cite{Ohlsson13nlbp} for more general nonlinear equations.
Formally, these problems can be formulated as
\begin{align}\label{eq:P0}
	\min_{\g x\in\R^n} \ & \|\g x\|_0 \\
	\mbox{s.t.}\ & y_i = f_i(\g x),\quad i=1,\dots,N, \nonumber
\end{align}
where $y_i\in \R$, $f_i : \R^n \rightarrow \R$ are nonlinear functions and $\|\g x\|_0 = \left|\{j\in\{1,\dots, n\} : x_j \neq 0\}\right|$ denotes the $\ell_0$-pseudo-norm of $\g x$, i.e, the number of nonzero components $x_j$. 

Here, we focus on the case where the $f_i$'s in~\eqref{eq:P0} are polynomial functions of maximal degree $d$:
\begin{equation}\label{eq:polyf}
	y_i = f_i(\g x) = p_i^d(\g x) = b_{i} + \sum_{k=1}^M a_{ik} \g x^{\alpha_k} ,\quad i=1,\dots,N,
\end{equation}
where $\{\alpha_k\}_{k=1}^M$ is the set of $M = \sum_{q=1}^d\begin{pmatrix}n + q-1\\q\end{pmatrix}$ multi-indexes with $1\leq |\alpha_k| = \sum_{j=1}^n (\alpha_k)_j \leq d$ and $\g x^{\alpha_k} = \prod_{j=1}^n x_j^{(\alpha_k)_j}$ are the corresponding monomials.
This includes the particular case considered in \cite{Ohlsson13quadratic} with $d=2$, while in \cite{Ohlsson13nlbp} this setting is used to deal with the more general case via the Taylor expansions of the nonlinear functions $f_i$. Note that the formulation in~\eqref{eq:P0}--\eqref{eq:polyf} also entails cases outside of the quasi-linear setting considered in \cite{Ehler13}.

The present paper proposes a new BP approach to solve \eqref{eq:P0}--\eqref{eq:polyf}, which, in comparison with the previous works \cite{Ohlsson13quadratic,Ohlsson13nlbp}, is more simple and amounts to solving an easier optimization problem. More precisely, the method proposed in Sect.~\ref{sec:ell1} relies on a simple change of variable  sufficient to result in an efficient algorithm implemented as a classical $\ell_1$-minimization problem. However, the structure of the polynomials is discarded and the solution may not satisfy the original polynomial constraints. Note that this change of variable is also found in \cite{Ohlsson13nlbp}, though with constraints handled in a more complex manner, and closely related to the lifting technique of \cite{Ohlsson13quadratic} for quadratic BP and the one of \cite{Vidal05} for subspace clustering.

The method in \cite{Ohlsson13quadratic} enforces the structure of the quadratic polynomials via rank constraints, which leads to optimization problems with additional levels of relaxation and the introduction of a parameter tuning the trade-off between the minimization of the $\ell_1$-norm on the one hand and the satisfaction of the rank constraint on the other hand. 
In \cite{Ohlsson13nlbp}, the structure of higher-degree polynomials is enforced by a set of quadratic constraints on the monomials\footnote{For example, consider the monomials $u_1=x_1$, $u_2=x_2$, $u_3=x_1x_2$, $u_4=x_1^2 x_2$, then the structure of $u_3$ and $u_4$ is enforced by $u_3=u_1u_2$ and $u_4 = u_1u_3$. But note that since these constraints are relaxed in the final formulation, the estimation can yield $u_3\neq u_1u_2$, which then recursively implies that all monomial constraints involving $u_3$ are meaningless.}, which are then relaxed as in \cite{Ohlsson13quadratic}. 

Instead, the method proposed in Sect.~\ref{sec:groupsparse} implements the structural knowledge via group-sparsity. This results in an $\ell_1/\ell_2$-minimization, i.e., a second-order cone program (SOCP) which can be solved more efficiently. In addition, this SOCP formulation is easily extended in Sect.~\ref{sec:weighting} to benefit from reweighting techniques for improving the sparsity of the solution. 

The conditions for exact recovery of the proposed BP methods are analyzed in Sect.~\ref{sec:analysis}. In particular, we show that though the simple $\ell_1$-minimization does not include structural constraints, exact recovery occurs for sufficiently sparse cases. A similar condition is proved for the $\ell_1/\ell_2$-minimization based on group-sparsity.

The greedy approach is discussed in Sect.~\ref{sec:greedy}, where two variants are proposed: an exact algorithm for solving the group-sparsity optimization problem in small-scale cases and an approximate one which remains efficient in higher-dimensional settings. 
Previous greedy approaches \cite{Blumensath13,Beck13} considered the problem in its sparsity-constrained form, where the sum of squared errors over the equations \eqref{eq:polyf} is minimized subject to $\|\g x\|_0\leq s$ for an {\em a priori} fixed $s$. The iterative hard thresholding of \cite{Blumensath13} is a gradient descent algorithm with an additional projection onto the feasible set at each iteration via a simple thresholding operation. In \cite{Beck13}, this is interpreted as a fixed point iteration enforcing a necessary (but not sufficient) optimality condition, which requires a well-chosen step-size to converge satisfactorily. 
In particular, the step-size must be an upper bound on the Lipschitz constant of the gradient of the objective function, which however is not (globally) Lipschitz continuous for polynomials in \eqref{eq:polyf} with degree $d>1$. 
The sparse-simplex algorithm of \cite{Beck13} is a coordinate descent method which enjoys similar but less restrictive convergence properties while being parameter-free. However, for polynomial equations as in \eqref{eq:polyf}, each iteration requires solving several one-dimensional  minimizations of a polynomial of degree $2d$, which becomes difficult for $d>2$. On the contrary, the proposed greedy algorithms remain simple thanks to the group-sparse and linearized formulation of the problem: each iteration requires only solving least-squares problems. 

Extensions of the methods are discussed in Sections~\ref{sec:pureNL} and~\ref{sec:noise}. In particular, Section~\ref{sec:pureNL} deals with the issue of purely nonlinear polynomials, for which the solution to \eqref{eq:P0}--\eqref{eq:polyf} cannot be estimated as directly as for polynomials involving linear monomials, but which arise in important applications such as phase retrieval \cite{Kohler72,Gonsalves76,Gerchberg72,Fienup82,Marchesini07,Candes13,Candes13b}. Then, the case where the equations~\eqref{eq:polyf} are perturbed by noise is considered in Sect.~\ref{sec:noise}. In particular, the analysis provides stable recovery results for polynomial BP denoising in the flavor of the one obtained in \cite{Donoho06noise} for linear BP denoising.

Finally, numerical experiments in Sect.~\ref{sec:exp} show the efficiency of the proposed methods and extensions. Results are in line with the ones found in classical sparse optimization with linear constraints. In particular, all methods can recover the sparsest solution in sufficiently sparse cases and the greedy approach is the fastest while the BP methods based on convex relaxations benefit from a slightly higher probability of success. 

\section{Polynomial basis pursuit}
\label{sec:polyBP}

This section develops two basis pursuit approaches to find sparse solutions of systems of polynomial equations via the minimization of an $\ell_1$-norm for the first one and of a mixed $\ell_1/\ell_2$-norm for the second one. The methods are developed under the following assumption, which will be relaxed in Sect.~\ref{sec:pureNL}. 
\begin{assumption}\label{ass:nozerocol1}
For all $j\in\{1,\dots,n\},\ \exists i\in\{1,\dots,N\}$ such that $a_{ik}\neq 0$ for $k$ determined such that $\left(\alpha_{k}\right)_l = 1$ for $l=j$ and $\left(\alpha_{k}\right)_l =0$ for $l\in \{1,\dots,n\}\setminus \{j\}$.
\end{assumption}
In particular, Assumption~\ref{ass:nozerocol1} ensures that the polynomials include a linear part, or more precisely, that for any variable $x_j$, $j=1,\dots,n$, the monomial $x_j$ has a nonzero coefficient in at least one of the  polynomials.

\subsection{Polynomial basis pursuit via $\ell_1$-minimization} 
\label{sec:ell1}

We start by rewriting the constraints in \eqref{eq:P0}--\eqref{eq:polyf} as 
\begin{equation}\label{eq:linearize}
	y_i = f_i(\g x) = b_{i} + \sum_{k=1}^M a_{ik} \g x^{\alpha_k} 
	 =\g a_i^T \phi( \g x) + b_{i} ,
\end{equation}
where $\g a_i = [a_{i1},\dots,a_{iM}]^T$ and $\phi : \R^n \rightarrow \R^M$ is a mapping that computes all the monomials of degree $q$, $1\leq q\leq d$, with $n$ variables. 
Note that $\g x$ is embedded in $\phi(\g x)$ as the components corresponding to the $n$ multi-indexes $\alpha_{k(j)}$, $j=1,\dots,n$, such that $\left(\alpha_{k(j)}\right)_l = 1$ for $l=j$ and $0$ for $l\neq j$. Assuming that these are the first components of $\phi(\g x)$, i.e., $k(j)=j$, we also define the (linear) inverse mapping $\phi^{-1} : \R^M \rightarrow \R^n$, such that $\phi^{-1} (\phi(\g x))=\g x$, as
$$
	\phi^{-1}(\g \phi) = \begin{bmatrix}\g I_n & 0 \end{bmatrix} \g \phi,
$$
where $\g I_n$ is the $n$-by-$n$ identity matrix. 

The high-dimensional lifting by $\phi$ allows us to recast \eqref{eq:P0}--\eqref{eq:polyf} as a standard (i.e., linear) problem in sparse optimization:
\begin{align}\label{eq:P0poly}
	\min_{\g \phi \in\R^M} \ & \|\g \phi\|_0 \\
	\mbox{s.t.}\ &  \g A\g \phi = \g y - \g b ,\nonumber
\end{align}
where $\g A = [\g a_1, \dots, \g a_N]^T$, $\g b=[b_1,\dots,b_N]^T$ and $\phi(\g x)$ is replaced by an unstructured vector $\g \phi$. This constitutes the first level of relaxation in the proposed approach, where the components of $\g \phi$ are not constrained to be interdependent monomials. While this yields a rather crude approximation, it will serve as the basis for the refined approach of Sect.~\ref{sec:groupsparse}, where additional structure will be imposed on $\g \phi$. 

The second level of relaxation comes from the BP approach, in which problems such as~\eqref{eq:P0poly} are typically solved via the convex relaxation 
\begin{align}\label{eq:P1poly}
	\hat{\g \phi} = \arg\min_{\g \phi \in\R^M} \ & \|\g W \g \phi\|_1 \\
	\mbox{s.t.}\ &  \g A\g \phi = \g y - \g b , \nonumber
\end{align}
where $\g W = diag\left(\|\g A_1\|_2,\dots, \|\g A_M\|_2 \right)$, with $\g A_k$ the $k$th column of $\g A$, is a diagonal matrix of precompensating weights. 
Then, for polynomial BP, an estimate of the solution to \eqref{eq:P0}--\eqref{eq:polyf} can be easily computed under Assumption~\ref{ass:nozerocol1} as $\hat{\g x} = \phi^{-1}(\hat{\g \phi})$.

Problem~\eqref{eq:P1poly} can be formulated as a linear program and solved by generic interior-point algorithms in polynomial time, while it can also be solved by more efficient and dedicated algorithms, see, e.g., \cite[Chapter~15]{Foucart13}. However, the complexity with respect to $n$ remains exponential because of the dimension $M$ which scales exponentially with $n$. 

The literature on basis pursuit and compressed sensing \cite{Bruckstein09,Foucart13} tells us that the sparser the solution to \eqref{eq:P0poly} is, the more likely the convex relaxation \eqref{eq:P1poly} is to yield its recovery. Here, by construction we know that there is at least a very sparse vector $\phi(\g x_0)$ satisfying the constraints:  with $\g x_0$ the sparse solution to \eqref{eq:P0}--\eqref{eq:polyf}, the sparsity level $\|\phi(\g x_0)\|_0 / M$ is better than  $\|\g x_0\|_0 / n$, as stated by Proposition~\ref{prop:phisparsity} below. Note that a better bound implying an increased level of sparsity depending on $d$ for very sparse cases is derived in Appendix~\ref{app:phisparsity}.
\begin{proposition} \label{prop:phisparsity}
Let the mapping $\phi : \R^n \rightarrow \R^M$ be defined as above. Then, the vector $\phi(\g x)$ is at least as sparse as the vector $\g x$ in the sense that the inequality
$$
	\frac{\|\phi(\g x)\|_0}{M} \leq \frac{ \|\g x\|_0 }{n} 
$$
holds for all $\g x\in\R^n$. 
\end{proposition} 
Proposition~\ref{prop:phisparsity} implies that if~\eqref{eq:P0}--\eqref{eq:polyf} has a sparse solution then so does~\eqref{eq:P0poly}. 
To prove Proposition~\ref{prop:phisparsity}, we first need the following lemma. 
\begin{lemma} \label{lem:binoms}
For all triplet $(a,b,c) \in (\mathbb{N}^*)^3$ such that $a\geq b$, the inequality 
$$
	\frac{1}{a} \begin{pmatrix}a + c-1\\c\end{pmatrix} \geq \frac{1}{b} \begin{pmatrix}b + c-1\\c\end{pmatrix} 
$$
holds.
\end{lemma}
\begin{proof}
On the one hand, we have
$$
	\frac{1}{a} \begin{pmatrix}a + c-1\\c\end{pmatrix} = \frac{(a+c-1)!}{a\, c! (a-1)!} = \frac{1}{a\,c!}\prod_{i=0}^{c-1}(a+i) 
	= \frac{1}{c!}\prod_{i=1}^{c-1}(a+i)
$$
and a similar expression with $b$ instead of $a$. 
On the other hand, with $a\geq b$, we have
$$
	\forall i\in\mathbb{N},\ a+i\geq  b + i\quad \Rightarrow\quad  \frac{1}{c!}\prod_{i=1}^{c-1}(a+i) \geq \frac{1}{c!}\prod_{i=1}^{c-1}( b +i) 
$$ 
which then yields the sought statement. 

\end{proof}
We now give the proof of Proposition~\ref{prop:phisparsity}.
\begin{proof}
By construction, the number of nonzeros in $\phi(\g x)$ is equal to the sum over $q$, $1\leq q\leq d$, of the number of monomials of degree $q$ in $\|\g x\|_0$ variables. This yields
$$
	\frac{\|\phi(\g x)\|_0}{M} =  \frac{ \sum_{q=1}^d\begin{pmatrix}\|\g x_0\|_0 + q-1\\q\end{pmatrix}}{ \sum_{q=1}^d\begin{pmatrix}n + q-1\\q\end{pmatrix}} 
	=  \frac{ \|\g x\|_0 }{n} \frac{  \frac{1}{\|\g x\|_0 } \sum_{q=1}^d\begin{pmatrix}\|\g x_0\|_0 + q-1\\q\end{pmatrix}}{  \frac{1 }{n} \sum_{q=1}^d\begin{pmatrix}n + q-1\\q\end{pmatrix}} 
$$
Since $n\geq \|\g x\|_0$, Lemma~\ref{lem:binoms} gives the bound
$$
	 \frac{1}{\|\g x\|_0 } \sum_{q=1}^d\begin{pmatrix}\|\g x_0\|_0 + q-1\\q\end{pmatrix}\leq   \frac{1 }{n} \sum_{q=1}^d\begin{pmatrix}n + q-1\\q\end{pmatrix}
$$
from which the statement follows. 

\end{proof}

\subsection{Polynomial basis pursuit via $\ell_1$/$\ell_2$-minimization (group sparsity)} 
\label{sec:groupsparse}

The approach proposed above relies on a rather crude approximation. Thus, the polynomial equations are not guaranteed to be satisfied by the solution $\hat{\g x}$, due to the factorization of the polynomial that may not hold for $\hat{\g \phi}$. In other words, the linearization in \eqref{eq:linearize} together with the direct optimization of $\g \phi$ discard the desired structure for $\g \phi$. 
In practice, the solution $\hat{\g x}$ can be checked a posteriori with $y_i \overset{?}{=} p_i^d(\hat{\g x})$, $i=1,\dots,N$. 
But in order to increase the probability of obtaining a satisfactory solution, i.e., one which satisfies the original polynomial constraints, we must embed the structure of the polynomial in the problem formulation. 

Let us define the index sets 
\begin{equation}\label{eq:indexsets}
	I_j= \{k\in\{1,\dots,M\} : (\alpha_k)_j \neq 0\},\quad j=1,\dots,n .
\end{equation}
Then, the structural information we aim at embedding is given by the following implication:
\begin{equation}\label{eq:implication}
	\forall j\in\{1,\dots,n\},\ \forall k \in I_j,\quad 
	\begin{cases}
		x_j = 0\ \Rightarrow \ \left(\phi(\g x) \right)_k = 0,\\
		\left(\phi(\g x) \right)_k \neq 0 \ \Rightarrow \ x_j \neq 0, 
	\end{cases}
\end{equation}
which formalizes the fact that whenever a variable is zero, all monomials involving this variable must be zero. 
Such a structure can be favored via group-sparsity optimization, as detailed next.

Let us define the set of mappings $\varphi_j : \R^n \rightarrow \R^m$, each computing the subset of components of $\phi$  involving the variable $x_j$ (see Appendix~\ref{app:m} for the precise value of $m = |I_j|$). Note that these mappings are nonlinear in $\g x$ but linear in $\phi(\g x)$: $\varphi_j(\g x) = \g B_j \phi(\g x)$, where $\g B_j$ is an $m$-by-$M$ binary matrix filled with zeros except for a 1 on each row at the $k$th column for all $k\in I_j$.
Further define the $\ell_0$-pseudo-norm of a vector-valued sequence as the number of nonzero vectors in the sequence: 
$$
	\|\{\g u_j\}_{j=1}^n\|_0 = \left| \left\{j \in \{1,\dots,n\} : \g u_j \neq \g 0 \right\} \right| .
$$ 
We call $\g \phi$ a group-sparse vector if $\|\{\g B_j \g \phi\}_{j=1}^n\|_0$ is small. By construction and due to \eqref{eq:implication}, a sparse $\g x$ leads to a group-sparse $\phi(\g x)$ with 
$$
	\|\{\g B_j \phi(\g x)\}_{j=1}^n\|_0 = \|\g x\|_0 .
$$ 
Therefore, in order to solve \eqref{eq:P0}--\eqref{eq:polyf} we search for a group-sparse solution to the linear system $\g A\g \phi = \g y - \g b$. 

Such group-sparse solutions can be found by solving
\begin{align}\label{eq:P0polygroup}
	\min_{\g \phi \in\R^M} \ & \|\{\g W_j \g\phi\}_{j=1}^n\|_0 \\
	\mbox{s.t.}\ &  \g A\g \phi = \g y  - \g b, \nonumber
\end{align}
where we replaced each $\g B_j$ by a weight matrix $\g W_j = \g B_j \g W$, with $\g W = diag\left(\|\g A_1\|_2,\dots, \|\g A_M\|_2 \right)$, 
without effect on the $\ell_0$-pseudo-norm. The reason for introducing these precompensating weights will become clear in the analysis of Sect.~\ref{sec:analysis}. 

Let us define the $\ell_p/\ell_q$-norm of a vector-valued sequence as
$$
	\|\{\g u_j\}_{j=1}^n\|_{p,q} = \left(\sum_{j=1}^n \|\g u_j\|_q^p\right)^{\frac{1}{p}} .
$$
Problem~\eqref{eq:P0polygroup} can be relaxed by replacing the $\ell_0$-pseudo-norm with an  $\ell_p/\ell_q$-norm. In this paper we only consider the case $p=1$ and $q=2$, i.e., 
$$
	\| \{\g W_j\g \phi\}_{j=1}^n \|_{1,2}= \left\|\begin{pmatrix}\|\g W_1\g\phi\|_2\\ \vdots \\\|\g W_n\g\phi\|_2\end{pmatrix}\right\|_1, 
$$
but other norms could be used, such as the $\ell_1/\ell_\infty$-norm.  
This yields the convex relaxation\footnote{Note that since all the groups have the same number of variables, they need not be weighted by a function of the number of variables in each group.} 
\begin{align}\label{eq:P1polygroup}
	\hat{\g \phi} = \arg\min_{\g \phi \in\R^M} \ & \sum_{j=1}^n \|\g W_j\g \phi\|_2 \\
	\mbox{s.t.}\ &  \g A\g \phi = \g y - \g b, \nonumber
\end{align}
which is easily reformulated as a Second Order Cone Program (SOCP) that can be solved by generic software such as CVX \cite{cvx,cvx2} or MOSEK \cite{MOSEK} (more efficient dedicated solvers can also be found, e.g., \cite{Deng11}). However, as for the $\ell_1$-minimization~\eqref{eq:P1poly}, the complexity remains exponential with respect to the number $n$ of base variables via the dimension $M$.

\paragraph{Adding structure via constraints.} 
All components of the mapping $\phi(\g x)$ involving only even degrees of base variables $x_j$ are nonnegative. These structural constraints can help to drive the solution towards one that correctly factorizes and corresponds to a solution of the polynomial equations. This is obtained by solving
\begin{align}\label{eq:P1polygroup+}
	\hat{\g \phi} = \arg\min_{\g \phi \in\R^M} \ & \sum_{j=1}^n \|\g W_j\g \phi\|_2 \\
	\mbox{s.t.}\ &  \g A\g \phi = \g y - \g b\nonumber\\
		& \phi_k \geq 0,\ \forall k\ \mbox{such that } (\alpha_k)_j \mbox{ is even for all } j.\nonumber
\end{align}	
This optimization problem has the same complexity as~\eqref{eq:P1polygroup} since we simply added nonnegativity constraints to some variables.

Other forms of prior knowledge can be easily introduced in~\eqref{eq:P1polygroup+}. For instance, if (tight) box constraints on the base variables are available, then lower and upper bounds on all monomials can be derived. 
Finally, note that these structural constraints can also be added to the $\ell_1$-minimization~\eqref{eq:P1poly}. 

\subsection{Analysis}
\label{sec:analysis}

The following derives conditions of exact recovery of the sparse solution to the system of polynomial equations via various convex relaxations. These conditions are based on the mutual coherence of the matrix $\g A$ as defined e.g. in \cite{Donoho01,Bruckstein09}.
\begin{definition}\label{def:mu} The {\em mutual coherence} of a matrix $\g A = [\g A_1,\dots, \g A_M]$ is 
$$
	\mu(\g A) = \max_{1\leq i< j \leq M} \frac{|\g A_i^T \g A_j|}{\|\g A_i\|_2 \|\g A_j\|_2} .
$$
\end{definition}
In order for the mutual coherence of $\g A$ to be defined, we focus on the case where the following assumption holds. 
\begin{assumption}\label{ass:nozerocol}
All columns $\g A_k$ of the matrix $\g A$ are nonzero, i.e., $\g A_k\neq \g 0$, $k=1,\dots,M$, or, equivalently, for all $k\in\{1,\dots,M\},\ \exists i\in\{1,\dots,N\}$ such that the corresponding polynomial coefficient $a_{ik}\neq 0$. 
\end{assumption}
Assumption~\ref{ass:nozerocol} is slightly more restrictive than Assumption~\ref{ass:nozerocol1}, which only constrains the first $n$ columns of $\g A$. 
If Assumption~\ref{ass:nozerocol} does not hold, the following analysis can be reproduced under Assumption~\ref{ass:nozerocol1} by considering the submatrix $\tilde{\g A}\in\R^{N\times \tilde{M}}$ containing the $\tilde{M}$ nonzero columns of $\g A$ and a similarly truncated mapping $\phi : \R^n\rightarrow \R^{\tilde{M}}$. Adjustments then need to be made where numbers of columns are used, i.e., by substituting $\tilde{M}\leq M$ and $\tilde{m}_j \leq m$ for $M$ and $m$.
However, for the case where both Assumptions~\ref{ass:nozerocol1} and~\ref{ass:nozerocol} do not hold, i.e., of zero columns corresponding to base variables, $\g A_k= \g 0$, $k\in \{1,\dots,n\}$, $\hat{\g x}$ cannot be obtained by the inverse mapping $\phi^{-1}$. This particular case will be discussed in Sect.~\ref{sec:pureNL}.

Note that whenever a condition requires the mutual coherence $\mu(\g A)$ to be defined, Assumption~\ref{ass:nozerocol} implicitly holds. In such cases, Assumption~\ref{ass:nozerocol} will not be explicitly stated in the theorems below. 

\subsubsection{$\ell_1$-minimization method} 

The result below characterizes a case where the simple $\ell_1$-minimization method is sufficient to solve the sparse optimization problem \eqref{eq:P0}--\eqref{eq:polyf}.
\begin{theorem}
	Let $\g x_0$ denote the unique solution to \eqref{eq:P0}--\eqref{eq:polyf}. If the inequality 
	\begin{equation}\label{eq:conditionP1}
		\|\g x_0\|_0 < \frac{n}{2M}\left(1+ \frac{1}{\mu(\g A)}\right) 
	\end{equation}
	holds, then the solution $\hat{\g \phi}$ to \eqref{eq:P1poly} is unique and equal to $\phi(\g x_0)$, thus providing $\hat{\g x} = \phi^{-1}(\hat{\g \phi})= \g x_0$. 
\end{theorem}
\begin{proof}
Assume \eqref{eq:P0}--\eqref{eq:polyf} has a unique solution $\g x_0$. Then, $\g A\g \phi = \g y - \g b$ has a solution $\g \phi_0 = \phi(\g x_0)$ with a sparsity bounded by Proposition~\ref{prop:phisparsity} as 
$$
	\|\g \phi_0\|_0 \leq \frac{M}{n}\|\g x_0\|_0 .
$$
But, according to Theorem 7 in \cite{Bruckstein09}, if 
$$
	\|\g \phi_0\|_0 < \frac{1}{2}\left(1+ \frac{1}{\mu(\g A)}\right) 
$$
then, on the one hand, $\g \phi_0$ is the sparsest solution to $\g A\g\phi =\g y - \g b$, and on the other hand, it is also the unique solution to \eqref{eq:P1poly}. 
Thus, if \eqref{eq:conditionP1} holds, $\g \phi_0$ is the unique solution to both \eqref{eq:P0poly} and \eqref{eq:P1poly}, i.e., $\hat{\g \phi} = \g \phi_0$. Since $\hat{\g x}$ is given by the first components of $\hat{\g \phi}$ and $\g x_0$ by the ones of $\g \phi_0$, this completes the proof.

\end{proof}

Other less conservative conditions can be similarly obtained by considering the exact value of $\|\g \phi_0\|_0$ or tighter bounds (see Appendix~\ref{app:otherconditions}), but these do not take the form of a simple inequality on $\|\g x_0\|_0$. Another condition for very sparse cases (with $\|\g x_0\|_0 \leq n/d - 1$) can be similarly obtained by using Proposition~\ref{prop:phisparsity2} in Appendix~\ref{app:phisparsity} instead of Proposition~\ref{prop:phisparsity}.

\subsubsection{$\ell_1$/$\ell_2$-minimization method} 

The first result below shows that the group-sparse problem~\eqref{eq:P0polygroup} can be used as a proxy for the polynomial problem \eqref{eq:P0}--\eqref{eq:polyf}.
\begin{theorem}\label{thm:grouppoly}
	If the solution $\g \phi^*$ to \eqref{eq:P0polygroup} is unique and yields $\g x^* = \phi^{-1}(\g \phi^*)$ such that $\g x^*$ is a solution to the system of polynomial equations~\eqref{eq:polyf}, then $\g x^*$ is the sparsest solution to the system of polynomial equations, i.e., the unique solution to \eqref{eq:P0}--\eqref{eq:polyf}.
\end{theorem}
\begin{proof}
Assume there is an $\g x_0\neq \g x^*$ solution to the polynomial system~\eqref{eq:polyf} and at least as sparse as $\g x^*$. Then 
$$
	\g A\phi(\g x_0) = \g y - \g b
$$
and 
$$
	\|\{\g W_j \phi(\g x_0)\}_{j=1}^n\|_0 = \|\g x_0\|_0 \leq \|\g x^*\|_0\leq \|\{\g W_j\g \phi^*\}_{j=1}^n\|_0 ,
$$
which contradicts the fact that $\g \phi^*$ is the unique solution to \eqref{eq:P0polygroup} unless $\phi(\g x_0)=\g \phi^*$ and $\g x_0=\g x^*$. 

\end{proof}

The next result provides a condition on the sparsity of the solution to \eqref{eq:P0}--\eqref{eq:polyf} under which it can be recovered by solving the convex problem~\eqref{eq:P1polygroup} or its variant with nonnegativity constraints~\eqref{eq:P1polygroup+}. This result requires the following lemma. 
\begin{lemma}\label{lem:bounddelta2}
	Let $\g A = [\g A_1,\dots, \g A_M]$ be an $N\times M$ matrix with mutual coherence $\mu(\g A)$ as defined in Definition~\ref{def:mu}. Let $\g W$ be the $M\times M$-diagonal matrix of entries $w_i = \|\g A_i\|_2$. Then, for all $\g \delta \in Ker(\g A)$ and $i\in\{1,\dots,M\}$, the bound
\begin{equation}\label{eq:bounddelta2}
	w_i^2 \delta_i^2 \leq \frac{\mu^2(\g A)}{1+\mu^2(\g A)} \|\g W \g \delta\|_2^2 
\end{equation}
holds. 
\end{lemma}
\begin{proof}
For all $\g \delta \in Ker(\g A)$, we have $\g A \g \delta = 0 \Rightarrow \g A^T\g A \g \delta = 0$, which further implies $ \g W^{-1} \g A^T\g A \g W^{-1} \g W \g \delta = 0$ and 
$$
	-\g W \g \delta = \left( \g W^{-1} \g A^T\g A \g W^{-1} - \g I\right)\g W \g \delta .
$$
Then, we have
\begin{align*}
	w_i^2\delta_i^2= \left(\g W \g \delta\right)_i^2 &= \left( ( \g W^{-1} \g A^T\g A \g W^{-1} - \g I) \g W\g \delta \right)_i^2  \\
	&= \left(\sum_{j=1}^M \left( \g W^{-1} \g A^T\g A \g W^{-1} - \g I\right)_{i,j}  (\g W \g \delta)_j \right)^2\\
	&\leq \sum_{j=1}^M  \left( \g W^{-1} \g A^T\g A \g W^{-1} - \g I\right)_{i,j}^2 (\g W \g \delta)_j^2 .
\end{align*}
Note that $\g W^{-1} \g A^T\g A \g W^{-1}$ is a normalized matrix with ones on the diagonal and off diagonal $(i,j)$-entries equal to $\frac{|\g A_i^T \g A_j|}{\|\g A_i\|_2 \|\g A_j\|_2}$ and bounded by $\mu(\g A)$ via Definition~\ref{def:mu}. Thus, we obtain
$$
	w_i^2\delta_i^2 \leq \mu^2(\g A) \sum_{j\neq i} w_j^2\delta_j^2
$$
and, by adding $\mu^2(\g A) w_i^2\delta_i^2$ to both sides, the bound  
$$
	w_i^2\delta_i^2 \leq  \frac{\mu^2(\g A)}{1+\mu^2(\g A)}\sum_{j=1}^M w_j^2\delta_j^2 ,
$$
which is precisely~\eqref{eq:bounddelta2}.

\end{proof}

\begin{theorem}\label{thm:groupsparse}
Let $\g x_0$ be a solution to the polynomial equations~\eqref{eq:polyf} and $\g \phi_0=\phi(\g x_0)$. If the condition
$$
	 \|\g x_0\|_0 <  \frac{1}{2\sqrt{m}}\sqrt{1+ \frac{1}{\mu^2(\g A)}} ,
$$
where $m$ is the number of components of $\phi$ associated to a variable, holds, then $\g \phi_0$ is the unique solution to both~\eqref{eq:P1polygroup} and~\eqref{eq:P1polygroup+}. 
\end{theorem}
\begin{proof}
The vector $\g \phi_0$ is the unique solution to \eqref{eq:P1polygroup} if  the inequality
$$
	\sum_{j=1}^n \|\g W_j (\g \phi_0 + \g \delta)\|_2 > \sum_{j=1}^n \|\g W_j \g \phi_0\|_2
$$
holds for all $\g \delta \neq \g 0$ satisfying $\g A\g \delta = 0$. 
The inequality above can be rewritten as
$$
	\sum_{j\in I_0} \|\g W_j \g \delta\|_2 +\sum_{j\notin I_0} \|\g W_j (\g \phi_0 + \g \delta)\|_2 -  \|\g W_j \g \phi_0\|_2 > 0 ,
$$
where $I_0=\{j\in\{1,\dots,n\} : \g W_j \g \phi_0 = \g 0\}$.  
By the triangle inequality, $\|\g u + \g v\|_2 - \|\g u\|_2 \geq -\|\g v\|_2$, this condition is met if
$$
	\sum_{j\in I_0} \|\g W_j \g \delta\|_2 -\sum_{j\notin I_0} \|\g W_j\g \delta\|_2  > 0
$$
or
\begin{equation}\label{eq:proof1}
	\sum_{j=1}^n \|\g W_j\g \delta\|_2 - 2\sum_{j\notin I_0} \|\g W_j\g \delta\|_2  > 0 .
\end{equation}

By defining $G_j$ as the set of indexes corresponding to nonzero columns of $\g W_j$, Lemma~\ref{lem:bounddelta2} yields
\begin{align*}
	\|\g W_j\g \delta\|_2^2 &= \sum_{i\in G_j} w_i^2 \delta_i^2  \leq m_j\frac{\mu^2(\g A)}{1+\mu^2(\g A)} \|\g W \g \delta\|_2^2 ,
\end{align*}
where $m_j=m$ is the number of components of $\phi$ associated to a variable $x_j$. 
Due to the fact that $\bigcup_{k\in\{1,\dots,n\}} G_k = \{1,\dots,M\}$, we also have
$$
	\|\g W \g \delta\|_2^2 = \sum_{i=1}^M w_i^2\delta_i^2 \leq  \sum_{k=1}^n \sum_{i\in G_k} w_i^2 \delta_i^2 
	= \sum_{k=1}^n \| \g W_k \g \delta\|_2^2 \leq \left(\sum_{k=1}^n \| \g W_k \g \delta\|_2 \right)^2 ,
$$
which then leads to
\begin{align*}
	\|\g W_j\g \delta\|_2^2 
	& \leq m\frac{\mu^2(\g A)}{1+\mu^2(\g A)}  \left( \sum_{k=1}^n \| \g W_k \g \delta\|_2\right)^2 .
\end{align*}
Introducing this result in \eqref{eq:proof1} gives the condition
$$
	\sum_{j=1}^n \|\g W_j\g \delta\|_2 - 2 (n - |I_0|) \frac{\mu(\g A) \sqrt{m}}{\sqrt{1+\mu^2(\g A)}} \sum_{k=1}^n \| \g W_k \g \delta\|_2   > 0 .
$$
Finally, given that $|I_0| = n - \|\{\g W_j\g \phi_0\}_{j=1}^n\|_0 = n - \|\g x_0\|_0$, this yields
\begin{equation}\label{eq:conduniqueP1}
	\sum_{j=1}^n \|\g W_j\g \delta\|_2 - 2 \|\g x_0\|_0 \frac{\mu(\g A) \sqrt{m}}{\sqrt{1+\mu^2(\g A)}} \sum_{k=1}^n \| \g W_k \g \delta\|_2   > 0 .
\end{equation}
or, after dividing by $\| \{\g W_j\g \delta\}_{j=1}^n \|_{1,2}$ and rearranging the terms,  
$$
	 \|\g x_0\|_0 < \frac{\sqrt{1+\mu^2(\g A)}}{2\mu(\g A) \sqrt{m} } ,
$$
which can be rewritten as in the statement of the Theorem. 

It remains to prove that $\g \phi_0$ is also the unique solution to~\eqref{eq:P1polygroup+} in the case where this condition is satisfied and $\g \phi_0$ is the unique solution to~\eqref{eq:P1polygroup}.  To see this, note that $\g \phi_0 = \phi(\g x_0)$ is by definition a feasible point of~\eqref{eq:P1polygroup+}. In addition, the feasible set of~\eqref{eq:P1polygroup+} is included in the one of~\eqref{eq:P1polygroup}, while the problems~\eqref{eq:P1polygroup} and~\eqref{eq:P1polygroup+} share the same cost function. Thus, if $\g\phi_0$ is the unique solution to~\eqref{eq:P1polygroup}, there cannot be another feasible $\g \phi \neq \g\phi_0$ with lower or equal cost function value for~\eqref{eq:P1polygroup+}.

\end{proof}

\begin{corollary}\label{cor:uniqueness}
Let $\g x_0$ be a solution to the polynomial equations~\eqref{eq:polyf} and $m$ be the number of components of $\phi$ associated to a variable. If the condition
$$
	 \|\g x_0\|_0 <  \frac{1}{2\sqrt{m}}\sqrt{1+ \frac{1}{\mu^2(\g A)}} 
$$
holds, then $\g x_0$ is the unique solution to the minimization problem~\eqref{eq:P0}--\eqref{eq:polyf} and it can be computed as $\g x_0 = \phi^{-1}(\hat{\g \phi})$ with $\hat{\g \phi}$ the solution to either~\eqref{eq:P1polygroup} or~\eqref{eq:P1polygroup+}.
\end{corollary}
\begin{proof}
Assume there exists another solution $\g x_1\neq\g x_0$ to \eqref{eq:P0}--\eqref{eq:polyf}, and thus with $\|\g x_1\|_0 \leq \|\g x_0\|_0$. Then, Theorem \ref{thm:groupsparse} implies that both $\phi(\g x_1)$ and $\phi(\g x_0)$ are {\em unique} solutions to either~\eqref{eq:P1polygroup} or~\eqref{eq:P1polygroup+}, and thus that $\phi(\g x_1)=  \phi(\g x_0) = \hat{\g \phi}$ (where $\hat{\g \phi}$ is similarly defined as the solution to either~\eqref{eq:P1polygroup} or~\eqref{eq:P1polygroup+} in this case). 
But this contradicts the definition of the mapping $\phi$ implying $\phi(\g x_1)\neq \phi(\g x_0)$ whenever $\g x_1\neq\g x_0$. 
Therefore, the assumption $\g x_1\neq\g x_0$ cannot hold and $\g x_0$ is the unique solution to \eqref{eq:P0}--\eqref{eq:polyf}, while $\phi^{-1}(\hat{\g \phi}) = \phi^{-1}(\phi(\g x_0)) = \g x_0$.

\end{proof}

\begin{theorem}\label{thm:groupsparse2}
Let $\hat{\g \phi}$ be a solution to \eqref{eq:P1polygroup} and $m$ be the number of components of $\phi$ associated to a variable. If the condition
$$
	  \|\{\g W_j \hat{\g\phi}\}_{j=1}^n\|_0 < \frac{1}{2\sqrt{m}}\sqrt{1+ \frac{1}{\mu^2(\g A)}} 
$$
holds, then $\hat{\g \phi}$ is the unique solution to both~\eqref{eq:P0polygroup} and \eqref{eq:P1polygroup}. If, in addition, $\hat{\g x}=\phi^{-1}(\hat{\g \phi})$ satisfies the polynomial constraints~\eqref{eq:polyf}, then $\hat{\g x}$ is the unique solution to \eqref{eq:P0}--\eqref{eq:polyf}.
\end{theorem}
\begin{proof}
By following steps similar to those in the proof of Theorem~\ref{thm:groupsparse} with $\hat{\g \phi}$ instead of $\g \phi_0$, we obtain (by replacing $\|\g x_0\|_0$ by $\|\{\g W_j \hat{\g \phi}\}_{j=1}^n\|_0$ in \eqref{eq:conduniqueP1}) that $\hat{\g \phi}$ is the unique solution to \eqref{eq:P1polygroup}. 
Then, let $\g \phi^*$ be a solution to \eqref{eq:P0polygroup}. This implies $ \|\{\g W_j \g \phi^*\}_{j=1}^n\|_0 \leq \|\{\g W_j\hat{\g \phi}\}_{j=1}^n\|_0$ and, under the condition of the Theorem,  
\begin{equation}\label{eq:phistarsparsity}
	 \|\{\g W_j \g \phi^*\}_{j=1}^n\|_0 < \frac{1}{2\sqrt{m}}\sqrt{1+ \frac{1}{\mu^2(\g A)}} .
\end{equation}
By following similar steps again, we obtain that $\g \phi^*$ is also the unique solution to \eqref{eq:P1polygroup} and thus that $\g \phi^*=\hat{\g \phi}$ is the unique solution to \eqref{eq:P0polygroup}. 

Finally, since $\hat{\g \phi}$ is the unique solution to \eqref{eq:P0polygroup}, if $\hat{\g x}$ satisfies the polynomial equations~\eqref{eq:polyf}, Theorem~\ref{thm:grouppoly} implies that $\hat{\g x}$ is the solution to \eqref{eq:P0}--\eqref{eq:polyf}.

\end{proof}

In comparison with Theorem~\ref{thm:groupsparse}, Theorem~\ref{thm:groupsparse2} provides a condition that only depends on the estimate obtained by solving \eqref{eq:P1polygroup} rather than on the sought solution.

\subsection{Enhancing sparsity} 
\label{sec:weighting}

In practice, convex relaxations such as the ones described above provide a good step towards the solution but might fail to yield the exact solution with sufficient sparsity. 
In such cases, it is common practice to improve the sparsity of the solution by repeating the procedure with a well-chosen weighting of the variables as described in \cite{Candes08,Le13}. These techniques can be directly applied to improve the $\ell_1$-minimization method of Sect.~\ref{sec:ell1} while they are adapted below to group-sparsity as considered in Sect.~\ref{sec:groupsparse}. 

\subsubsection{Iterative reweighting}
\label{sec:reweighted}

The classical reweighting scheme of \cite{Candes08} for sparse recovery improves the sparsity of the solution by solving a sequence of linear programs. It can be adapted to the group-sparse recovery problem by iteratively solving
\begin{align}\label{eq:P1polygroupweighted}
	\hat{\g \phi} = \arg\min_{\g \phi \in\R^M} \ & \sum_{j=1}^n \mu_j \|\g W_j\g \phi\|_2 \\
	\mbox{s.t.}\ &  \g A\g \phi = \g y - \g b \nonumber\\
		& \phi_k \geq 0,\ \forall k\ \mbox{such that } (\alpha_k)_j \mbox{ is even for all } j, \nonumber
\end{align}	
with weights $\mu_j$ initially set to 1 and refined at each iteration by
$$
	\mu_j = \frac{1}{ \|\g W_j \hat{\g \phi}\|_2 + \epsilon}
$$
for a given small value of $\epsilon>0$.

The basic idea is to decrease the influence of groups of variables with large $\ell_2$-norms that are assumed to be nonzero in the solution while increasing the weight of groups with small norms in order to force them towards zero. 

\subsubsection{Selective $\ell_1$/$\ell_2$-minimization}
\label{sec:Sl1l2M}

The S$\ell_1$M algorithm proposed in \cite{Le13} is another reweighted $\ell_1$-minimization mechanism which sets a single weight at zero at each iteration. Though typically requiring more computation time than the previous approach due to a number of iterations equal to the number of non-zero elements, this algorithm can recover sparse solutions in cases where the classical reweighting scheme of \cite{Candes08} fails. Other advantages include the absence of a tuning parameter and the presence of a convergence analysis \cite{Le13}. The S$\ell_1\ell_2$M algorithm below is an adaptation of S$\ell_1$M to group-sparse problems. 
\begin{enumerate}
	\item Initialize all weights $\mu_j = 1$, $j=1,\dots,n$.
	\item Solve the weighted group-sparse problem \eqref{eq:P1polygroupweighted}.
	\item Find $k \in \arg\max_{j\in\{1,\dots,n\}} \|\g W_j\hat{\g \phi}\|_2$.  
	\item Set $\mu_k = 0$ (to relax the sparsity constraint on the $k$th group).
	\item Repeat from Step 2 until $\sum_{j=1}^n \mu_j \|\g W_j\hat{\g \phi}\|_2 = 0$.
\end{enumerate}
Note that in this algorithm, the number of iterations is equal to the number of nonzero groups\footnote{The maximal number of iterations is the number of groups $n$, but if the correct sparsity pattern is recovered then the algorithm stops earlier.}, which, for polynomial basis pursuit, is $\|\g x_0\|_0$ and is typically small. This results in a fast and accurate method for polynomial basis pursuit, as will be seen in Sect.~\ref{sec:exp}.

\section{Greedy approach}
\label{sec:greedy}

As mentioned in the introduction, there are two major techniques to minimize an $\ell_0$-pseudo-norm: the BP approach and the greedy approach. We now consider the second one to solve problem~\eqref{eq:P0polygroup}, and more particularly develop two variants of the greedy approach: the exact method and the approximate method. The exact method is intended for small-scale problems where the number of possible combinations of base variables remains small. The approximate method is designed to circumvent this limitation and applies to much larger problems. 

\paragraph{Exact greedy algorithm.} The exact method is implemented as follows, where we let $\epsilon = 0$ if the polynomial system of equations is assumed to be feasible, and $\epsilon> 0$ otherwise. 
\begin{enumerate}
	\item Initialize: $\hat{n} = 0$ and $e=+\infty$.
	\item $\hat{n}\leftarrow \hat{n} + 1$.
	\item For all combinations $C$ of $\hat{n}$ variables among $n$:
	\begin{enumerate}
		\item Set $S = \{1,\dots,M\} \setminus \bigcup_{k\notin C} I_k$, where the index sets $I_k$ are defined as in~\eqref{eq:indexsets}.
		\item Build the submatrix $\g A_{S}$ with the columns of $\g A$ with index in $S$.
		\item Solve
		$$
			 \g \phi_S = \arg\min_{\g \phi\in \R^{|S|}} \|\g A_{S} \g \phi  + \g b - \g y\|_2^2 .
		$$
		\item Update $e \leftarrow \min	\{e, \ \|\g A_{S} \g \phi_S  + \g b - \g y\|_2^2\}$. 
		\item If $e\leq \epsilon$, compute $\hat{\g \phi}$ by setting its components of index in $S$ to the values in $\g \phi_S$ and the others to 0. Return $\hat{\g x} = \phi^{-1}(\hat{\g \phi})$.	
	\end{enumerate}
	\item If $\hat{n} < n$, repeat from Step 2, otherwise return an infeasibility certificate.
\end{enumerate}
In Step 3.(a), $S$ corresponds to the support of $\phi(\g x)$ when $supp(\g x)= C$, where $C$ is a combination of $\hat{n}$ indexes from 1 to $n$. 
The maximal number of least squares (LS) problems to solve in Step 3.(c) is $2^n$. But many of these are spared by starting with the sparsest combinations and stopping as soon as a solution is found. Thus, if a sparse solution $\g x_0$ with $\|\g x_0\|_0 \leq n/2$ exists, only $\sum_{t=1}^{\|\g x_0\|_0} \begin{pmatrix}n\\t\end{pmatrix}\leq \|\g x_0\|_0\begin{pmatrix}n\\ \|\g x_0\|_0\end{pmatrix}$ LS problems are solved. At iteration $t$, the LS problem is of size $N\times |S|$, with $|S|= \sum_{q=1}^d \begin{pmatrix}t + q -1 \\ q\end{pmatrix} \leq d \begin{pmatrix}t + d -1 \\ d\end{pmatrix}$. Thus, assuming a complexity $\mathcal{LS}\left(N_1, N_2\right)$ for an LS problem of size $N_1\times N_2$, the overall complexity of the algorithm scales as $ \|\g x_0\|_0 \begin{pmatrix}n\\ \|\g x_0\|_0\end{pmatrix}\mathcal{LS}\left(N, d \begin{pmatrix}\|\g x_0\|_0 + d -1 \\ d\end{pmatrix}\right)$, which is exponential in $n$, $d$ and $\|\g x_0\|_0$.

With $\epsilon=0$, the exact greedy algorithm above can be slightly modified to compute all the solutions to~\eqref{eq:P0polygroup}, simply by letting the {\em for} loop in Step 3 complete instead of returning as soon as a solution is found. As a result, the algorithm could provide a uniqueness certificate for the solution of~\eqref{eq:P0polygroup} from which Theorem~\ref{thm:grouppoly} could be applied to conclude that the solution coincides with the unique minimizer of~\eqref{eq:P0}--\eqref{eq:polyf}.

\paragraph{Approximate greedy algorithm.} The approximate method is similar except that it explores only a single branch of the tree of possible combinations. Its implementation uses a set $S$ of retained variables (more precisely, $S$ contains the indexes of these variables):
\begin{enumerate}
	\item Initialize the set of nonzero variables: $S = \emptyset$.
	\item For all $j\in \{1,\dots,n\} \setminus S$, 
	\begin{enumerate}
		\item Set $S_j = \{1,\dots,M\} \setminus \bigcup_{k\notin S \cup \{j\}} I_k$, where the index sets $I_k$ are defined as in~\eqref{eq:indexsets}.
		\item Build the submatrix $\g A_{S_j}$ with the columns of $\g A$ with index in $S_j$.
		\item Solve
		$$
			 \g \phi_j = \arg\min_{\g \phi\in \R^{|S_j|}} \|\g A_{S_j} \g \phi  + \g b - \g y\|_2^2 .
		$$
	\end{enumerate}
	\item Select the variable that minimizes the error if added to $S$:
	$$
		k = \arg\min_{j\in\{1,\dots,n\}} \|\g A_{S_j} \g \phi_j  + \g b - \g y\|_2^2 .
	$$
	\item Update $S \leftarrow S \cup \{k\}$.  	
	\item Repeat from Step 2 until $\|\g A_{S_k} \g \phi_k  + \g b - \g y\|_2^2\leq \epsilon$. 
	\item Compute $\hat{\g \phi}$ by setting its components of index in $S_k$ to the values in $\g \phi_k$ and the others to 0.
	\item Return $\hat{\g x} = \phi^{-1}(\hat{\g \phi})$ and the error $\|\g A_{S_k} \g \phi_k + \g b - \g y\|_2^2$.
\end{enumerate}
The algorithm starts with an empty set of nonzero variables $S$ and adds a single variable to that set at each iteration. The variable retained at a given iteration is the one that, if added, leads to the minimum sum of squared errors for the equations $\g A_{S_j}\g \phi_j = \g y - \g b$. In Step 2.(a), $S_j$ corresponds to the support of $\phi(\g x)$ when $supp(\g x)= S\cup \{j\}$. 
Note that the value of the minimizer $\g \phi_k$ is not retained but re-estimated at the next iteration. The reason for this is that there is no guarantee that the components of $\g \phi_k$ correspond to monomials of base variables. 

Since one base variable is added at each iteration $t$, the number of iterations equals the sparsity of the returned solution and $t\leq \|\hat{\g x}\|_0$. In this approximate variant of the algorithm, each iteration requires solving only $(n-t+1)$ LS problems, which yields a total number of LS problems equal to $\sum_{t=1}^{\|\hat{\g x}\|_0} (n-t+1) \leq n \|\hat{\g x}\|_0$. Then, bounding the complexity of each LS problem as for the exact greedy algorithm leads to an overall complexity bounded by $n \|\hat{\g x}\|_0 \mathcal{LS}\left(N, d\begin{pmatrix} \|\hat{\g x}\|_0 + d -1 \\ d\end{pmatrix} \right)$.

Since we do not have a result on the convergence of the algorithm to the sparsest solution, we cannot bound $ \|\hat{\g x}\|_0$ except by $n$ and the worst-case complexity is exponential in $n$. However, in practice, if the algorithm finds a sparse solution $\hat{\g x}$, the complexity is exponential in its sparsity $ \|\hat{\g x}\|_0$, but only linear in $n$. This is rather promising given the empirical results to be shown in Sect.~\ref{sec:exp} which suggest that this occurs with high probability.

\section{Purely nonlinear polynomials} 
\label{sec:pureNL}

We now consider the case of purely nonlinear polynomials $p_i^d(\g x)$ without a linear part for some variables, i.e., with $a_{ik} = 0$, $i=1,\dots,N$, for some $k\in\{1,\dots,n\}$ (according to the ordering of the multi-indexes, the linear monomials correspond to the first coefficients with $1\leq k\leq n$).  
For this specific case where Assumption~\ref{ass:nozerocol1} does not hold, some of the first $n$ columns of $\g A$ are zero and the corresponding components of $\g \phi$ are unconstrained, thus set to arbitrary values in $\hat{\g \phi}$. As a result, not only does the analysis of Sect.~\ref{sec:analysis} not hold, but the estimate $\hat{\g x}$ cannot be directly obtained by the inverse mapping $\phi^{-1}$ as the first components of $\hat{\g \phi}$. 

However, the core of the method remains applicable to purely nonlinear polynomials. More precisely, we can solve \eqref{eq:P1poly}, \eqref{eq:P1polygroup} or~\eqref{eq:P1polygroup+} (or apply a reweighting scheme of Sect.~\ref{sec:weighting}) to obtain $\hat{\g \phi}$ and the corresponding support of $\hat{\g x}$ as
$$
	supp(\hat{\g x}) = \{j \in\{1,\dots,n\} : \|\g W_j \hat{\g \phi}\|_2 \neq 0\}, 
$$ 
while the greedy algorithms of Sect.~\ref{sec:greedy} directly estimate $supp(\hat{\g x})$.
Then, the estimate $\hat{\g x}$ can be computed from the higher-degree monomials as, e.g., $\hat{x}_j = \pm \sqrt[\leftroot{-1}\uproot{2}\scriptstyle q]{\hat{\phi}_{jq}}$, where the subscript $_{jq}$ denotes the index such that $\left(\phi(\cdot)\right)_{jq} : \g x\mapsto x_j^q$. 
The precise procedure to compute $\hat{\g x}$ actually depends on the monomials involved in the polynomials. The most straightforward manner is to compute $\hat{x}_j$ from the estimate of its smallest nonzero odd power:
\begin{equation}\label{eq:xhatsqrt}
	\forall j\in supp(\hat{\g x}),\quad \hat{x}_j = 
	\sqrt[\leftroot{-1}\uproot{2}\scriptstyle 2\hat{p}+1]{ \hat{\phi}_{j(2\hat{p}+1)}},\quad \mbox{with } \hat{p} = \min_{p\in\{0,\dots,(d-1)/2\}} p,\ \mbox{s.t.}\ \hat{\phi}_{j(2p+1)} \neq 0 .
\end{equation}
But for polynomial systems that involve only monomials with even degrees, the minimization computing $\hat{p}$ in \eqref{eq:xhatsqrt} has no solution and the procedure is slightly more complex. For instance, with purely quadratic equations, the absolute value of $\hat{x}_j$ is given by  $|\hat{x}_j| = \sqrt{\hat{\phi}_{j2}}$ and the sign must be determined by looking at the signs of the estimates of the bilinear terms $\widehat{x_i x_j}$, $i\neq j$.
In addition, note that for such cases, the solution of \eqref{eq:P0}--\eqref{eq:polyf} is not unique for symmetry reasons and the method cannot be analyzed as in Sect.~\ref{sec:analysis} in terms of convergence towards the sparsest solution. Then, a different notion of uniqueness is usually considered in the literature dedicated to purely quadratic equations \cite{Balan06,Bandeira13,Ohlsson13c,Ranieri13}.

\section{Polynomial denoising} 
\label{sec:noise}

In many applications, the equations $y_i=f_i(\g x)$ need to be relaxed to an error-tolerant form for various reasons, which can for instance be interpreted as having access to noisy measurements, $y_i=f_i(\g x) + e_i$, with unknown noise terms $e_i$. 
In this case, we reformulate the general problem~\eqref{eq:P0} as a denoising one: 
\begin{align}\label{eq:P0noise}
	\min_{\g x\in\R^n, \g e\in\R^N} \ & \|\g x\|_0 \\
	\mbox{s.t.}\ & y_i = f_i(\g x) + e_i,\quad i=1,\dots,N, \nonumber \\
				& \|\g e\|_p \leq \varepsilon , \nonumber
\end{align}
where $\varepsilon> 0$ is a fixed threshold on the noise $\ell_p$-norm, with typical choices for $p$ being $p=1$, $p=2$ or $p=\infty$.

With polynomial constraints, 
\begin{equation}\label{eq:polyfnoise}
	y_i = f_i(\g x) + e_i = p_i^d(\g x)  + e_i = b_{i} + \sum_{k=1}^M a_{ik} \g x^{\alpha_k} + e_i ,\quad i=1,\dots,N,
\end{equation}
convex relaxations similar to the ones described in Sect.~\ref{sec:polyBP} can be derived to obtain polynomial BP denoising methods. This leads to solving
\begin{align}\label{eq:P1polynoise}
	\hat{\g \phi} = \arg\min_{\g \phi \in\R^M} \ & \|\g W\g \phi\|_1 \\
	\mbox{s.t.}\ &  \|\g A\g \phi + \g b - \g y\|_p\leq \varepsilon, \nonumber
\end{align}	
for the $\ell_1$-minimization method and
\begin{align}\label{eq:P1polygroupnoise}
	\hat{\g \phi} = \arg\min_{\g \phi \in\R^M} \ & \sum_{j=1}^n \|\g W_j\g \phi\|_2 \\
	\mbox{s.t.}\ &  \|\g A\g \phi + \g b - \g y\|_p\leq \varepsilon \nonumber
\end{align}	
with $\ell_1/\ell_2$-minimization. 
For $p\in\{1,\infty\}$,  problem~\eqref{eq:P1polynoise} remains a linear program, while $p=2$ leads to a SOCP. 
Problem~\eqref{eq:P1polygroupnoise} can still be written as a SOCP for all $p\in\{1,2,\infty\}$ and be solved by the generic solvers that apply to~\eqref{eq:P1polygroup}. Thus, the enhancements proposed in Sect.~\ref{sec:polyBP} for the formulations \eqref{eq:P1poly} and \eqref{eq:P1polygroup} (such as reweighting schemes or the addition of structural constraints) can be easily transposed to the noisy case.

Regarding the greedy algorithms of Sect.~\ref{sec:greedy}, they are already applicable to the noisy case, for which it suffices to set an appropriate threshold $\varepsilon$ on the noise $\ell_2$-norm. Adaptations of these algorithms to $p=1$ or $p=\infty$ are straightforward, but require solving a convex optimization problem in sub-step (c) without a closed-form solution and thus without the same computational benefit for the approximate greedy approach. 

\subsection{Stability analysis}

In the noisy case, the solution to~\eqref{eq:P0noise} is in general not unique and the analysis focuses on stability rather than on exact recovery. The following theorem provides such a stability result for the $\ell_1$-minimization. 

\begin{theorem}
	Let $(\g x_0, \g e_0)$ denote a solution to \eqref{eq:P0noise}--\eqref{eq:polyfnoise} for $p=2$. If the inequality 
	\begin{equation}\label{eq:conditionP1noise}
		\|\g x_0\|_0 < \frac{n}{4M}\left(1+ \frac{1}{\mu(\g A)}\right) 
	\end{equation}
	holds, then $\hat{\g x} = \phi^{-1}(\hat{\g \phi})$ with $\hat{\g \phi}$ the solution to 
	\eqref{eq:P1polynoise} for $p=2$ must obey
	\begin{equation}\label{eq:stabilityx}
		\|\hat{\g x} - \g x_0\|_2^2 \leq \frac{4\varepsilon^2}{1 - \mu(\g A)(4M\|\g x_0\|_0/n - 1)} .
	\end{equation}
\end{theorem}
\begin{proof}
Assume \eqref{eq:P0noise}--\eqref{eq:polyfnoise} has a solution $(\g x_0,\g e_0)$. Then, $\g A\g \phi_0 = \g y - \g b - \g e_0$ with $\g \phi_0 = \phi(\g x_0)$ and $\|\g e_0\|_2\leq \varepsilon$. By Proposition~\ref{prop:phisparsity}, we have $\|\g \phi_0\|_0 \leq \frac{M}{n}\|\g x_0\|_0$.
But, according to Theorem 9 in \cite{Bruckstein09}, if 
$
	\|\g \phi_0\|_0 < \frac{1}{4}\left(1+ \frac{1}{\mu(\g A)}\right) 
$
then $\hat{\g \phi}$ must obey 
\begin{equation}\label{eq:stabilityphi}
		\|\hat{\g \phi} - \g \phi_0\|_2^2 \leq \frac{4\varepsilon^2}{1 - \mu(\g A)(4\|\g \phi_0\|_0  - 1)}. 
\end{equation}
Thus, if \eqref{eq:conditionP1noise} holds, so does~\eqref{eq:stabilityphi} and 
$\|\hat{\g \phi} - \g \phi_0\|_2^2 \leq 4\varepsilon^2/ [1 - \mu(\g A)(4M \|\g x_0\|_0/ n  - 1)]$.
Since $\|\hat{\g x} - \g x_0\|_2^2 \leq \|\hat{\g \phi} - \g \phi_0\|_2^2$, this completes the proof.

\end{proof}

For the $\ell_1/\ell_2$-minimization, we have the following stability result.  
\begin{theorem}\label{thm:stabilitygroup}
	Let $(\g x_0, \g e_0)$ denote a solution to \eqref{eq:P0noise}--\eqref{eq:polyfnoise} for $p=2$. If the inequality 
	\begin{equation}\label{eq:conditionP1groupnoise}
		\|\g x_0\|_0 < \frac{1}{4n M}\left(1+ \frac{1}{\mu(\g A)}\right) 
	\end{equation}
	holds, then $\hat{\g x} = \phi^{-1}(\hat{\g \phi})$ with $\hat{\g \phi}$ the solution to 
	\eqref{eq:P1polygroupnoise} for $p=2$ must obey
	\begin{equation}\label{eq:stabilityxgroupW}
		\|\g W_x (\hat{\g x} - \g x_0)\|_2^2 \leq \frac{4n \varepsilon^2}{1 - \mu(\g A)(4n M\|\g x_0\|_0 - 1)} ,
	\end{equation}
	where $\g W_x= diag(\|\g A_1\|_2, \dots, \|\g A_n\|_2) $. 	
\end{theorem}
\begin{proof}
Assume \eqref{eq:P0noise}--\eqref{eq:polyfnoise} has a solution $(\g x_0,\g e_0)$. Let us define  $\g \phi_0 = \phi(\g x_0)$ and $\g \delta = \hat{\g \phi} - \g \phi_0$. To prove the theorem, we will first bound from above the norm $\|\g W\g \delta\|_2$, where $\g W$ is the diagonal matrix of precompensating weights. This part of the proof follows a path similar to that of Theorem~3.1 in \cite{Donoho06noise}, while adapting it to the group-sparse setting and mixed $\ell_p$/$\ell_q$ norms.

Due to the definition of $\hat{\g \phi}$ as a minimizer of \eqref{eq:P1polygroupnoise}, $\g \delta$ must satisfy either
$$
	\sum_{j=1}^n \|\g W_j (\g \phi_0 + \g \delta)\|_2 < \sum_{j=1}^n \|\g W_j \g \phi_0\|_2
$$
or $\g \delta = \g 0$, in which case the statement is obvious. 
The inequality above can be rewritten as
$$
	\sum_{j\in I_0} \|\g W_j \g \delta\|_2 +\sum_{j\notin I_0} \|\g W_j (\g \phi_0 + \g \delta)\|_2 -  \|\g W_j \g \phi_0\|_2 < 0 ,
$$
where $I_0=\{j\in\{1,\dots,n\} : \g W_j \g \phi_0 = \g 0\}$.  
By the triangle inequality, $\|\g u + \g v\|_2 - \|\g u\|_2 \geq -\|\g v\|_2$, this implies 
\begin{equation}\label{eq:proofainfb}
	\sum_{j\in I_0} \|\g W_j \g \delta\|_2 - \sum_{j\notin I_0} \|\g W_j\g \delta\|_2 < 0.
\end{equation}
In addition, $\g \delta$ must satisfy the constraints in \eqref{eq:P1polygroupnoise} as 
$$
	\|\g A(\g \phi_0 + \g \delta)  + \g b - \g y\|_2\leq \varepsilon ,
$$	
in which $\g y$ can be replaced by $\g A \g \phi_0 + \g b + \g e_0$, leading to 
$$
	\|\g A\g \delta - \g e_0\|_2\leq \varepsilon .
$$
Using $\|\g u\|_2 \leq \|\g u - \g v\|_2 + \|\g v\|_2$, this implies $\|\g A\g \delta\|_2\leq 2 \varepsilon$, which further gives
\begin{align}
	4\varepsilon^2 &\geq \|\g A\g \delta\|_2^2 = \|\g A\g W^{-1}\g W \g \delta\|_2 \nonumber
	= (\g W \g \delta)^T \g W^{-1} \g A^T \g A\g W^{-1} (\g W \g \delta) \nonumber\\
	&= \|\g W\g \delta\|_2^2  + (\g W \g \delta)^T ( \g W^{-1} \g A^T \g A\g W^{-1} - \g I) (\g W \g \delta) \nonumber\\
	&\geq \|\g W\g \delta\|_2^2  - \left| (\g W \g \delta)^T ( \g W^{-1} \g A^T \g A\g W^{-1} - \g I) (\g W \g \delta)\right|\nonumber\\
	&\geq \|\g W\g \delta\|_2^2  - | \g W \g \delta|^T | \g W^{-1} \g A^T \g A\g W^{-1} - \g I| |\g W \g \delta|\nonumber\\
	&\geq \|\g W\g \delta\|_2^2  - \mu(\g A) (\|\g W\g \delta\|_1^2 - \|\g W\g \delta\|_2^2 ) \nonumber\\
	&= ( 1 + \mu(\g A) )\|\g W\g \delta\|_2^2  - \mu(\g A) \|\g W\g \delta\|_1^2  
	\label{eq:proofstabilitygroup1}
\end{align}
where we used $\g W^{-1}\g W = \g I$ and the fact that the diagonal entries of $|\g W^{-1} \g A^T \g A\g W^{-1} - \g I|$ are zeros while off-diagonal entries are bounded from above by $\mu(\g A)$. 

Due to $\g W_j\g\delta$ being a vector with a subset of entries from $\g W\g \delta$, we have $\|\g W\g \delta\|_2^2 \geq  \|\g W_j\g \delta\|_2^2$, $j=1,\dots,n$, and thus
\begin{equation}\label{eq:bounddeltaL2}
	 \|\g W\g \delta\|_2^2 \geq  \frac{1}{n} \sum_{j=1}^n  \|\g W_j\g \delta\|_2^2 =  \frac{1}{n} \|\{\g W_j\g \delta\}_{j=1}^n\|_{2,2}^2 . 
\end{equation}
Since the groups defined by the $\g W_j$'s overlap, $\|\g W\g \delta\|_2 \leq \sum_{j=1}^n\|\g W_j\g \delta\|_2$, and the squared $\ell_1$-norm in~\eqref{eq:proofstabilitygroup1} can be bounded by 
\begin{equation}\label{eq:bounddeltaL1}
	\|\g W\g \delta\|_1^2 \leq M \|\g W\g \delta\|_2^2 \leq M\left( \sum_{j=1}^n\|\g W_j\g \delta\|_2\right)^2 
	= M \|\{\g W_j\g \delta\}_{j=1}^n\|_{1,2}^2 .
\end{equation}
Introducing the bounds \eqref{eq:bounddeltaL2}--\eqref{eq:bounddeltaL1} in \eqref{eq:proofstabilitygroup1} yields
\begin{equation}\label{eq:proofstabilitygroup2}
 \frac{1 + \mu(\g A) }{n} \|\{\g W_j\g \delta\}_{j=1}^n\|_{2,2}^2  - \mu(\g A) M \|\{\g W_j\g \delta\}_{j=1}^n\|_{1,2}^2  \leq 4\varepsilon^2 .
\end{equation}
We will now use this inequality to derive an upper bound on $\|\{\g W_j\g \delta\}_{j=1}^n\|_{2,2}^2$, which will also apply to $\|\g W\g \delta\|_2^2 \leq \|\{\g W_j\g \delta\}_{j=1}^n\|_{2,2}^2$, since the groups overlap and the squared components of $\g W\g \delta$ are summed multiple times in $\|\{\g W_j\g \delta\}_{j=1}^n\|_{2,2}^2$. To derive the upper bound, we first introduce a few notations: 
$$
	a = \|\{\g W_j\g \delta\}_{j\in I_0}\|_{1,2},\quad b = \|\{\g W_j\g \delta\}_{j\notin I_0}\|_{1,2},
$$
and 
$$
	c_0 = \left(\frac{\|\{\g W_j\g \delta\}_{j\in I_0}\|_{2,2}}{\|\{\g W_j\g \delta\}_{j\in I_0}\|_{1,2}}\right)^2 \in \left[\frac{1}{|I_0|}, 1\right],
	\quad c_1 =  \left(\frac{\|\{\g W_j\g \delta\}_{j\notin I_0}\|_{2,2}}{\|\{\g W_j\g \delta\}_{j\notin I_0}\|_{1,2}}\right)^2 \in \left[\frac{1}{n - |I_0|}, 1\right] ,
$$
where the box bounds are obtained by classical relations between the $\ell_1$ and $\ell_2$ norms ($\forall\g u\in \R^k$, $\|\g u\|_2 \leq \|\g u\|_1 \leq \sqrt{k}\|\g u\|_2$). 
With these notations, the term to bound is rewritten as
$$
	\|\{\g W_j\g \delta\}_{j=1}^n\|_{2,2}^2 = \|\{\g W_j\g \delta\}_{j\in I_0}\|_{2,2}^2 +  \|\{\g W_j\g \delta\}_{j\notin I_0}\|_{2,2}^2 = c_0 a^2 + c_1 b^2 ,
$$
while the inequality~\eqref{eq:proofstabilitygroup2} becomes
$$
	\frac{1 + \mu(\g A) }{n} (c_0 a^2 + c_1 b^2)  - \mu(\g A) M (a + b)^2  \leq 4\varepsilon^2 .
$$
We further reformulate this constraint by letting $a=\rho b$:
\begin{equation}\label{eq:proofstabilitygroup3}
	\frac{1 + \mu(\g A) }{n} (c_0 \rho^2 + c_1) b^2  - \mu(\g A) M (1 + \rho)^2 b^2  \leq 4\varepsilon^2 .
\end{equation}
Let $\gamma =  (1 + \rho)^2 / (c_0 \rho^2 + c_1)$. Due to~\eqref{eq:proofainfb}, we have $a<b$ and thus $\rho\in[0,1)$, which, together with the bounds on $c_0$ and $c_1$, gives the constraints $1\leq \gamma\leq 4(n-|I_0|)$. By setting $V = (c_0 \rho^2 + c_1) b^2$, \eqref{eq:proofstabilitygroup2} is rewritten as 
$$
	\frac{1 + \mu(\g A) }{n} V  - \mu(\g A) M \gamma V \leq 4\varepsilon^2 , 
$$
where 
$$
	\frac{1 + \mu(\g A) }{n}  - \mu(\g A) M \gamma \geq \frac{1 + \mu(\g A) }{n}  - 4 (n-|I_0|) \mu(\g A) M > 0, 
$$
since $\gamma\leq 4(n-|I_0|)$ and the positivity is ensured by the condition \eqref{eq:conditionP1groupnoise} and the fact that $\|\g x_0\|_0 = n - |I_0|$. Thus, 
$$
	\|\g W\g \delta\|_2^2 \leq \|\{\g W_j\g \delta\}_{j=1}^n\|_{2,2}^2 = V \leq \frac{4n\varepsilon^2}{1+\mu(\g A) - 4 \mu(\g A) n M \|\g x_0\|_0 }
$$
and, since $\|\g W_x(\hat{\g x} - \g x_0)\|_2^2 \leq \|\g W (\hat{\g \phi} - \g \phi_0) \|_2^2=\|\g W\g \delta\|_2^2$, \eqref{eq:stabilityxgroupW} is proved. 

\end{proof}

\section{Experiments} 
\label{sec:exp}

This section evaluates the efficiency of the BP and greedy methods in terms of accuracy and computing time for the noiseless case in Sect.~\ref{sec:expnoiseless} and the noisy case in Sect.~\ref{sec:expnoise}. Here, the definition of accuracy depends on the presence of noise in the equations, while the computing time refers to Matlab implementations\footnote{The code for the proposed methods is available at \url{http://www.loria.fr/~lauer/software/} .} using MOSEK and CVX for the convex programs and running on a standard laptop (except for times reported in Fig.~\ref{fig:time}).

The following methods are compared: the iterative hard thresholding (IHT) algorithm \cite{Blumensath13} as implemented by \cite{Beck13}, the quadratic (QBP) and nonlinear (NLBP) BP methods of \cite{Ohlsson13quadratic,Ohlsson13nlbp}, the simple $\ell_1$-minimization ($\ell_1$M) solving \eqref{eq:P1poly}\footnote{The results reported here are obtained with 10 iterations of the reweighting procedure of \cite{Candes08} applied to~\eqref{eq:P1poly}.}, the $\ell_1/\ell_2$-minimization ($\ell_1\ell_2$M) solving \eqref{eq:P1polygroup+} with its iteratively reweighted counterpart (IR$\ell_1\ell_2$M) using 10 iterations as defined in Sect.~\ref{sec:reweighted} for $\epsilon=0.001$, the selective $\ell_1/\ell_2$-minimization (S$\ell_1\ell_2$M) of Sect.~\ref{sec:Sl1l2M}, and the exact (EGA) and approximate (AGA) greedy algorithms of Sect.~\ref{sec:greedy}. 
For the noisy cases, error tolerant formulations of these methods as described in Sect.~\ref{sec:noise} are used. 

\begin{table}
\centering
\caption{Results on the example from \cite{Ohlsson13quadratic} with quadratic equations. \label{tab:QBP}}
\begin{tabular}{lllllllll}\hline\noalign{\smallskip}
	\bf Method & IHT & QBP & $\ell_1$M & $\ell_1\ell_2$M & IR$\ell_1\ell_2$M  & S$\ell_1\ell_2$M & AGA & EGA \\
\noalign{\smallskip}\hline\noalign{\smallskip}
	\bf Success rate & $54 \%$ & $79 \%$  & $0 \%$ & $0 \%$ & $97 \%$ & $97 \%$ & $91 \%$ & $100 \%$\\
	\bf Mean time (s) &  1.92  & 1.76 & 0.26  & 0.07  & 0.46  & 0.15  & 0.03 & 0.14
	\\
	\hline\noalign{\smallskip}
\end{tabular}
\end{table}
\begin{table}
\centering
\caption{Results on the example from \cite{Ohlsson13nlbp} with polynomials of degree $d=4$. \label{tab:PBP}}
\begin{tabular}{lllllllll}\hline\noalign{\smallskip}
	\bf Method & IHT & NLBP  & $\ell_1$M & $\ell_1\ell_2$M & IR$\ell_1\ell_2$M  & S$\ell_1\ell_2$M & AGA & EGA
	\\
\noalign{\smallskip}\hline\noalign{\smallskip}
	\bf Success rate & $66 \%$ & $100 \%$  & $85 \%$ & $16 \%$ & $100 \%$ & $100 \%$ & $100 \%$ & $100 \%$\\
	\bf Mean time (s) & 0.43 & 13.1 & 0.22 & 0.06 &  0.42 & 0.08  & 0.006  & 0.006
\\	\hline\noalign{\smallskip}
\end{tabular}
\end{table}
\subsection{Exact recovery with noiseless equations} 
\label{sec:expnoiseless}

In the noiseless setting considered in the following experiments, accuracy is defined as the ability of a method to recover the sparsest solution $\g x_0$ of a polynomial system and is measured as a success rate, i.e., the percentage of systems for which it recovers $\g x_0$ (meaning $\|\hat{\g x} -\g x_0\|_2 \leq 10^{-6}$) in a Monte Carlo experiment involving $100$ trials with random polynomial coefficients but same $\g x_0$. More precisely, for each experiment, given a sparsity level $\|\g x_0\|_0$, we generate an $\g x_0$ with the first $\|\g x_0\|_0$ components set to 1 and the rest at 0, while the values of $y_i$, $i=1,\dots,N$, are given by~\eqref{eq:polyf} with $b_i$ and $a_{ik}$, $k=1,\dots,M$, randomly drawn for each trial according to a zero-mean Gaussian distribution with unit variance. 

\paragraph{Sparse solutions of quadratic equations.} 
Consider example A in \cite{Ohlsson13quadratic} with $N=25$ quadratic equations ($d=2$) in $n=20$ variables. The true $\g x_0$ that we aim at recovering has three components at 1 and the rest at 0: $\|\g x_0\|_0=3$. The results reported in Table~\ref{tab:QBP} show that all reweighted BP and greedy methods recover the solution in almost all trials in less than 1 second, whereas the IHT and QBP methods lead to more failures despite longer computing times. However, the $\ell_1$M method, which does not enforce the polynomial structure on $\hat{\g \phi}$, cannot recover the correct solution. The straightforward optimization of the $\ell_1/\ell_2$-norm ($\ell_1\ell_2$M) also fails, which shows the importance of reweighting for obtaining truly sparse solutions. 

\paragraph{Higher-degree polynomial equations.} 
Consider now example A in \cite{Ohlsson13nlbp} with  $N=50$ polynomial equations of degree $d=4$ in $n=5$ variables with two nonzeros ($\|\g x_0\|_0=2$). 
Though using higher degree polynomials, according to the results shown in Table~\ref{tab:PBP}, this problem setting seems easier than the quadratic example above. Indeed, the simple $\ell_1$M method already obtains a success rate of $85 \%$ while all reweighted BP and greedy methods achieve a $100 \%$ success rate. This is due to the increase of the number of equations, $N$, and the decrease of the number of variables, $n$, as will be emphasized by additional experiments below. 
Here, S$\ell_1\ell_2$M can be much faster than IR$\ell_1\ell_2$M based on the iterative reweighting of \cite{Candes08} (Sect.~\ref{sec:reweighted}) because the number of nonzero groups is very small: $2 < n<$ typical number of iterations for the reweighting of \cite{Candes08}. However, the AGA obtains the same success rate, while being one order of magnitude faster. 

The NLBP method also recovered the correct solution in all trials, but with a much longer computing time. This time is roughly divided in halves for the construction of quadratic constraints that can be handled by QBP on the one hand and the semi-definite optimization of QBP on the other hand. 
Here again, the IHT method offers a low success rate, though it uses additional information on $\|\g x_0\|_0$ to compute $\hat{\g x}$ with the optimal sparsity level (which is unknown to the other methods). This is due in part to the difficulty of tuning the gradient step-size to obtain convergence with an objective function whose gradient is not Lipschitz continuous. For this reason, we exclude the IHT method from the remaining experiments.

\begin{table}
\centering
\caption{Results with purely quadratic equations. \label{tab:pureQBP}}
\begin{tabular}{llllllll}\hline\noalign{\smallskip}
	\bf Method & QBP  & $\ell_1$M & $\ell_1\ell_2$M & IR$\ell_1\ell_2$M  & S$\ell_1\ell_2$M & AGA & EGA
	\\
\noalign{\smallskip}\hline\noalign{\smallskip}
	\bf Success rate & $0 \%$ & $0 \%$ &  $3 \%$ & $100 \%$ & $99 \%$ & $91 \%$ &$100 \%$  \\
	\bf Mean time (s) & 1.42 & 0.27 & 0.03 &  0.77 & 0.19 & 0.05 & 0.36
\\	\hline\noalign{\smallskip}
\end{tabular}
\end{table}
\begin{table}
\centering
\caption{Results with purely nonlinear equations of degree $d=4$. \label{tab:pureNLBP}}
\begin{tabular}{llllllll}\hline\noalign{\smallskip}
	\bf Method & NLBP  & $\ell_1$M & $\ell_1\ell_2$M & IR$\ell_1\ell_2$M  & S$\ell_1\ell_2$M & AGA & EGA
	\\
\noalign{\smallskip}\hline\noalign{\smallskip}
	\bf Success rate & $100 \%$ & $0 \%$ &  $15 \%$ & $100 \%$ & $100 \%$ & $100 \%$ &$100 \%$  \\
	\bf Mean time (s) & 13.1 & 0.21 & 0.07 & 0.70 &  0.12 & 0.01 & 0.012
\\	\hline\noalign{\smallskip}
\end{tabular}
%
\end{table}
\begin{table}[t!]
\centering
\caption{Results with purely quadratic equations generated as in phase retrieval. \label{tab:CPRL}}
\begin{tabular}{lllllll}\hline\noalign{\smallskip}
	\bf Method & GSS & QBP & IR$\ell_1\ell_2$M  & S$\ell_1\ell_2$M & AGA & EGA
	\\
\noalign{\smallskip}\hline\noalign{\smallskip}
	\bf Success rate & $32 \%$  & $0 \%$ & $79 \%$ & $72 \%$ & $71 \%$ &$100 \%$  \\
	\bf Mean time (s) & 0.10 & 1.45 &  0.71 & 0.25 & 0.05 & 0.38
\\	\hline\noalign{\smallskip}
\end{tabular}
%
\end{table}

\paragraph{Purely nonlinear equations.} 
Consider now the particular case of purely nonlinear equations discussed in Sect.~\ref{sec:pureNL}. Given $N=25$ purely quadratic polynomials $p_i^2(\g x) = \g x^T\g Q_i\g x$ in $n=20$ variables, we compute the values $y_i=p_i^2(\g x_0)$. 
Results are shown in Table~\ref{tab:pureQBP} for $\|\g x_0\|_0=3$ and where the definition of the success rate is slightly modified: since the sparsest solution is not unique in this case, a successful trial is defined as one for which the estimate $\hat{\g x}$ belongs to the set of solutions (meaning that $\min\{\|\hat{\g x} -\g x_0\|_2, \|\hat{\g x} + \g x_0\|_2\} \leq 10^{-6}$). Table~\ref{tab:pureNLBP} reports results of similar experiments with a higher degree $d=4$ and $N=50$, $n=5$, $\|\g x_0\|_0 = 2$. In this case, the estimates of base variables $\hat{x}_j$ are directly given as cube roots of the estimates of $x_j^3$. In both cases ($d=2$ or $4$), the success rates are similar to the ones reported in Tables~\ref{tab:QBP} and~\ref{tab:PBP}, which provides evidence that the alternative proposed in Sect.~\ref{sec:pureNL} yields satisfactory results when Assumption~\ref{ass:nozerocol1} does not hold. 

In order to compare with the greedy sparse-simplex (GSS) method of \cite{Beck13} using the implementation provided with that paper, we consider a typical application for purely quadratic equations, namely, phase retrieval. In such applications, the measurements obey $c_i =  |\g c_i^T\g x|$, which can be reformulated as $y_i = c_i^2 =  \g x^T\g c_i\g c_i^T\g x$ to be handled by the proposed methods. Here, the components of the vectors $\g c_i$, $i=1,\dots,N$, are randomly drawn in each trial according to a zero-mean Gaussian distribution of unit variance. Results in Table~\ref{tab:CPRL} obtained for $N=25$, $n=20$ and $\|\g x_0\|_0=3$ show that the proposed methods perform much better than the GSS method with data generated in this manner, despite the fact that GSS works with the additional knowledge of $\|\g x_0\|_0$. 
Note that many other methods have been proposed for sparse phase retrieval. The comparison is here limited to a sample of generic sparse recovery methods that apply to phase retrieval as a special case. 

\paragraph{Remark:} {\em All examples above consider a moderate dimensional setting ($n\leq 20$) with low sparsity levels $\|\g x_0\|_0\leq 3$. In such cases, the exact greedy algorithm (EGA) can be applied without complexity issues and yields a perfect recovery in all trials, even faster than the BP methods based on convex relaxations. However, despite its name and the observed success rates, the EGA is not an exact algorithm for solving polynomial systems but only for solving group-sparsity optimization problems in the form of \eqref{eq:P0polygroup}. While Theorem~\ref{thm:grouppoly} states that this is sufficient to solve \eqref{eq:P0}--\eqref{eq:polyf} if the polynomial constraints are satisfied, there is no guarantee that this is the case in general. 
}

\paragraph{Estimating the probability of successful recovery.} Figures~\ref{fig:diagramn10} (for $n=10$) and~\ref{fig:diagramn20} (for $n=20$) show the probability of successful recovery of an $\g x_0$ estimated via the success rate versus the sparsity level $\|\g x_0\|_0$ for various values of $\delta = N/n$. Contrary to the linear case where the system is overdetermined for values of $\delta > 1$, here $\delta$ can grow as large as $M/n$ before the system becomes overdetermined. Thus, classical studies on the linear case focus on small values of $\delta < 1$, whereas we only analyze the probability of successful recovery for $\delta \geq 1$ ($\delta < 1$ is a very difficult setting with polynomial equations). 

These results show that all the proposed methods can recover the sparsest solution for sufficiently sparse cases with high probability and that the sparsity level at which this occurs depends on the particular method on the one hand and on the ratio $\delta$ on the other hand. Larger values of $\delta$ correspond to larger systems with more equations and thus with more information on the sought solution. Another expected observation is that the probability of success decreases when the degree $d$ increases. For $n=20$, this leads to constant failure for the $\ell_1$M and IR$\ell_1\ell_2$M methods at $d=5$, though the S$\ell_1\ell_2$M and AGA methods can still recover sufficiently sparse solutions. 
Except for this particular setting ($n=20$, $d=5$), the results for IR$\ell_1\ell_2$M  and S$\ell_1\ell_2$M are comparable, while the AGA yields slightly lower probabilities. The simple $\ell_1$M method is much less effective than the others, but nonetheless obtains a high probability of recovery for not too difficult cases with $d=2$ and a sufficiently large $N$.

\begin{figure}
	\centering
	$\g{d=2}$\\
	\includegraphics[width = .24\linewidth]{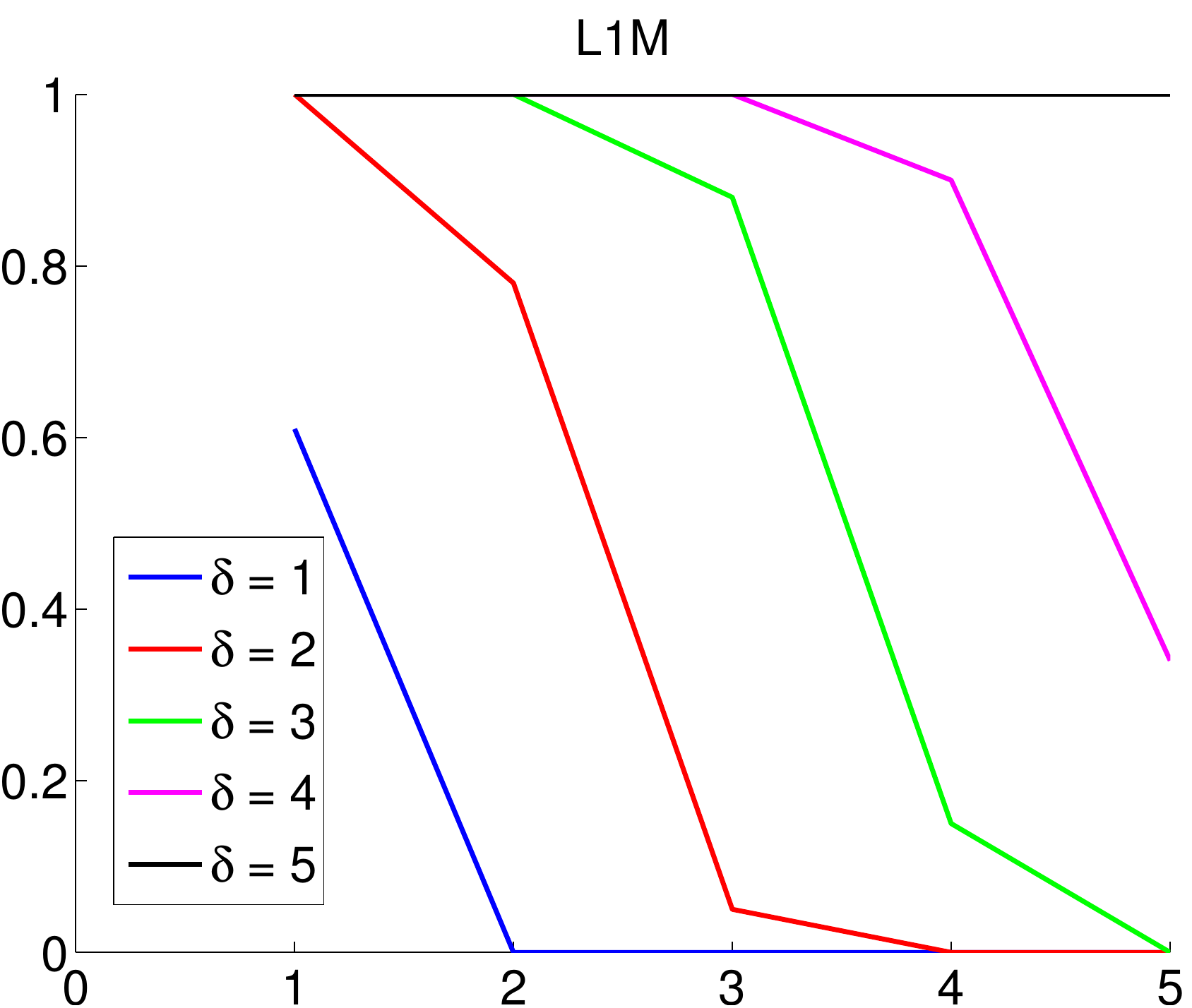}	
	\includegraphics[width = .24\linewidth]{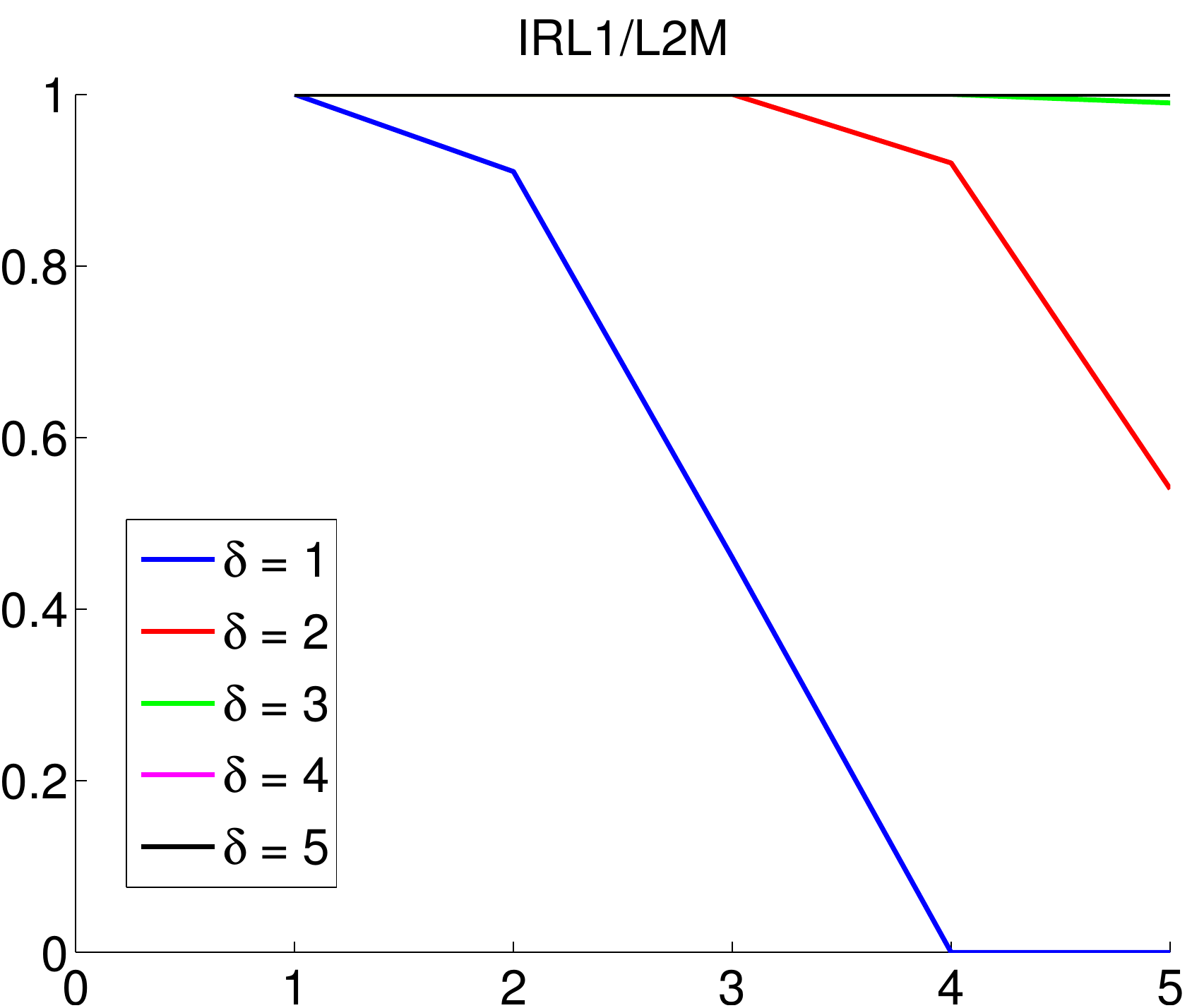}
	\includegraphics[width = .24\linewidth]{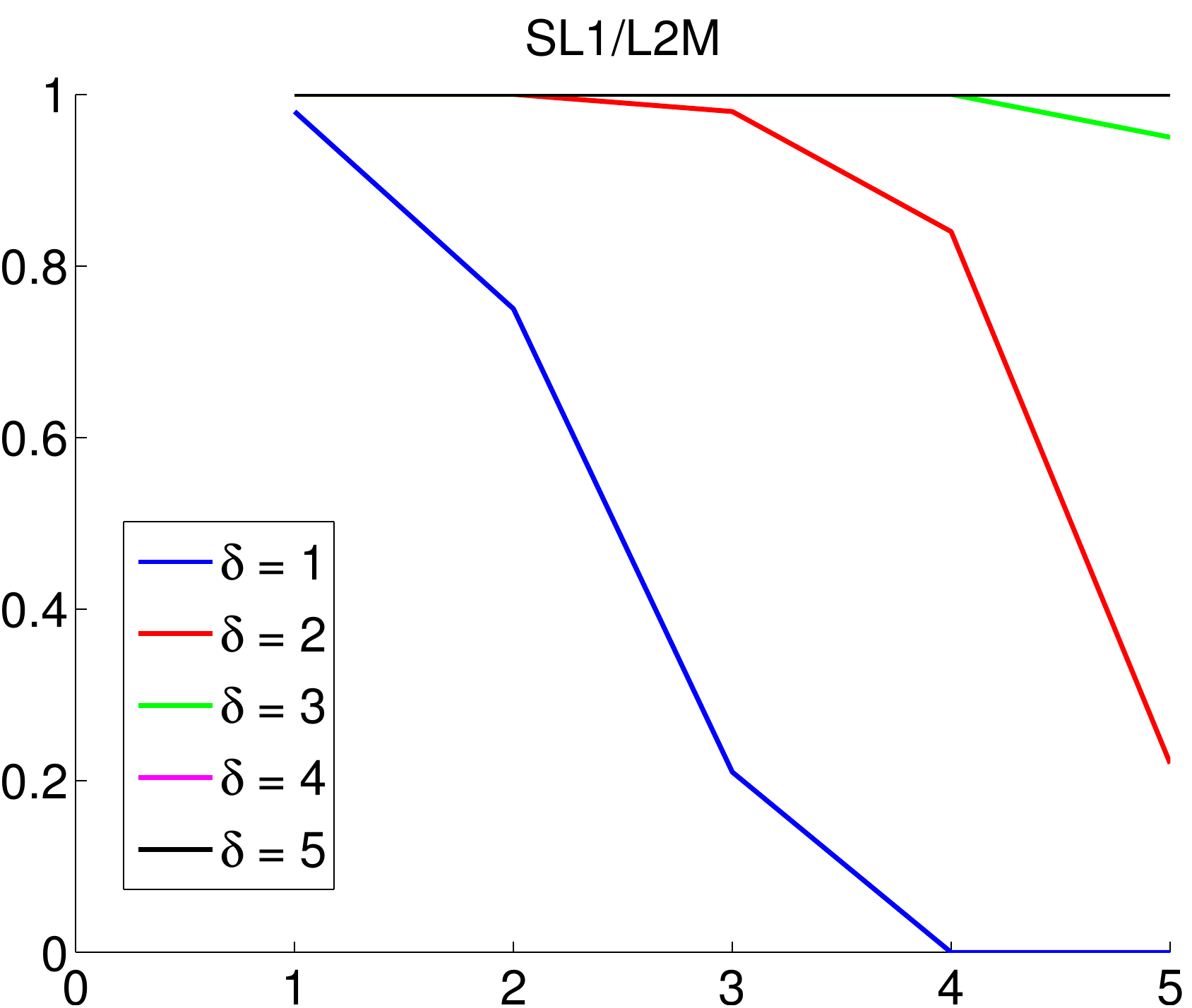}			
	\includegraphics[width = .24\linewidth]{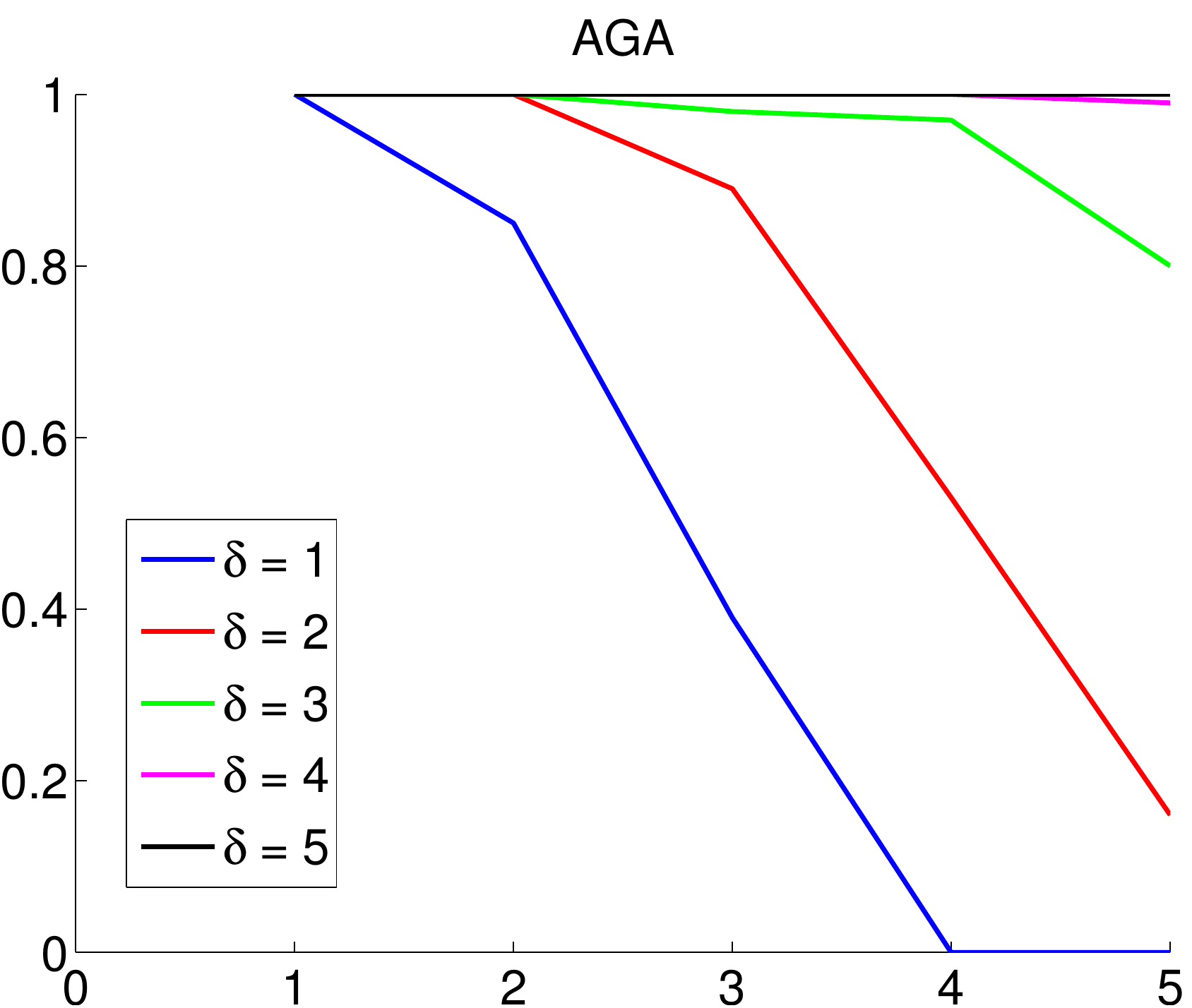}	
	\medskip\\
	$\g{d=3}$\\
	\includegraphics[width = .24\linewidth]{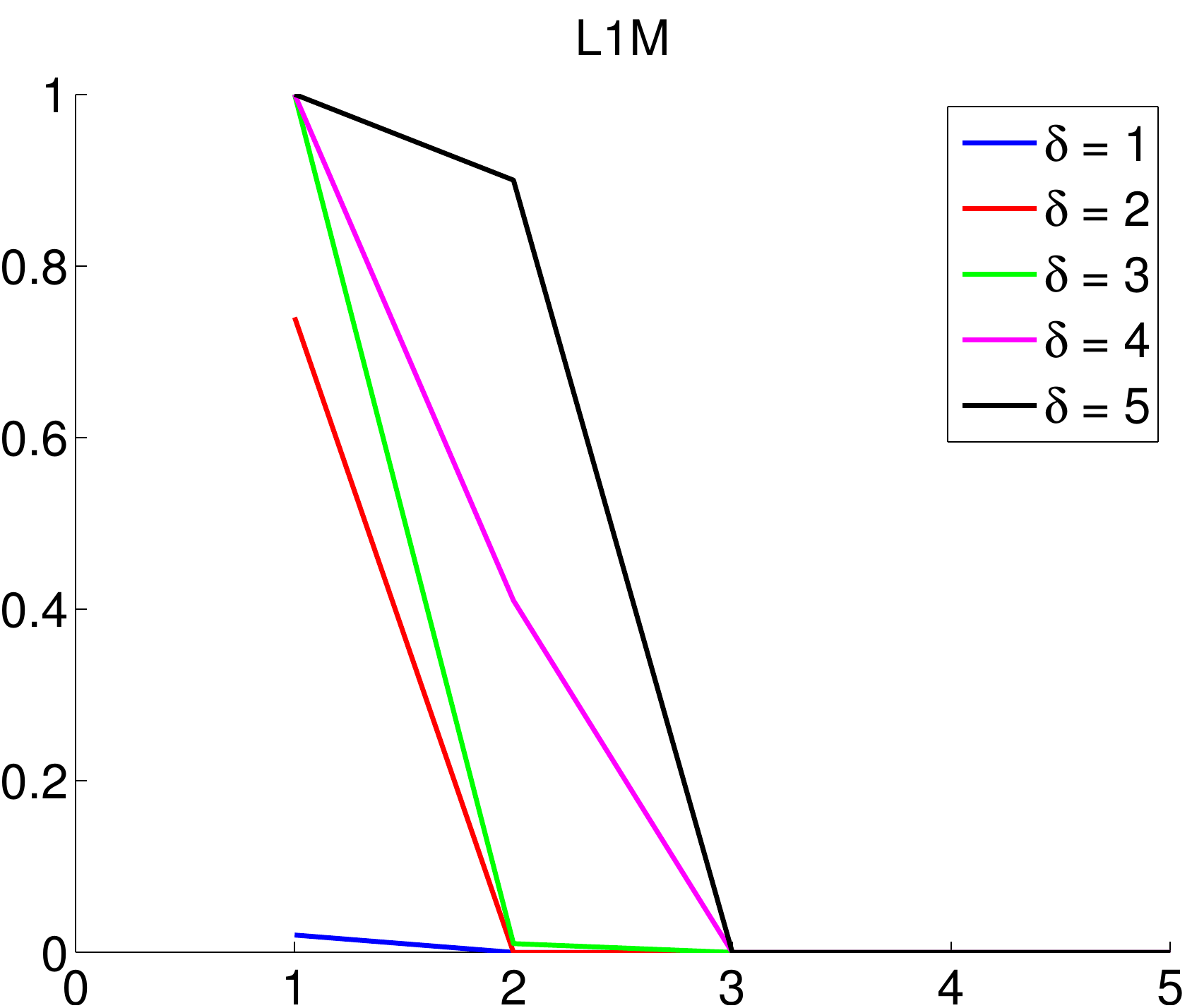}	
	\includegraphics[width = .24\linewidth]{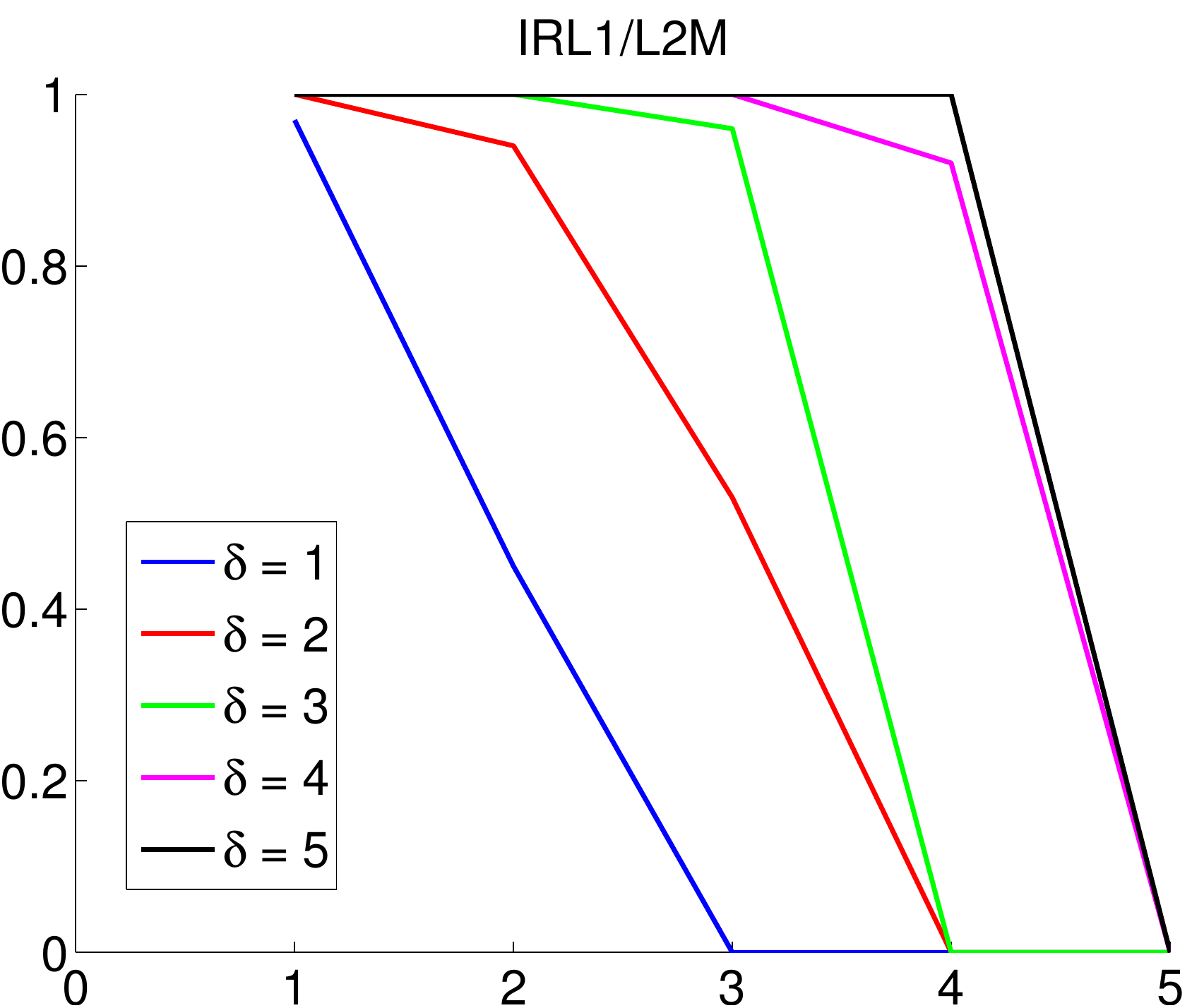}
	\includegraphics[width = .24\linewidth]{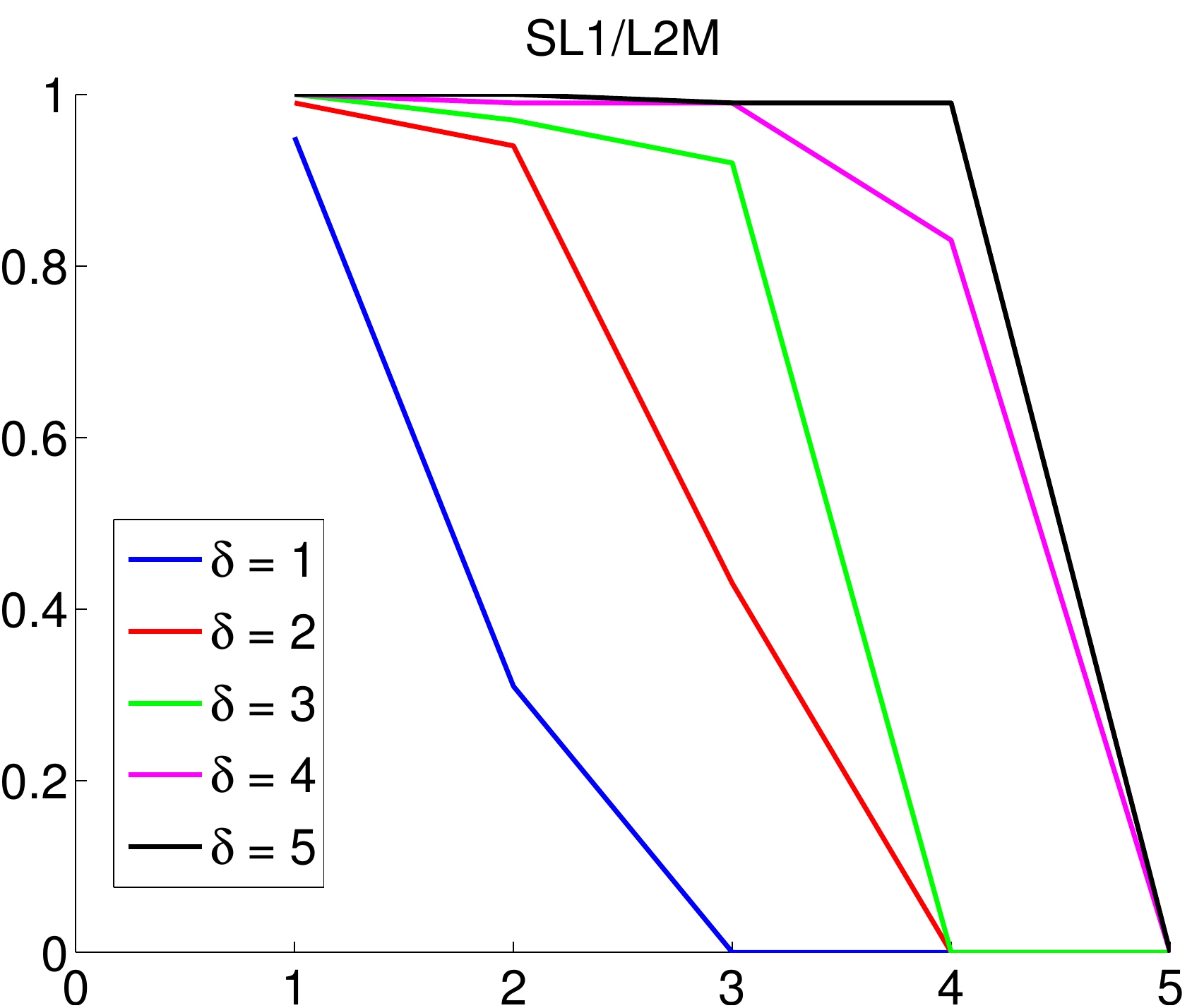}		
	\includegraphics[width = .24\linewidth]{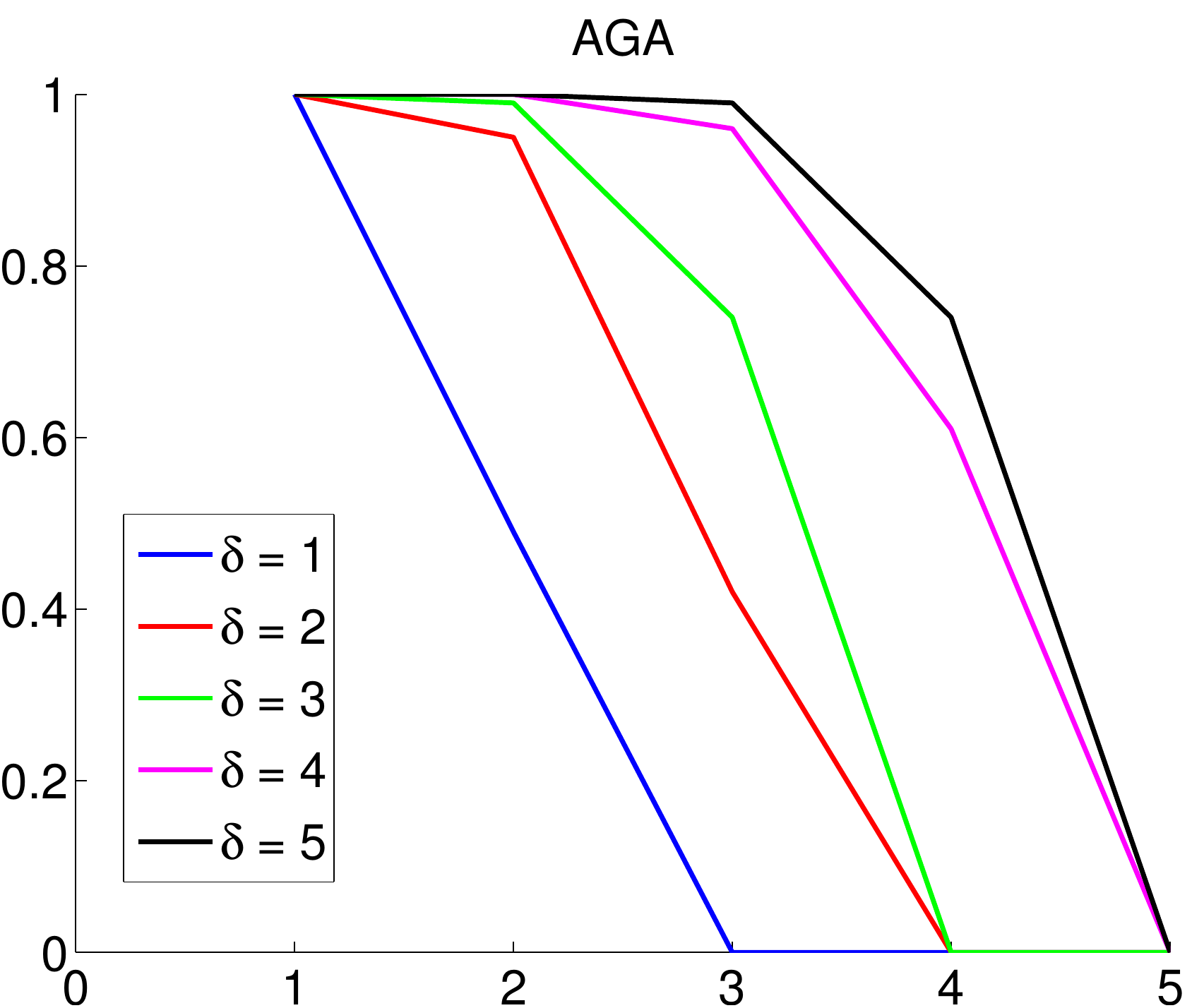}	
	\medskip\\
	$\g{d=4}$\\
	\includegraphics[width = .24\linewidth]{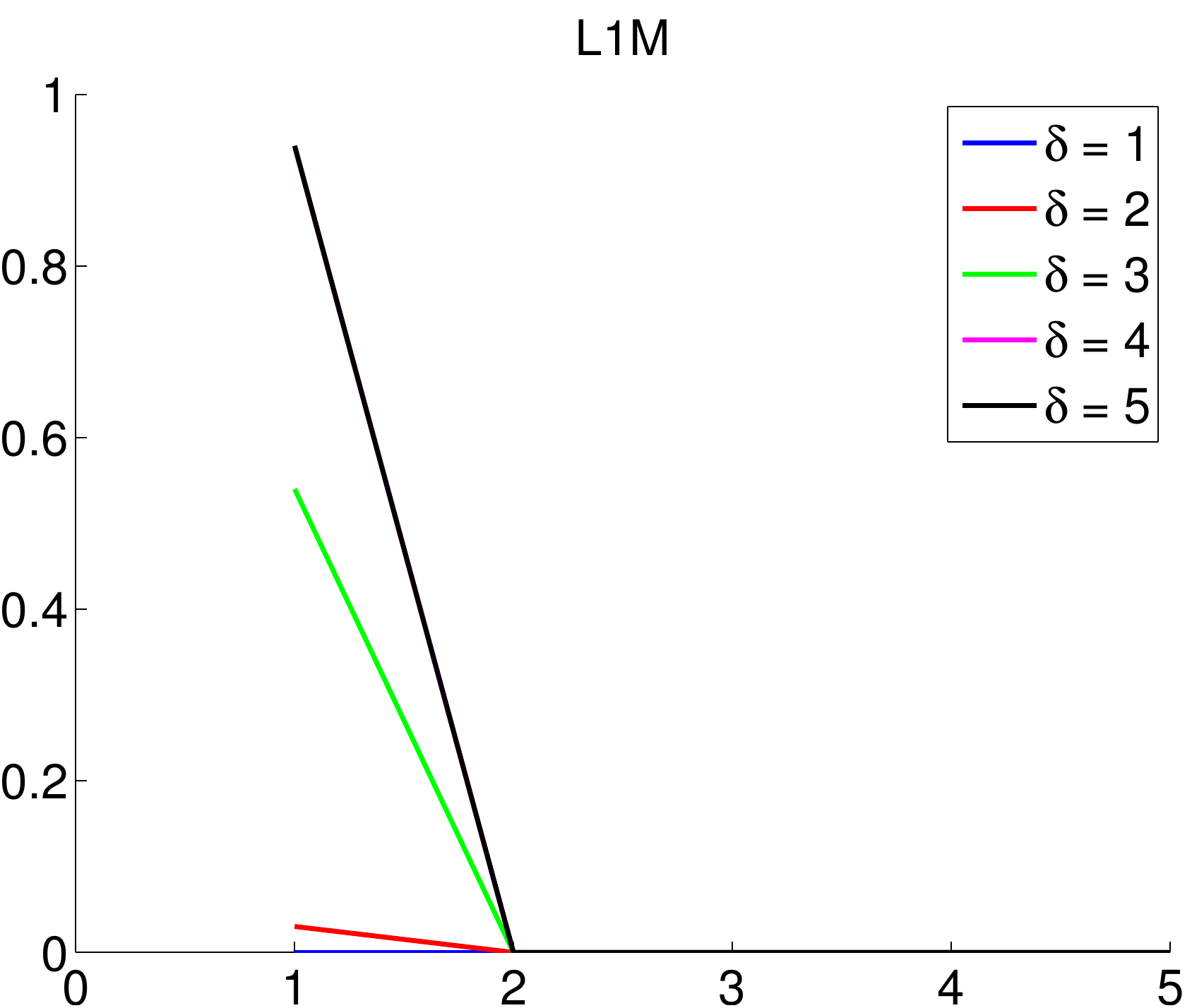}	
	\includegraphics[width = .24\linewidth]{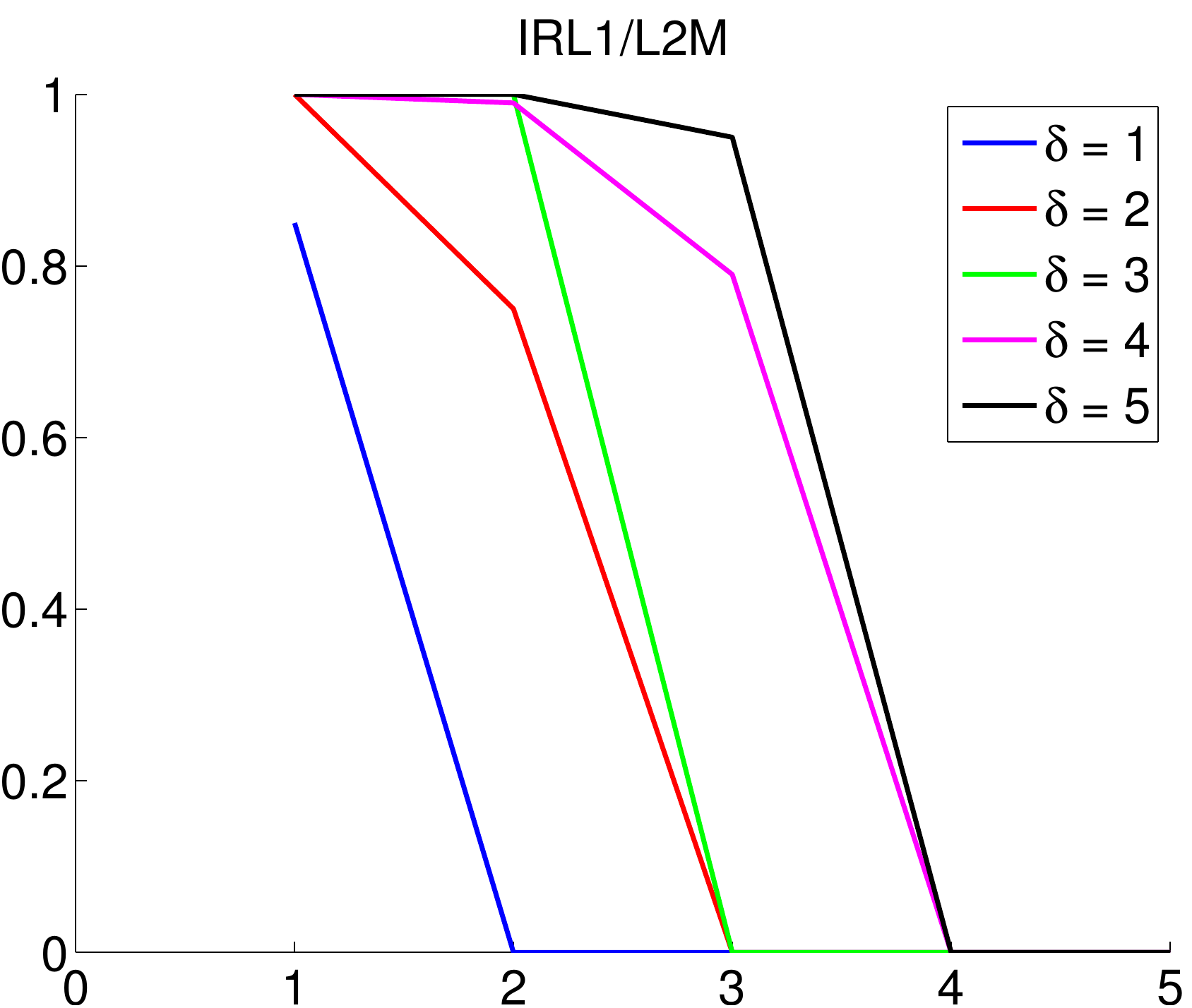}
	\includegraphics[width = .24\linewidth]{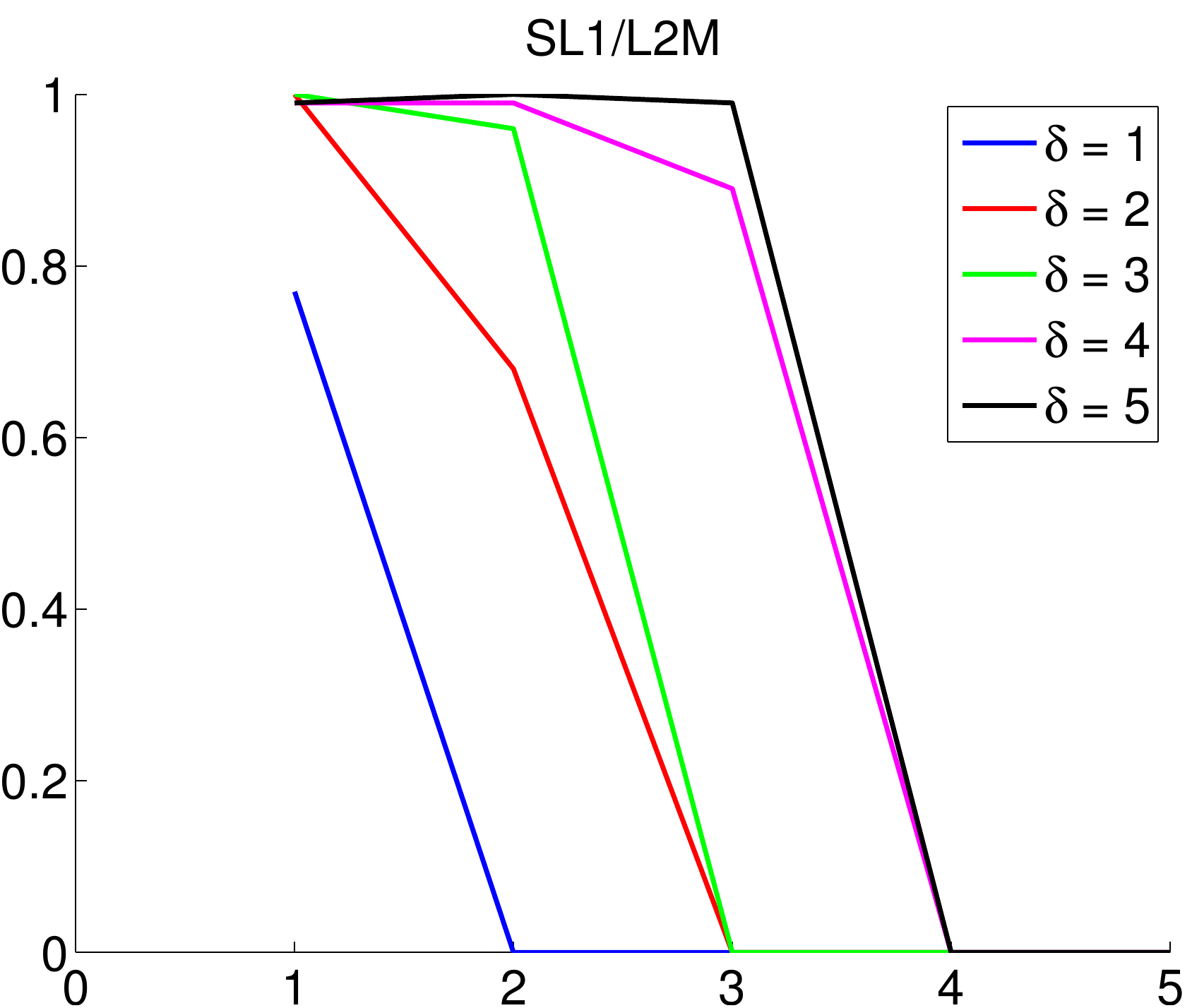}	
	\includegraphics[width = .24\linewidth]{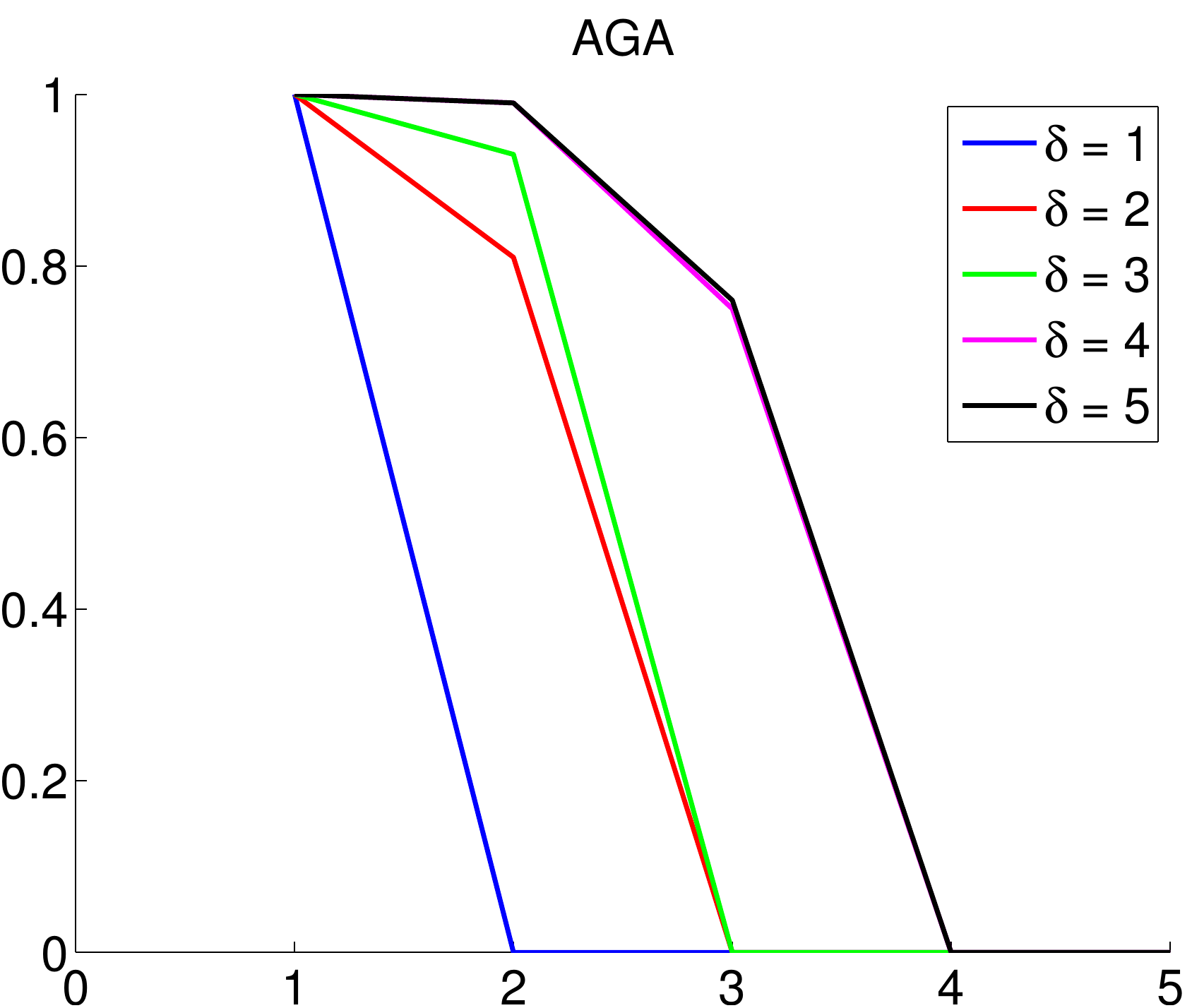}	
	\medskip\\
	$\g{d=5}$\\
	\includegraphics[width = .24\linewidth]{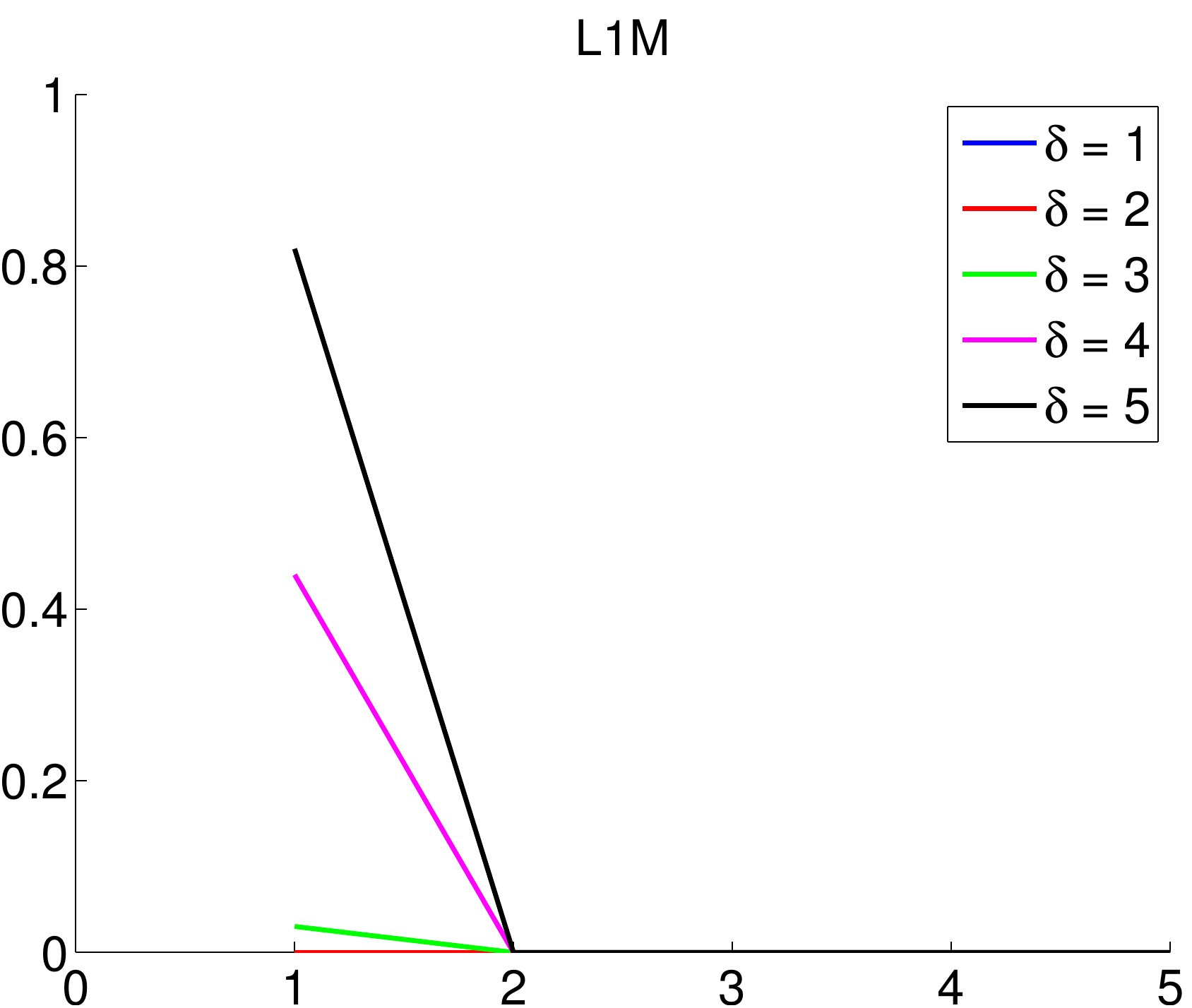}	
	\includegraphics[width = .24\linewidth]{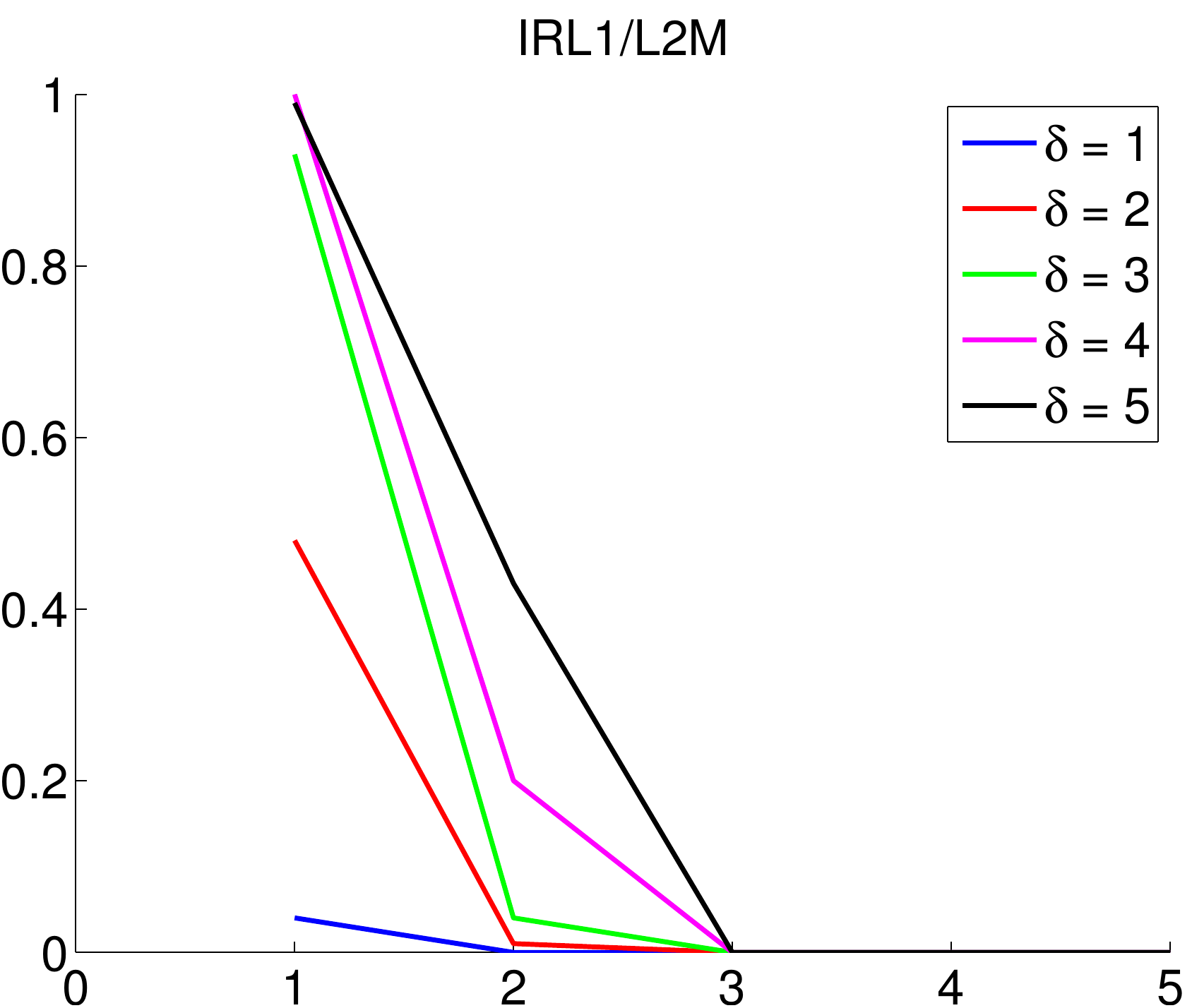}
	\includegraphics[width = .24\linewidth]{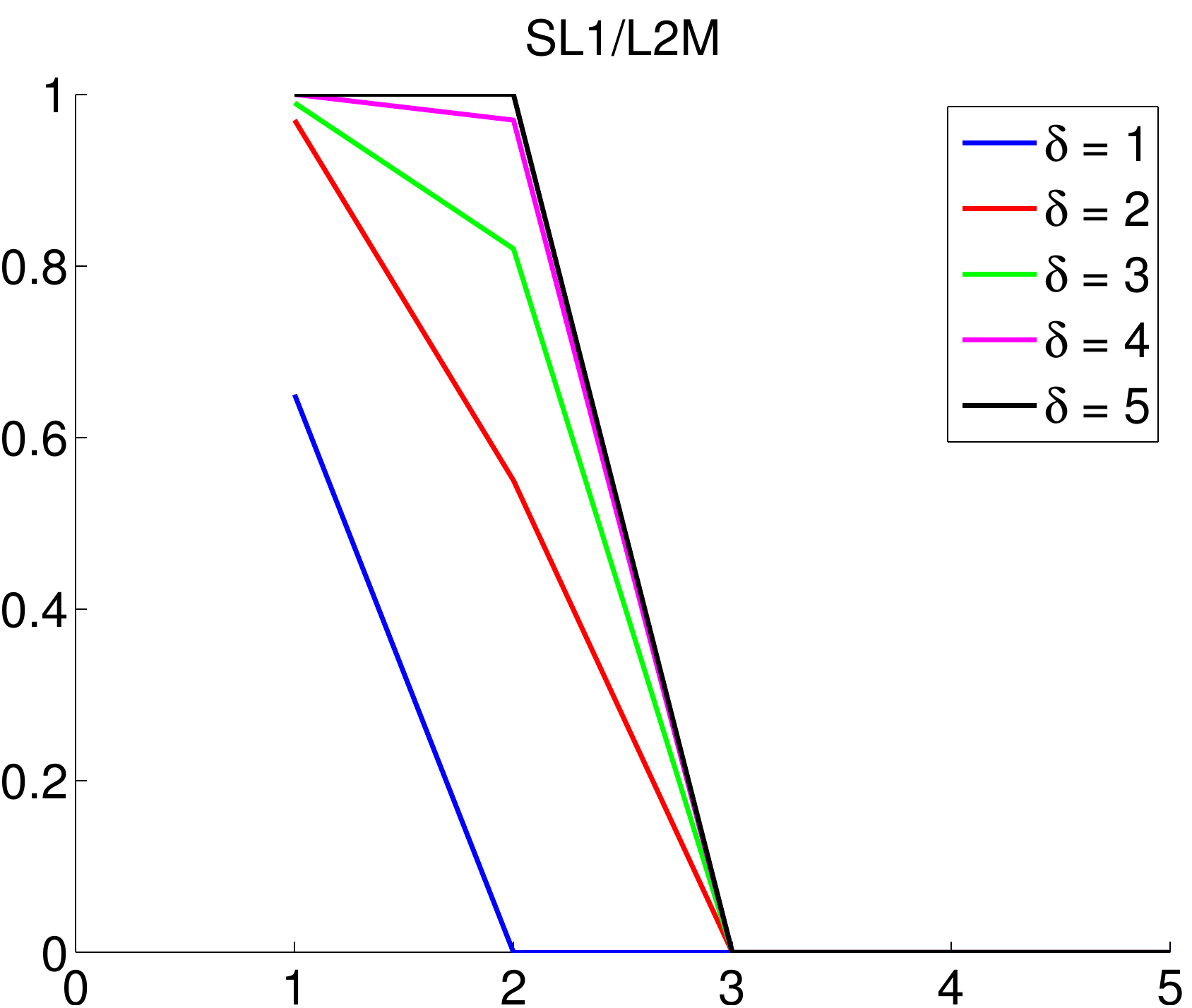}		
	\includegraphics[width = .24\linewidth]{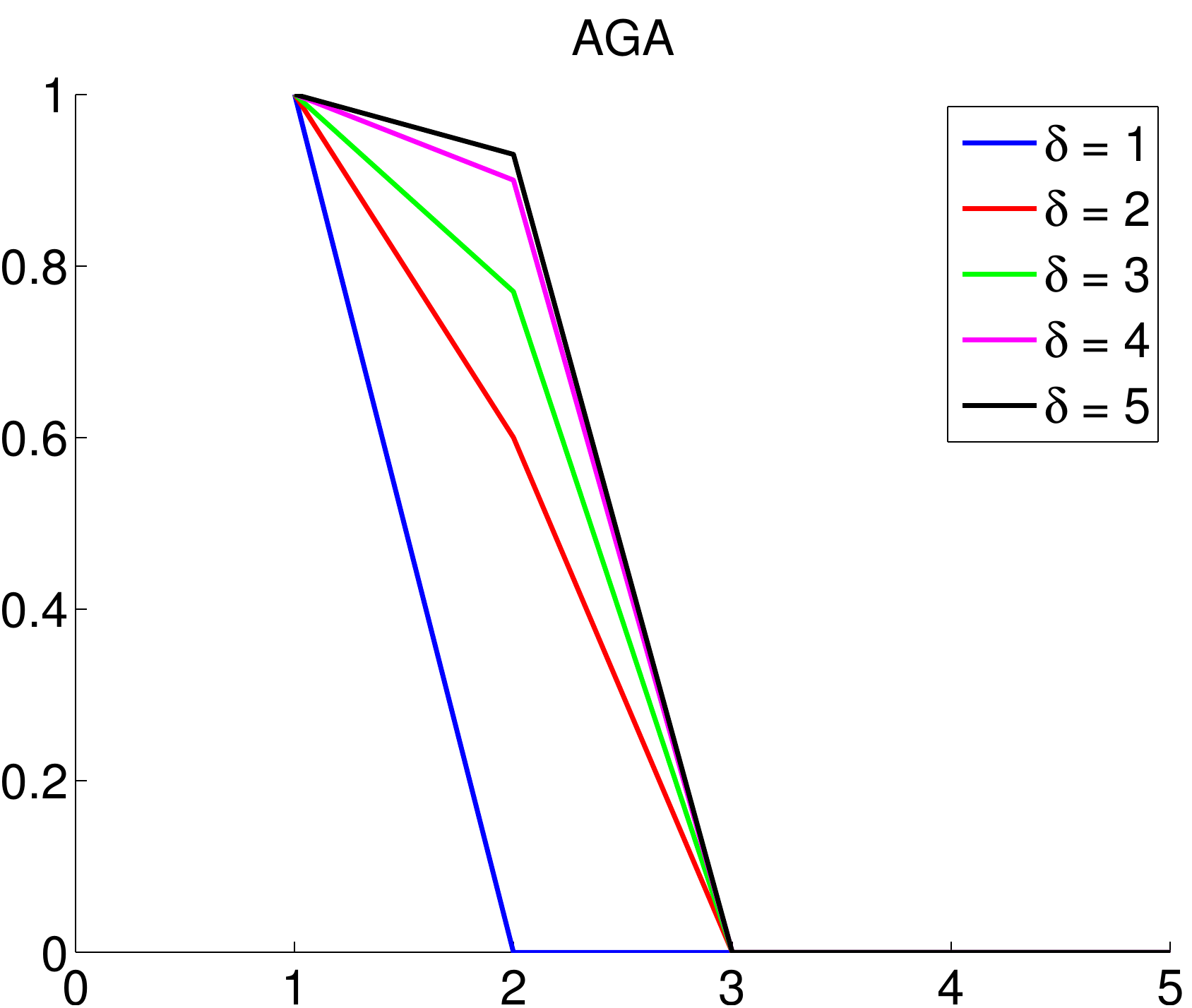}	
	\caption{Estimated probability of successful recovery versus $\|\g x_0\|_0$ for various $\delta = N/n$ in the case $n=10$. Each row considers a different degree $d$, while each column contains the results of a different algorithm (from left to right:  $\ell_1$M, IR$\ell_1\ell_2$M, S$\ell_1\ell_2$M and AGA). \label{fig:diagramn10}} 
\end{figure}
\begin{figure}
	\centering
	$\g{d=2}$\\
	\includegraphics[width = .24\linewidth]{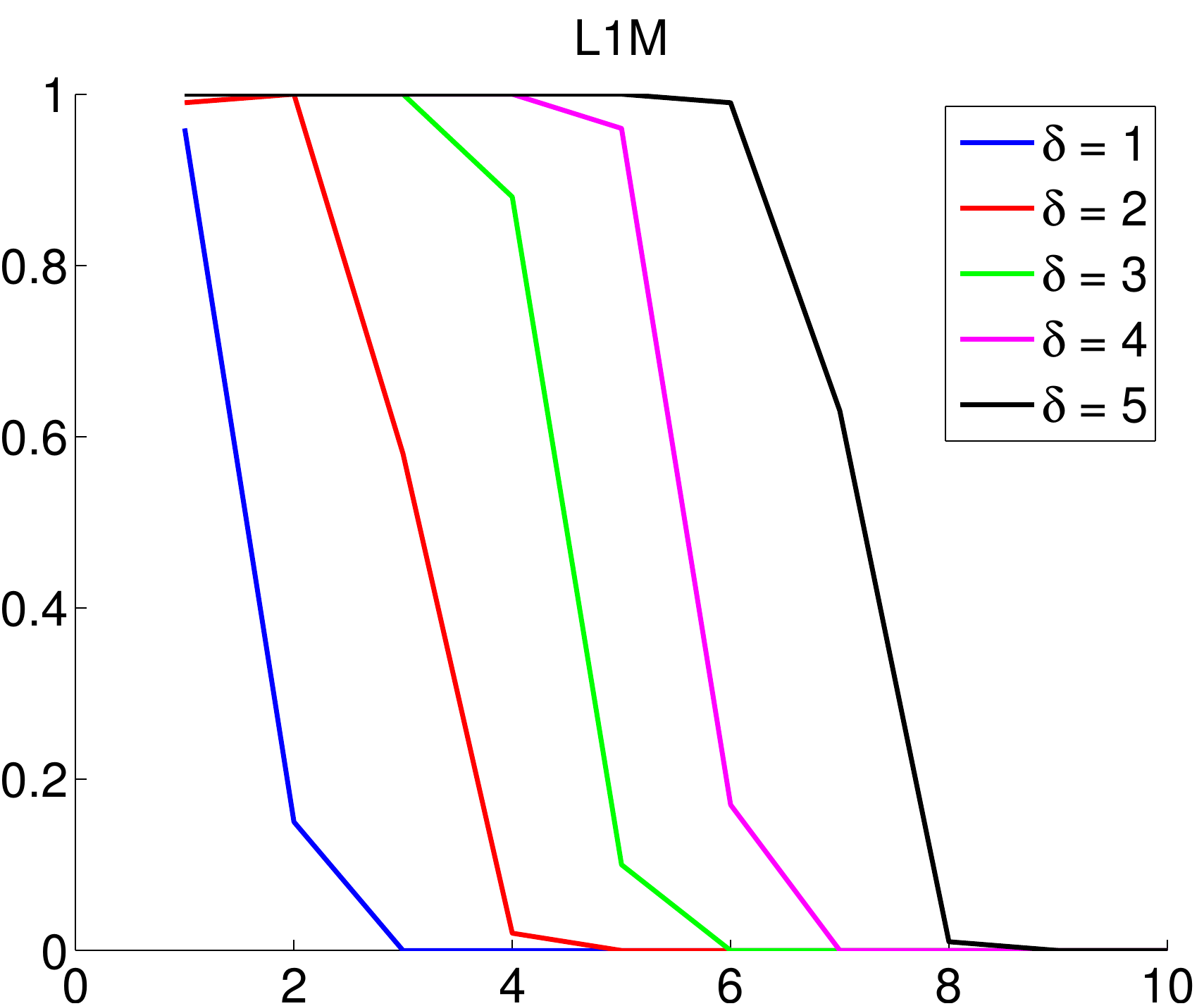}
	\includegraphics[width = .24\linewidth]{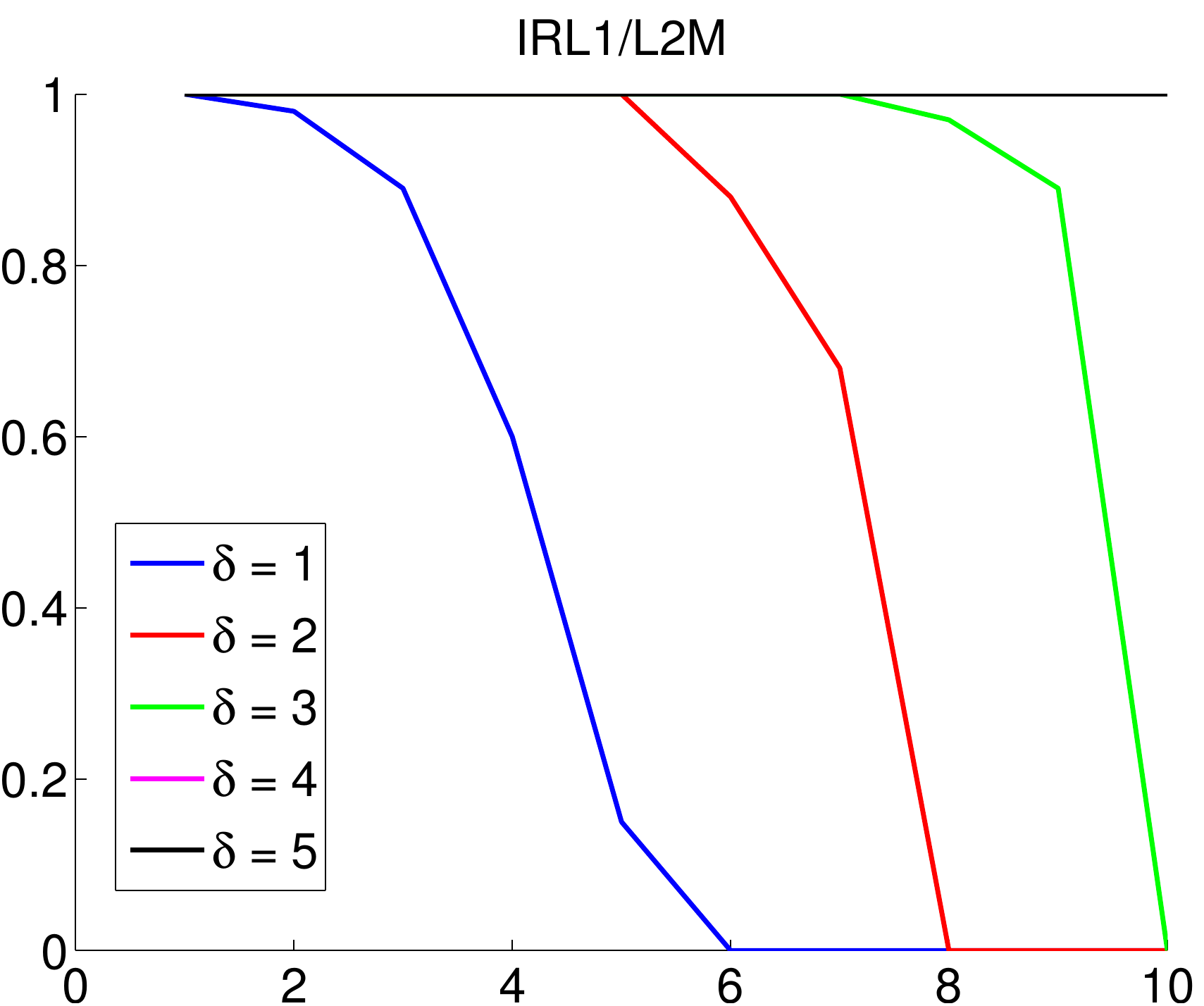}
	\includegraphics[width = .24\linewidth]{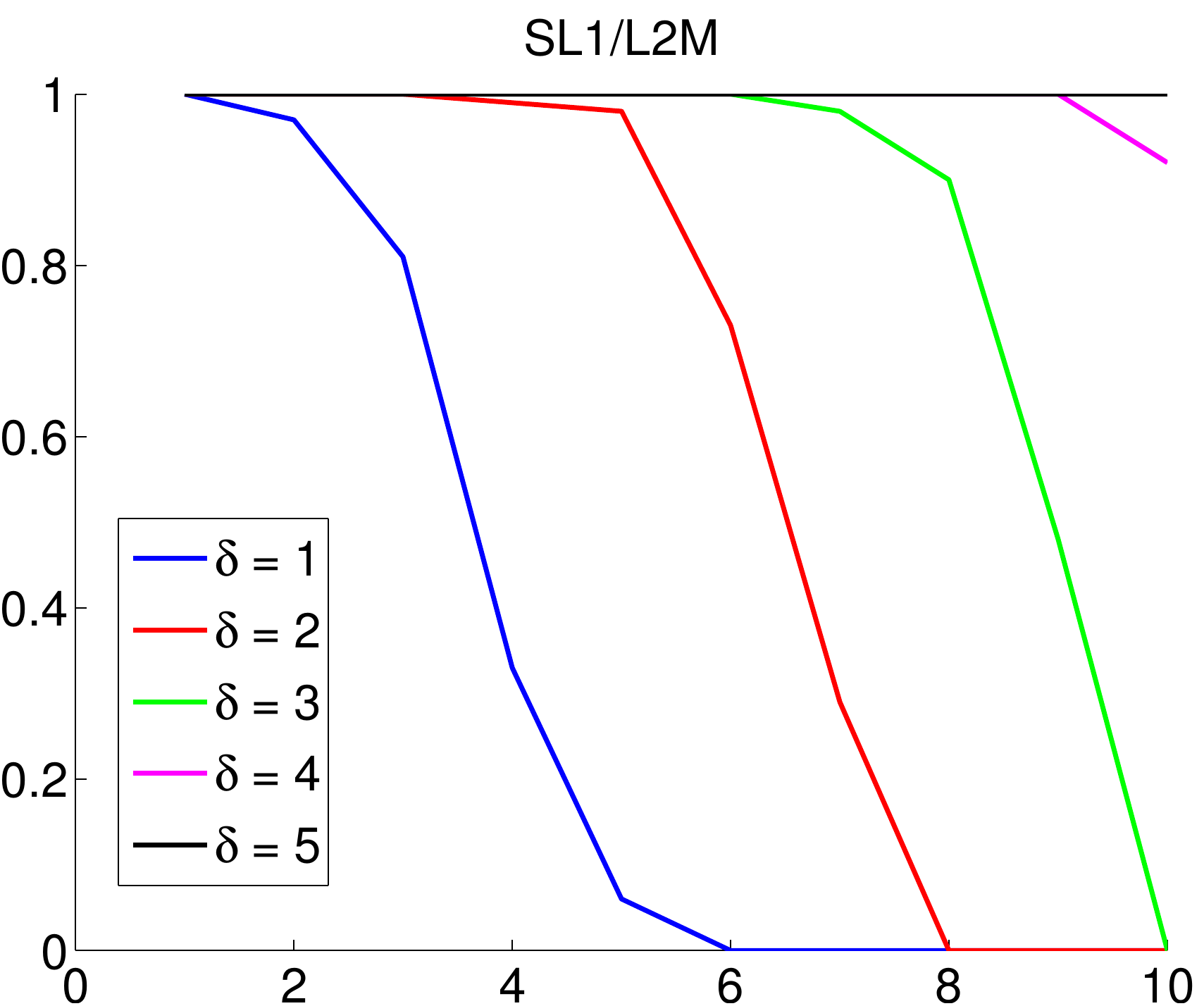}		
	\includegraphics[width = .24\linewidth]{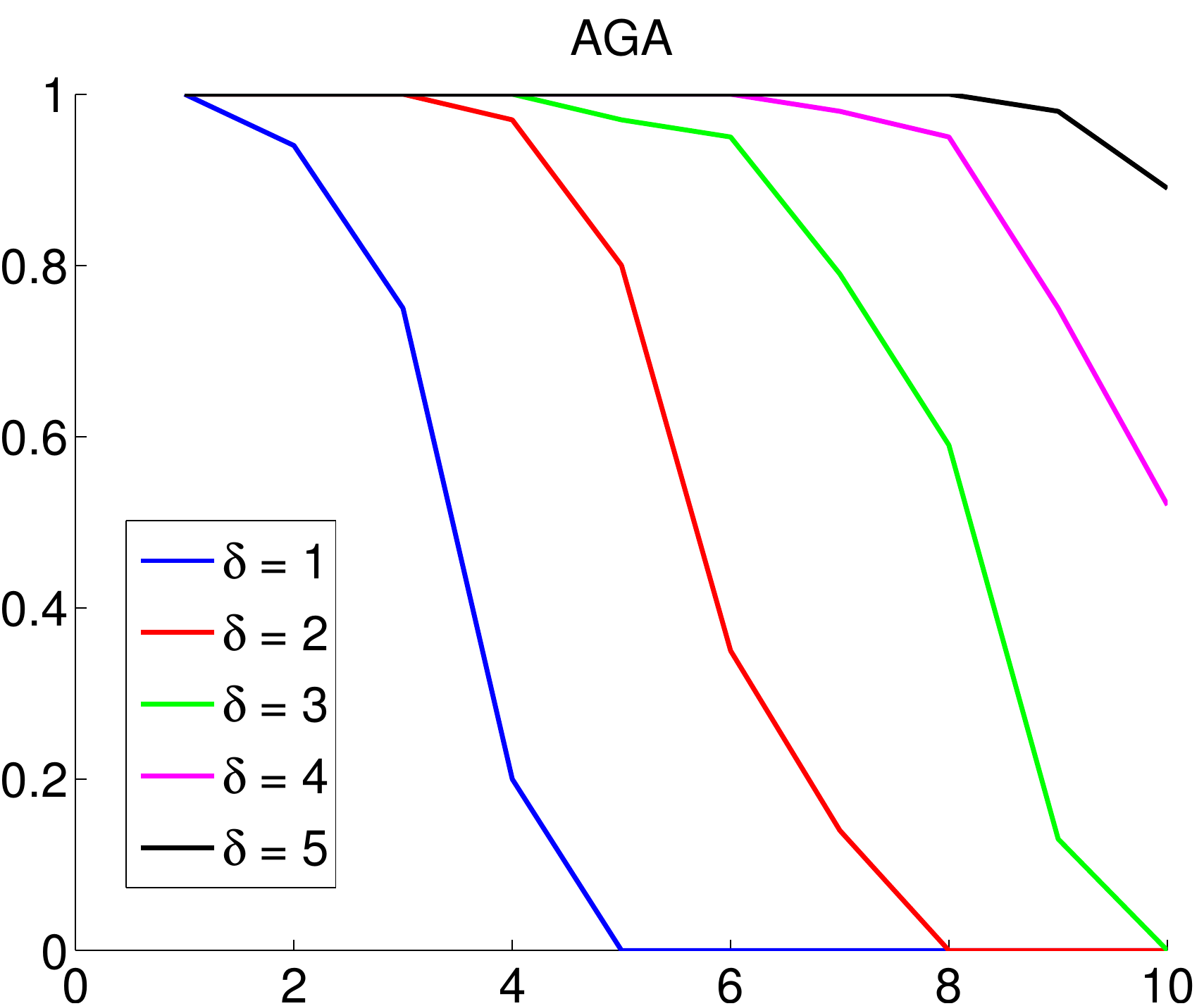}		
	\medskip\\
	$\g{d=3}$\\
	\includegraphics[width = .24\linewidth]{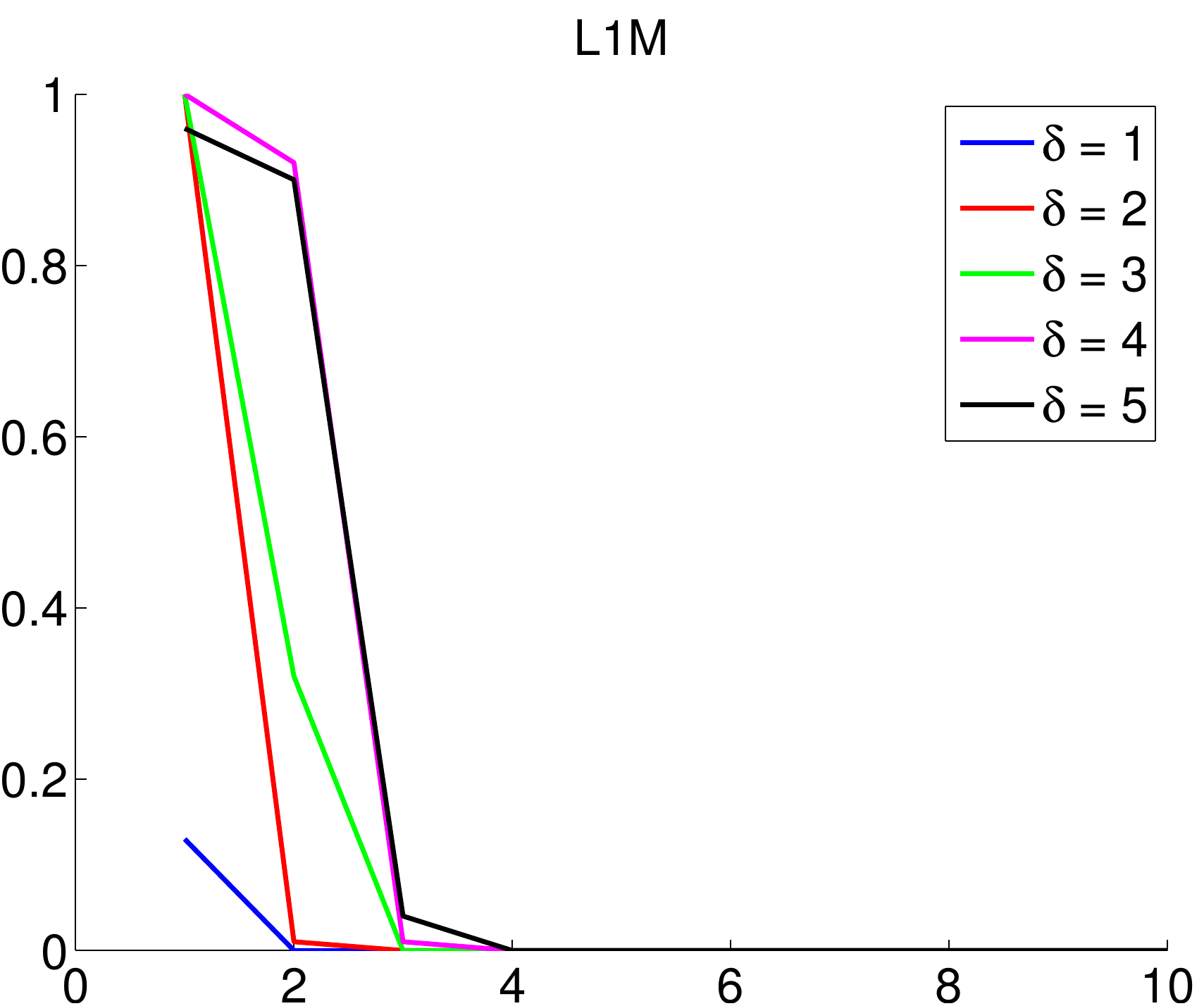}
	\includegraphics[width = .24\linewidth]{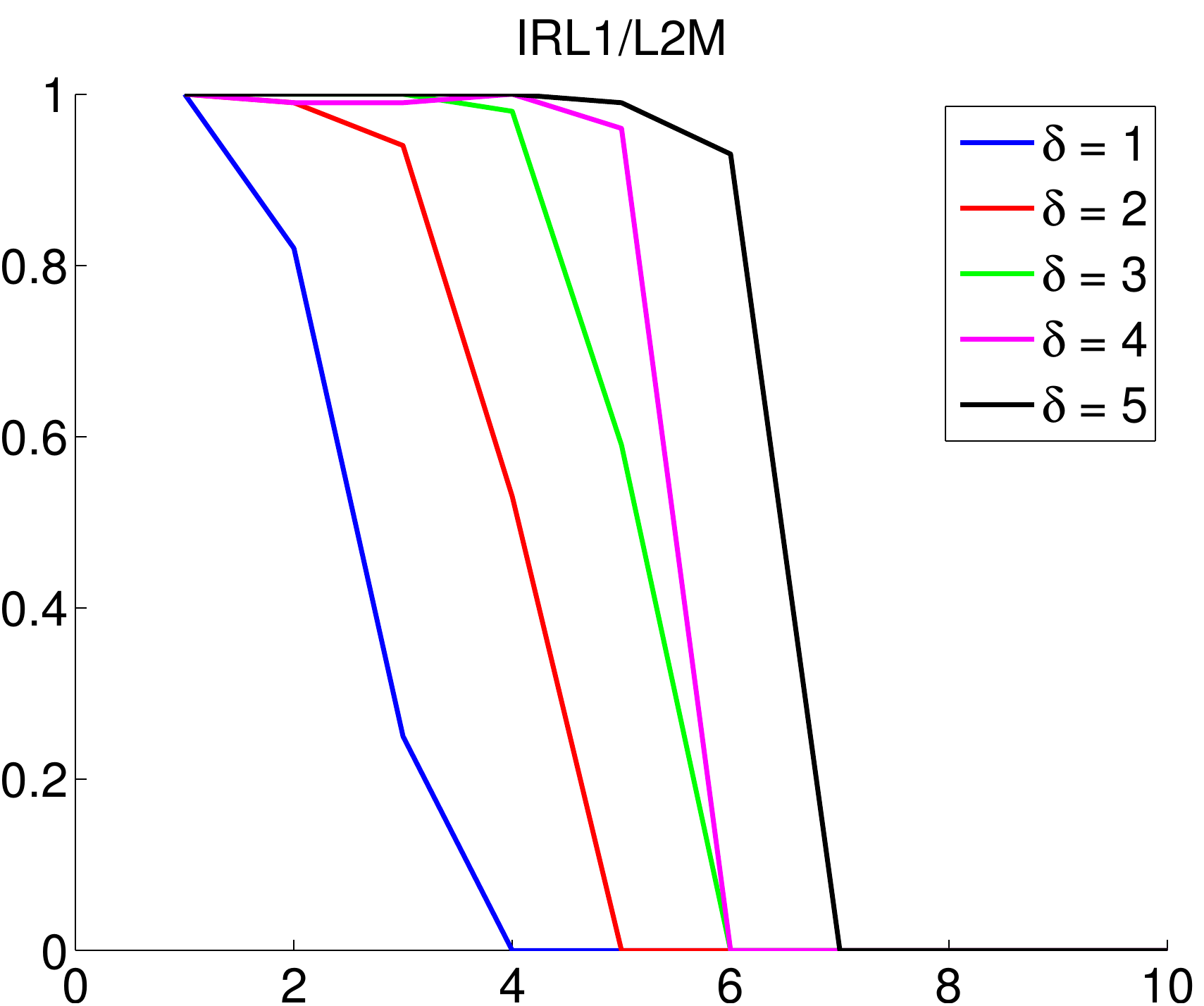}
	\includegraphics[width = .24\linewidth]{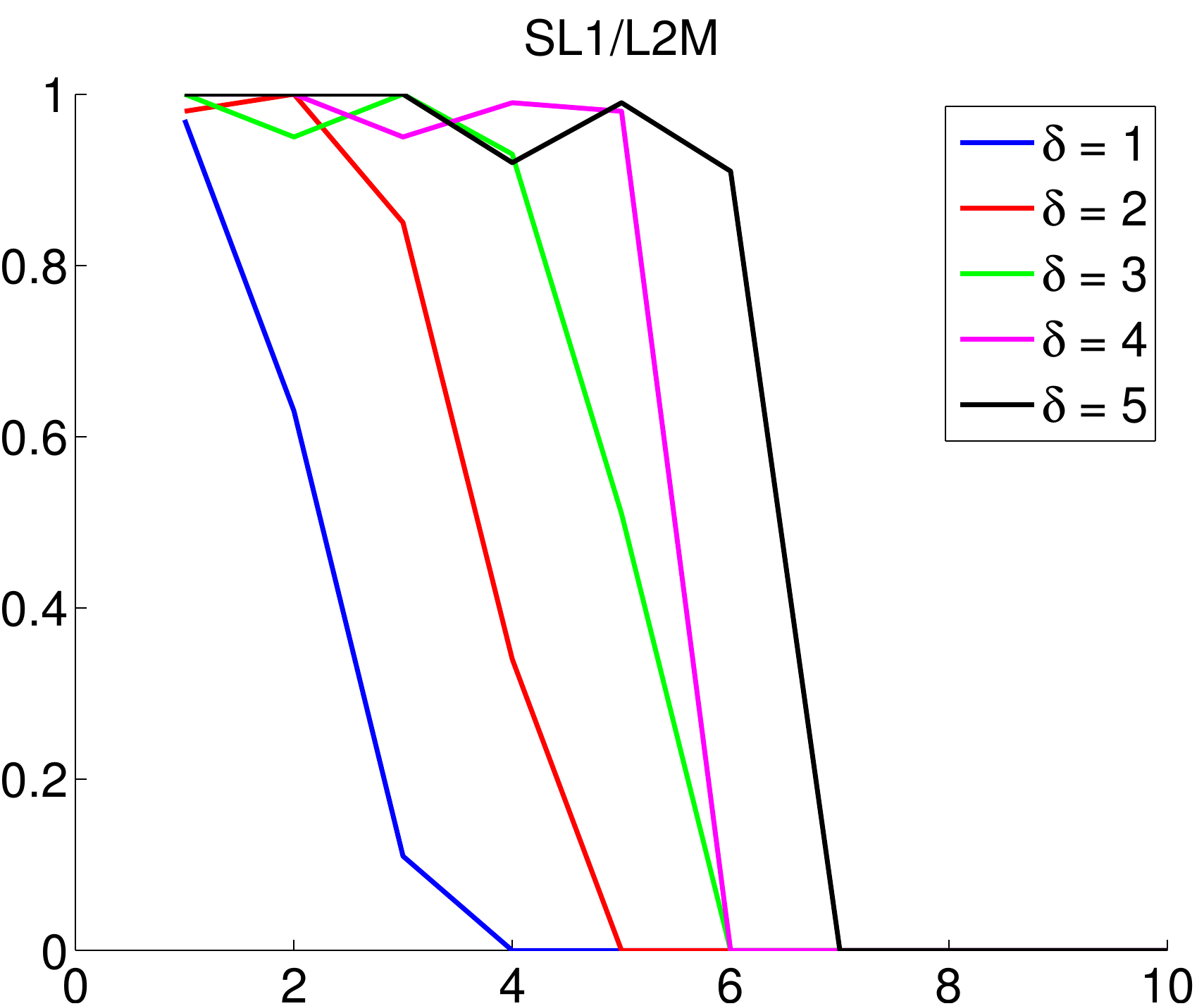}		
	\includegraphics[width = .24\linewidth]{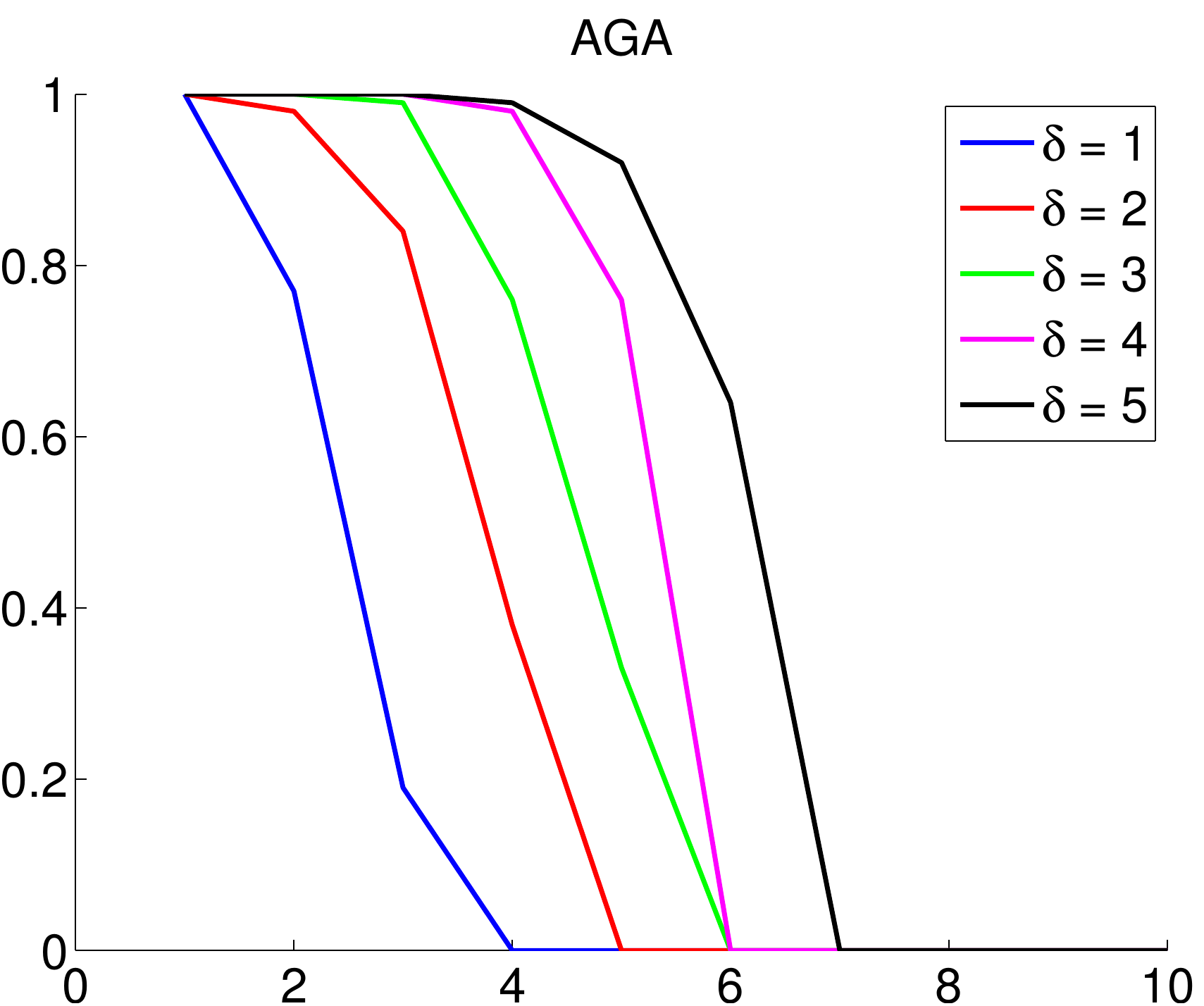}		
	\medskip\\
	$\g{d=4}$\\
	\includegraphics[width = .24\linewidth]{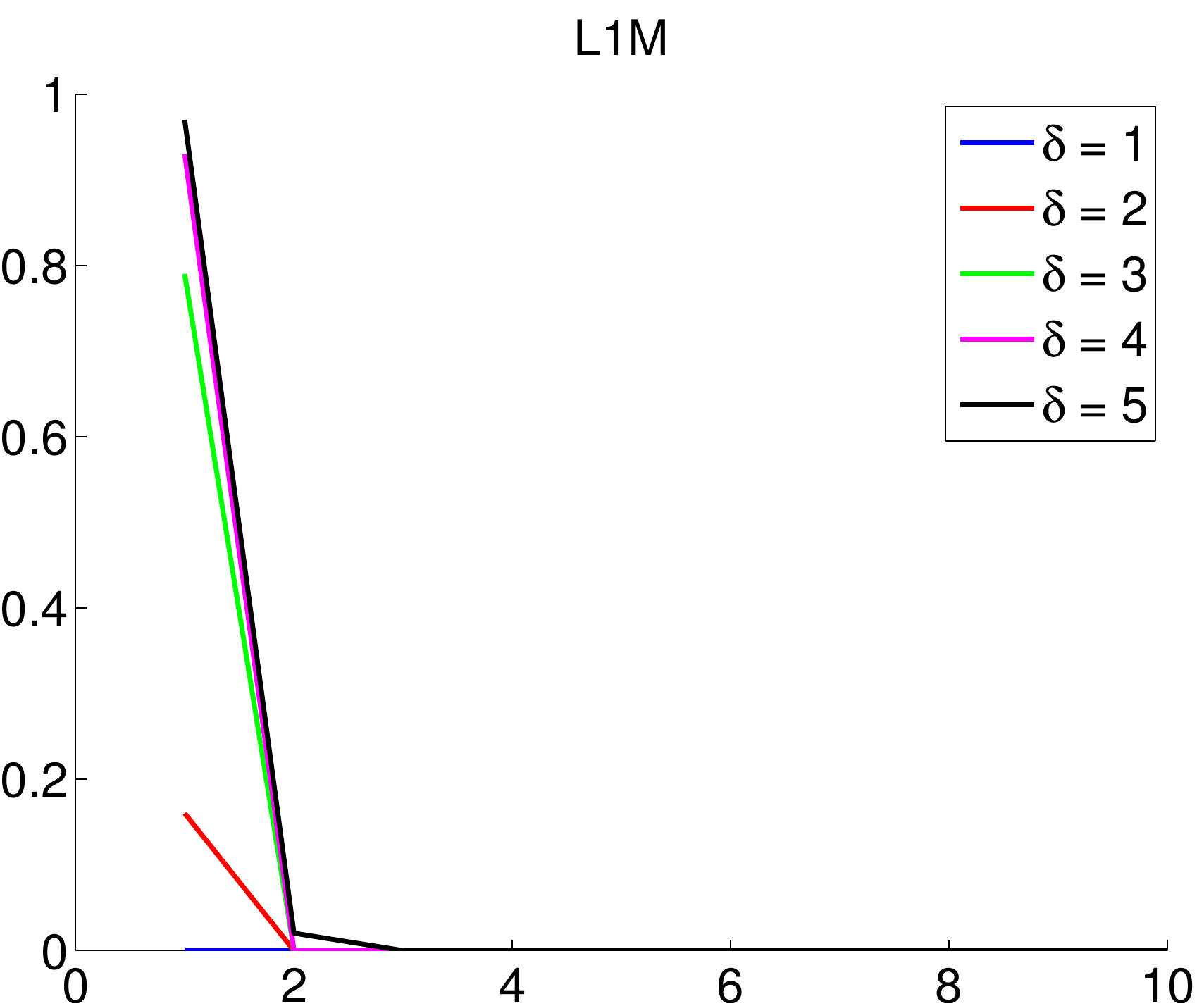}	
	\includegraphics[width = .24\linewidth]{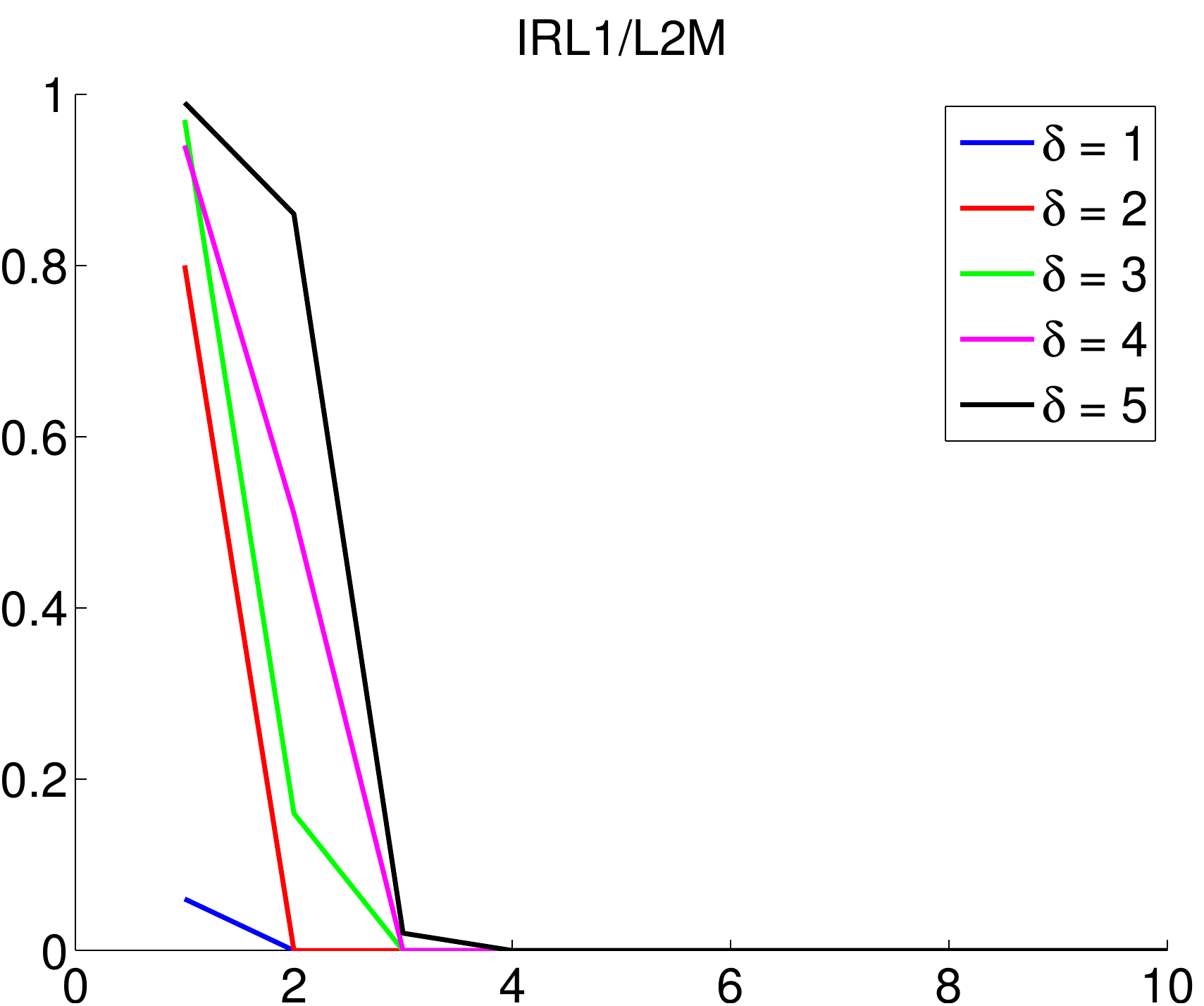}
	\includegraphics[width = .24\linewidth]{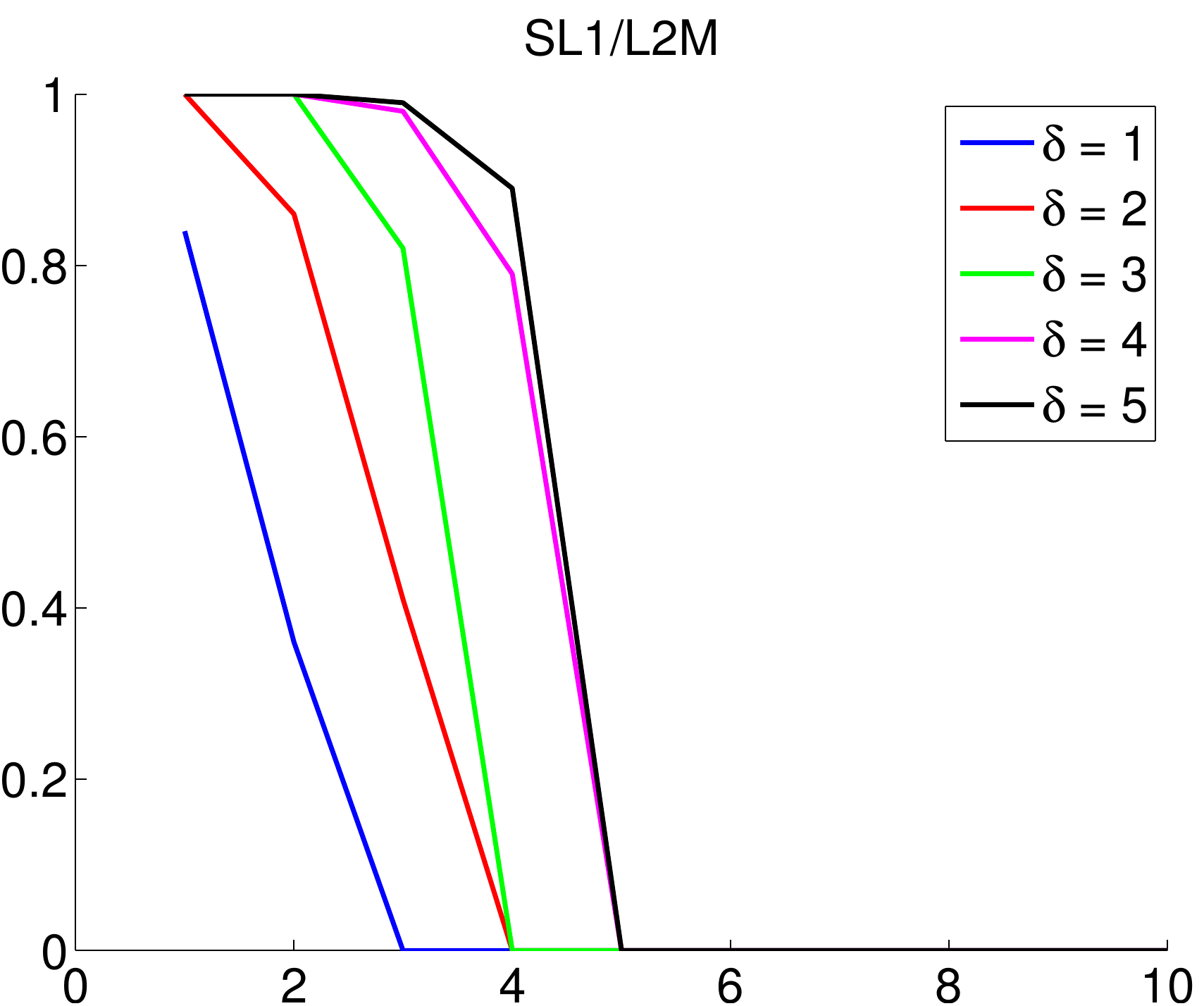}		
	\includegraphics[width = .24\linewidth]{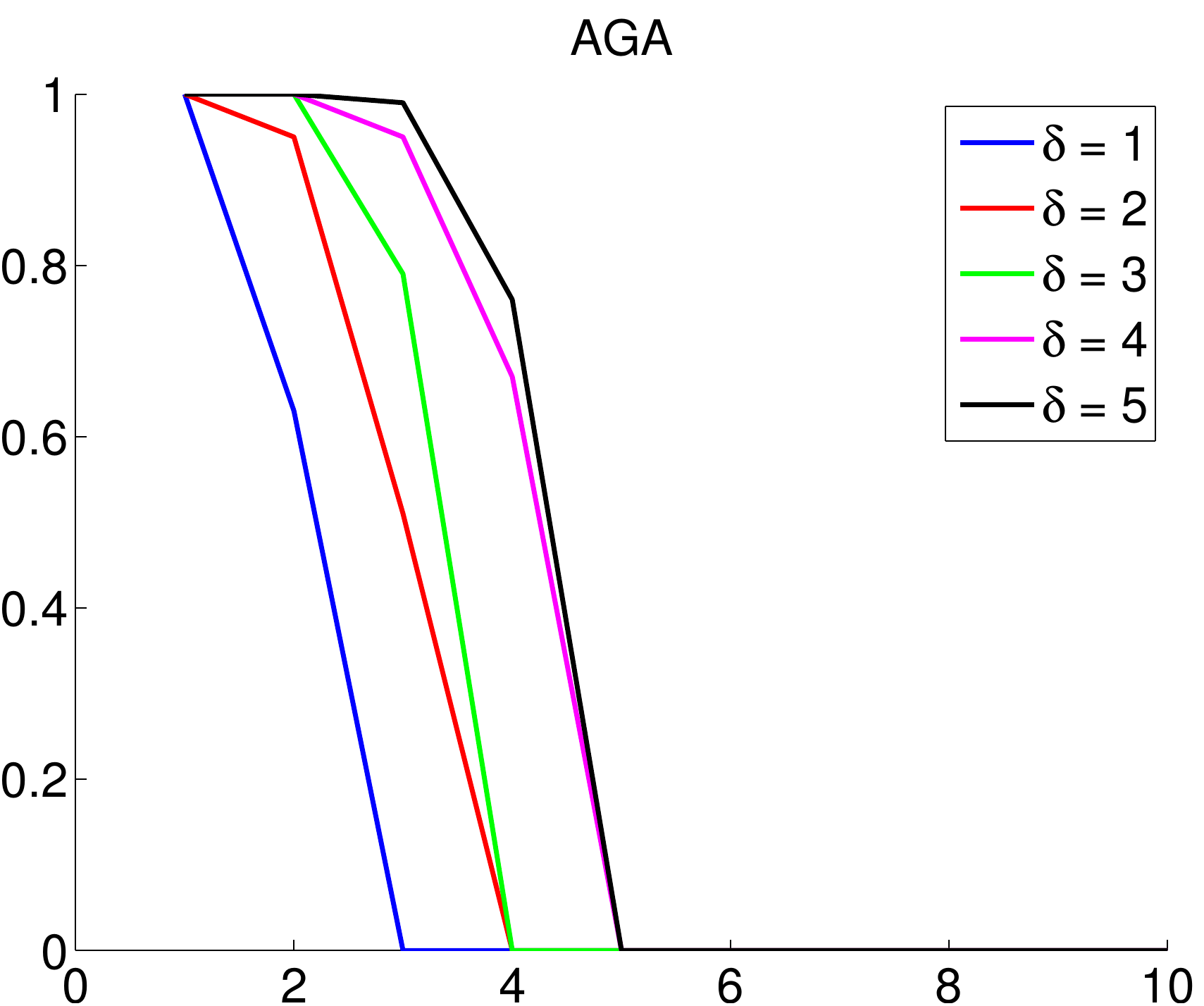}	
	$\g{d=5}$\\
	\includegraphics[width = .24\linewidth]{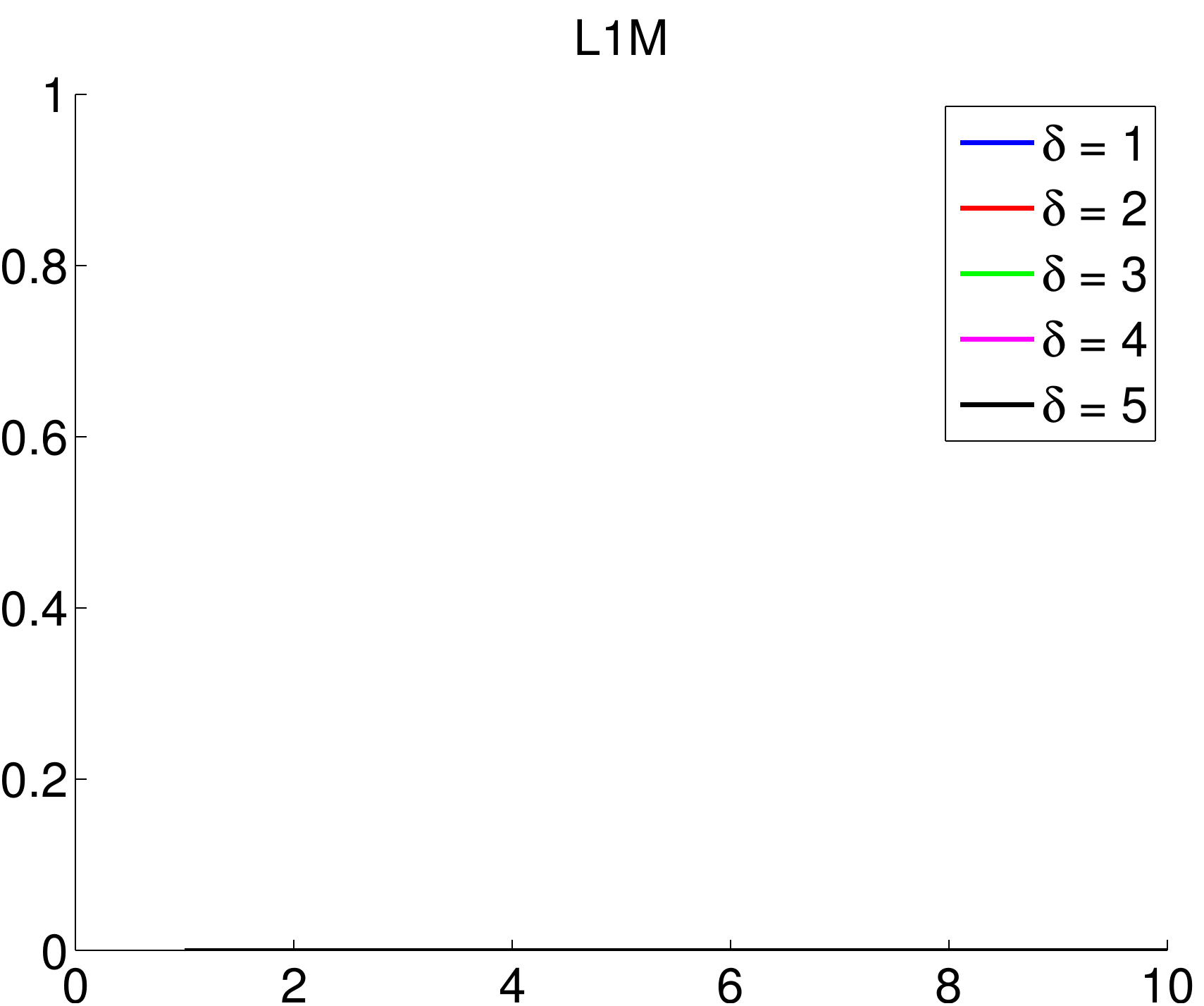}	
	\includegraphics[width = .24\linewidth]{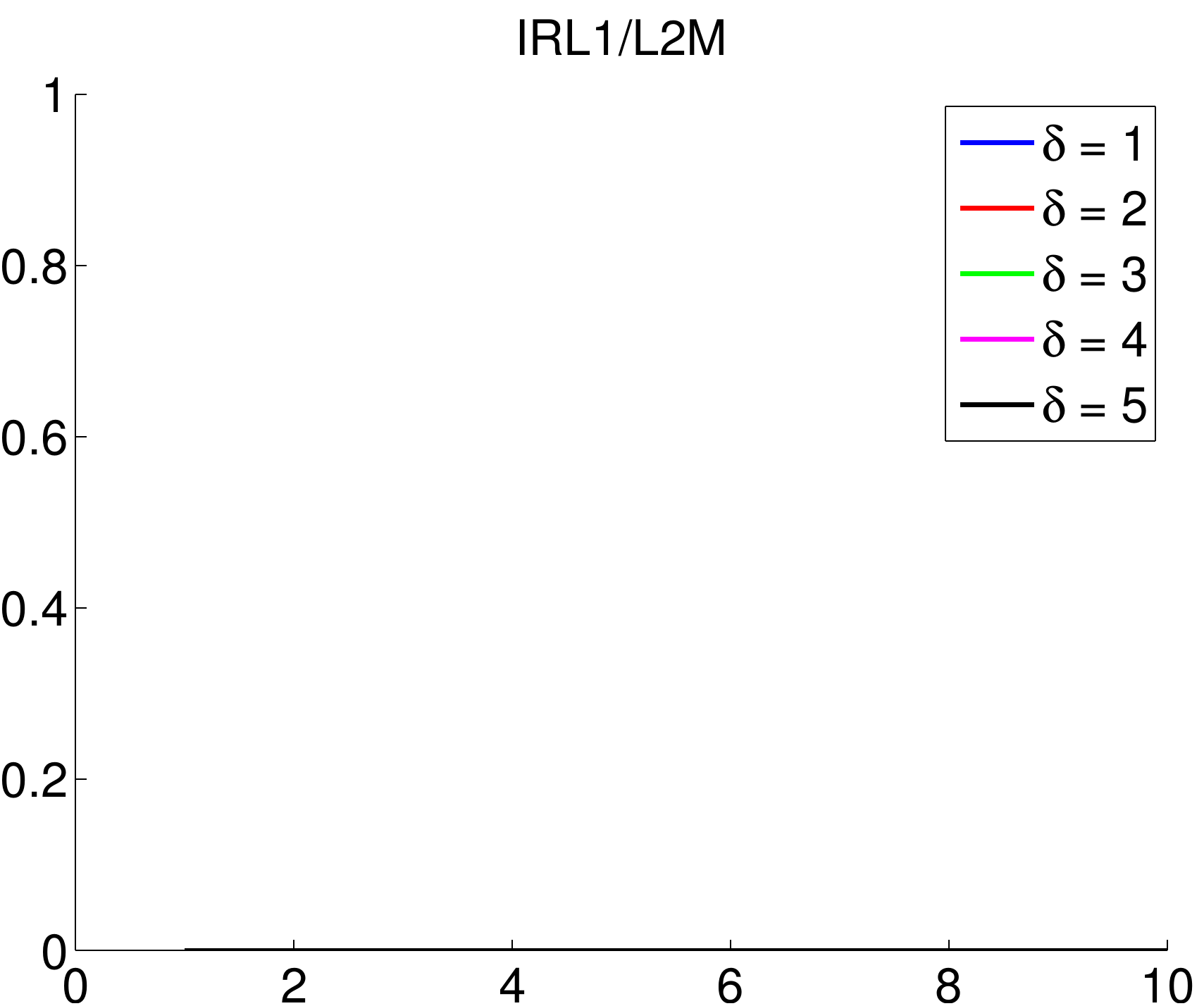}
	\includegraphics[width = .24\linewidth]{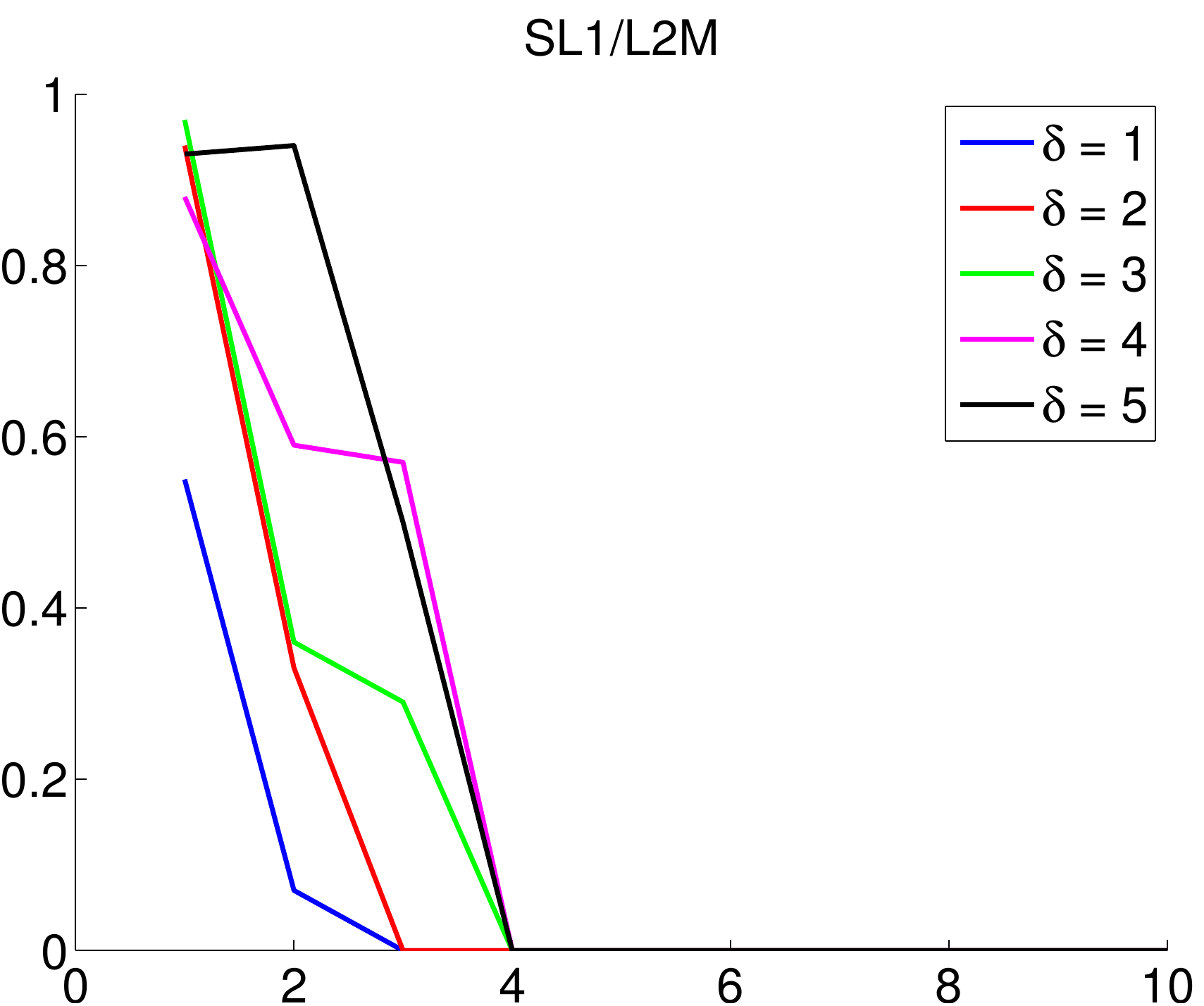}		
	\includegraphics[width = .24\linewidth]{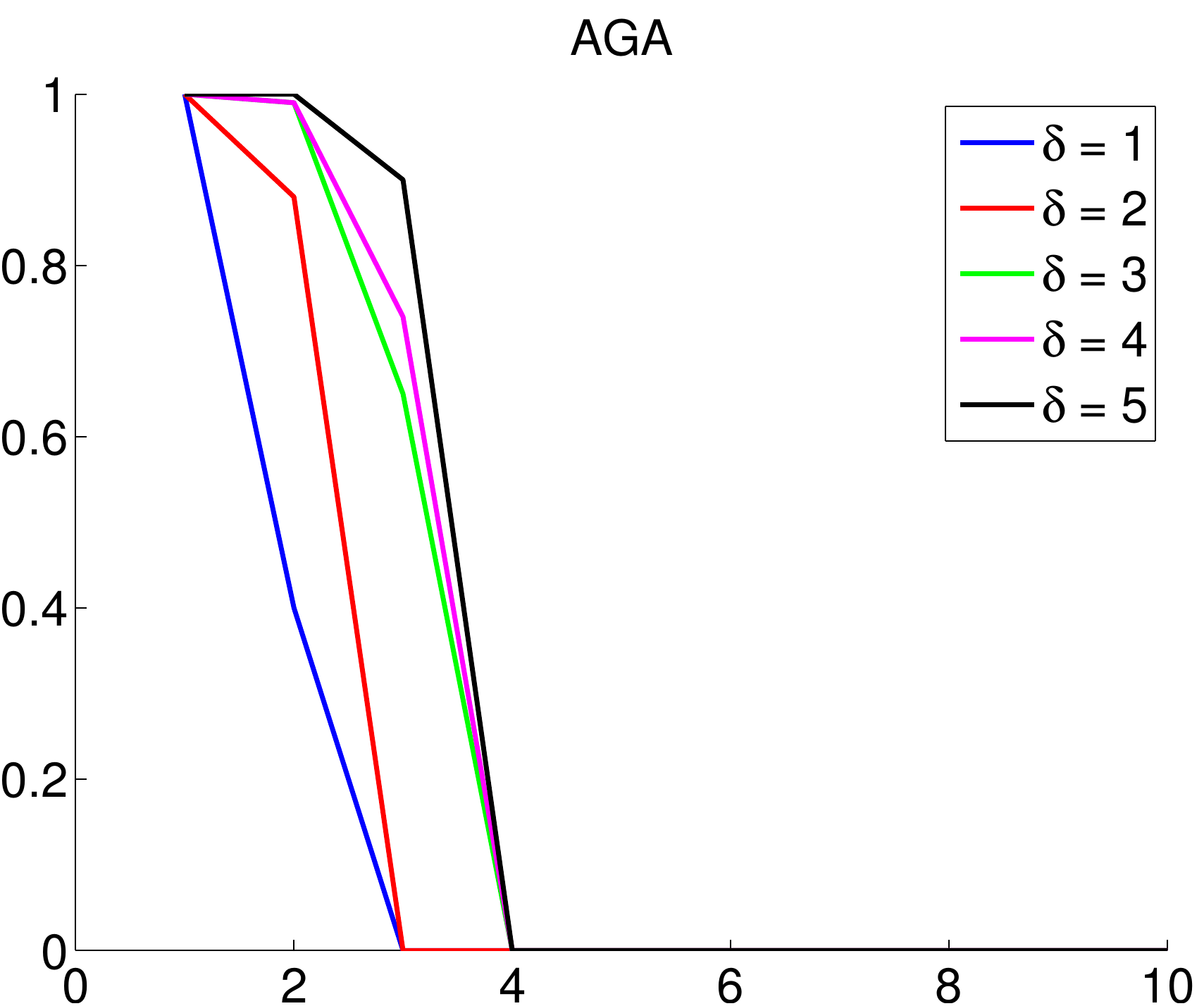}	
	\caption{Estimated probability of successful recovery versus $\|\g x_0\|_0$ for various $\delta = N/n$ in the case $n=20$. Each row considers a different degree $d$, while each column contains the results of a different algorithm (from left to right: $\ell_1$M, IR$\ell_1\ell_2$M, S$\ell_1\ell_2$M and AGA).\label{fig:diagramn20}} 
\end{figure}

\paragraph{Computational complexity.} 
The next experiments evaluate the computational complexity of the methods with respect to the number of variables $n$, the degree $d$ and the sparsity level $\|\g x_0\|_0$. Focusing only on the computing time, we consider favorable settings with these parameters and the number of equations $N$ set such that $\g x_0$ should be recovered in most cases. The left plot of Fig.~\ref{fig:time} shows the mean time of the methods for a range of $n\in[10,40]$ with $\|\g x_0\|_0=3$, $d=2$ and $N=50$. They all have a complexity exponential in $n$ due to the dimension $M$ of $\g \phi$. In addition, they have similar rates of increase, except the EGA for which the computing time grows faster with $n$ due to the combinatorial nature of the algorithm. Similar results, shown in the middle plot, are obtained with respect to the degree $d$ (for $n=10$, $\|\g x_0\|_0=2$ and $N=50$), except that the greedy algorithms benefit from a much lower rate of increase. 
In these experiments, the reweighting scheme of the S$\ell_1\ell_2$M is faster than the one used by the IR$\ell_1\ell_2$M, but its computing time highly depends on $\|\g x_0\|_0$ whereas the number of iterations is fixed for IR$\ell_1\ell_2$M. This is clearly seen from the right plot of Fig.~\ref{fig:time} (obtained with $n=20$, $d=2$ and $N=100$), where the computing time of the S$\ell_1\ell_2$M reaches the one of the IR$\ell_1\ell_2$M when $\|\g x_0\|_0$ equals the number of iterations of IR$\ell_1\ell_2$M. The exact greedy algorithm (EGA) suffers from a similar complexity issue with respect to $\|\g x_0\|_0$, but much more pronounced. On the contrary, its approximate variant (AGA) remained very efficient in all of our tests. 
\begin{figure}
\centering
	\includegraphics[width=.32\linewidth]{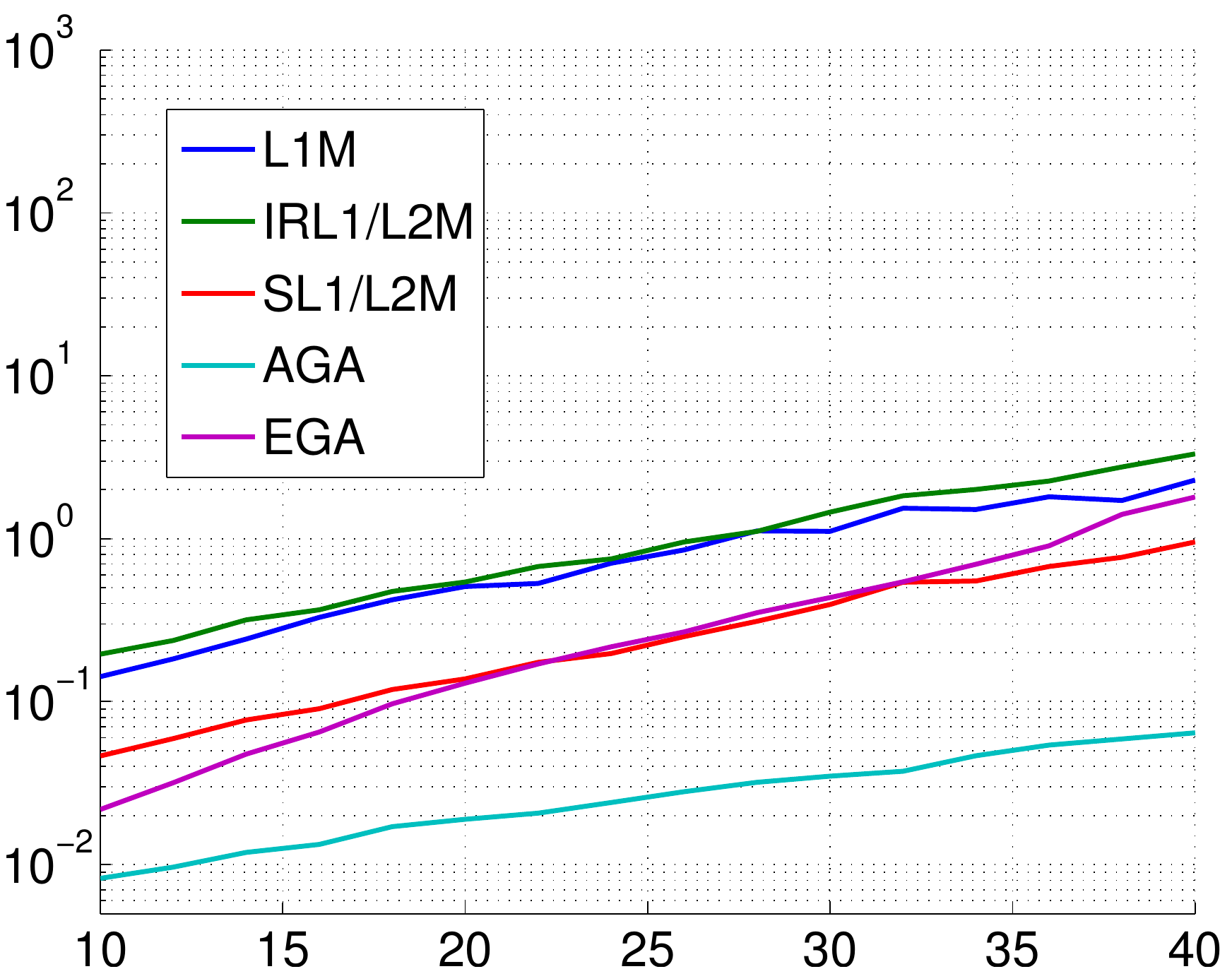} 
	\includegraphics[width=.32\linewidth]{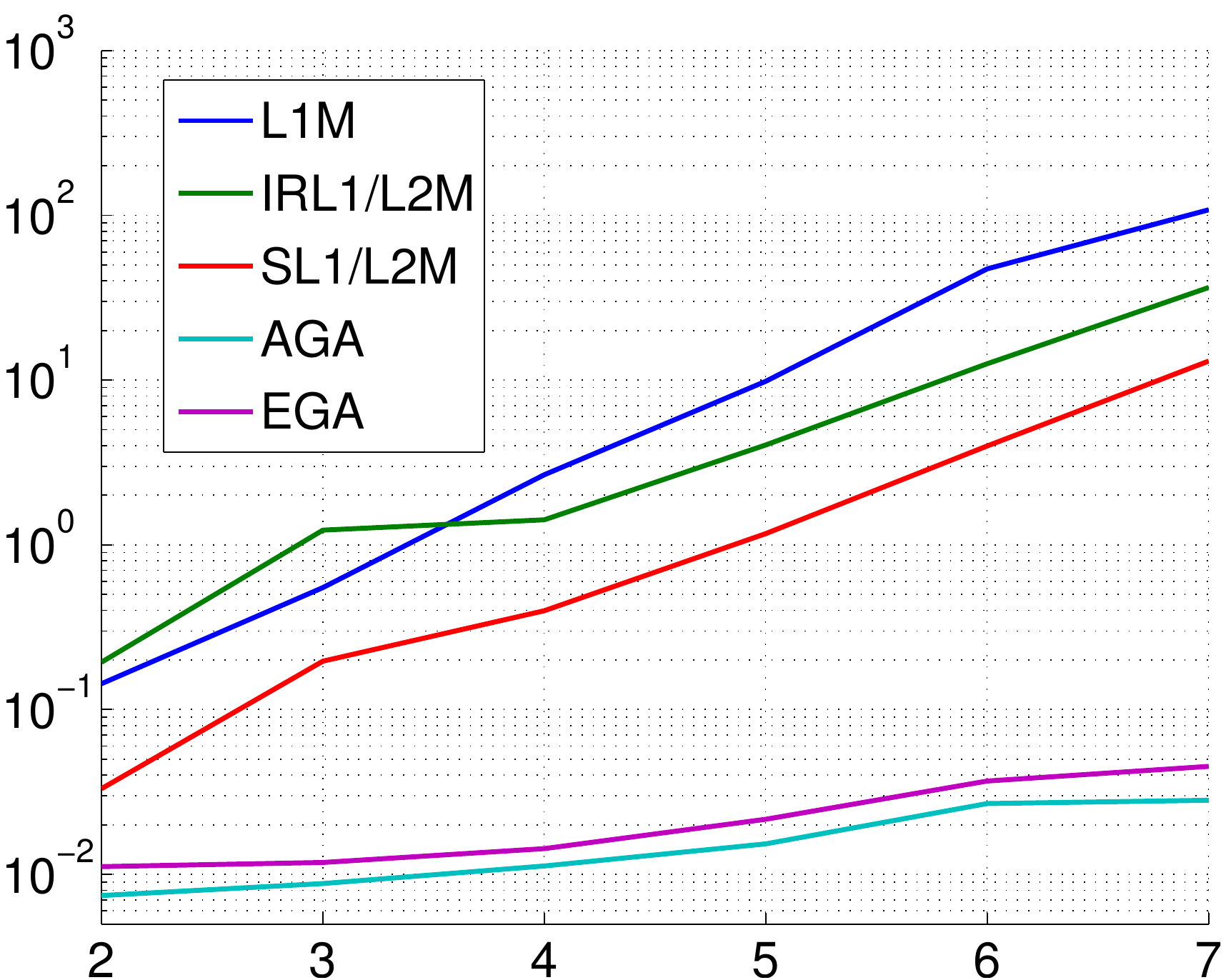}		
	\includegraphics[width=.32\linewidth]{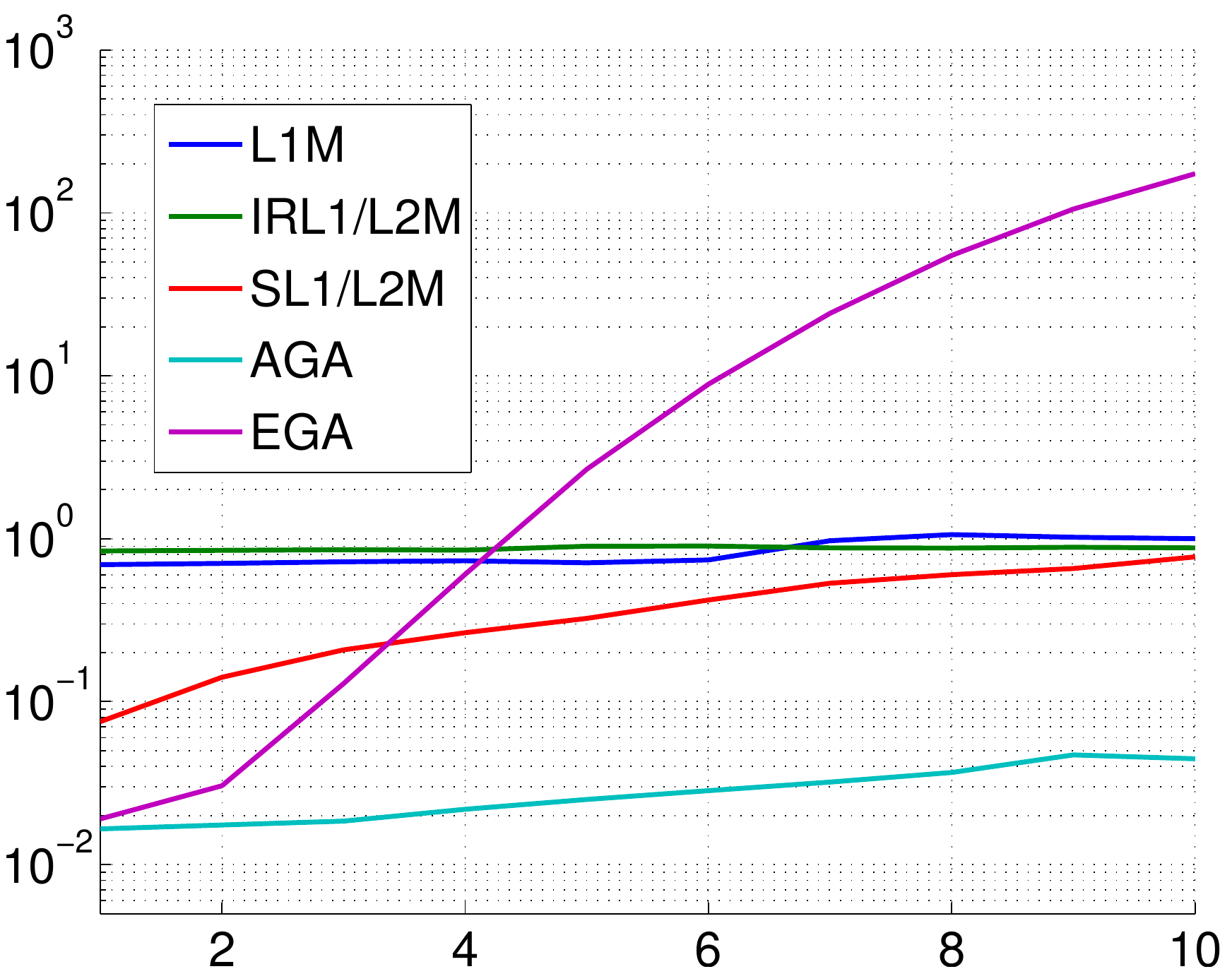}	
	\caption{Mean computing time of the methods in seconds versus the number of variables $n$ (left), the degree $d$ (middle) and the sparsity level $\|\g x_0\|_0$ (right). These times are obtained on a Linux server equipped with 2 Xeon X5690 at 3.47 GHz. \label{fig:time}}
\end{figure}

\subsection{Stable recovery in the presence of noise}
\label{sec:expnoise}

We now turn to the noisy case as discussed in Sect.~\ref{sec:noise}. 
In this setting, accuracy is measured by two performance indexes: 
the mean relative error, $\|\hat{\g x} -\g x_0\|_2 / \|\g x_0\|_2 $, and the success rate with respect to the estimation of the support of $\g x_0$ (where values $|\hat{x}_j|<10^{-6}$ are considered as zeros). Indeed, in many applications, the recovery of the correct support is more important that the precise values of the estimates. 

\paragraph{First results.}
The data are generated as in Sect.~\ref{sec:expnoiseless} except that $y_i = p_i^d(\g x_0) + e_i$, with noise terms $e_i$ forming a zero-mean Gaussian random vector $\g e$ satisfying $\|\g e\|_2 = \varepsilon = 3$. This corresponds to a signal-to-noise ratio, $\|\g y\|_2^2/\|\g e\|_2^2$, of about $18$ dB. 
Results with $N=50$ polynomials of degree $d=2$ in $n=20$ variables are shown in Table~\ref{tab:QBPnoise} and in Table~\ref{tab:NLBPnoise} for $d=4$ and $n=5$. The reweighted BP methods based on group-sparsity and the greedy methods lead to a rather small approximation error and they all recover the correct support in almost all trials despite the presence of noise. 

\begin{table}  
\centering
\caption{Results with noisy quadratic equations. \label{tab:QBPnoise}}
\begin{tabular}{llllllll}\hline\noalign{\smallskip}
	\bf Method & QBP & $\ell_1$M & $\ell_1\ell_2$M & IR$\ell_1\ell_2$M  & S$\ell_1\ell_2$M & AGA & EGA
	\\
\noalign{\smallskip}\hline\noalign{\smallskip}
	\bf Mean relative error & $16.1 \%$ &  $11.7 \%$ & $22.2 \%$ & $7.72 \%$ & $6.52 \%$ & $6.20 \%$ & $6.19 \%$ \\
	\bf Success rate    	& $14 \%$ &  $96 \%$ &  $0 \%$ & $100 \%$ & $100 \%$ & $99 \%$ &$100 \%$  \\
	\bf Mean time (s)   	& 2.02 &  0.57 & 0.13 & 0.89 & 0.32 & 0.02 & 0.15
\\	\hline\noalign{\smallskip}
\end{tabular}
%
\end{table}
\begin{table}
\centering
\caption{Results with noisy equations of degree $d=4$. \label{tab:NLBPnoise}}
\begin{tabular}{llllllll}\hline\noalign{\smallskip}
	\bf Method & NLBP  &  $\ell_1$M & $\ell_1\ell_2$M & IR$\ell_1\ell_2$M  & S$\ell_1\ell_2$M & AGA & EGA
	\\
\noalign{\smallskip}\hline\noalign{\smallskip}
	\bf Mean relative error & $16.5 \%$  & $22.3 \%$ & $29.8 \%$ & $7.65 \%$ & $5.83 \%$ & $6.74 \%$ &$5.84 \%$ \\
	\bf Success rate    	&  $55 \%$  & $87 \%$ &  $0 \%$ & $100 \%$ & $100 \%$ & $99 \%$ &$100 \%$  \\
	\bf Mean time (s)   	&  13.6 & 0.31 & 0.12 & 0.75 &  0.19 & 0.006 & 0.007
\\	\hline\noalign{\smallskip}
\end{tabular}
%
\end{table}

\paragraph{Influence of the noise level $\varepsilon$.}
The influence of the noise level $\varepsilon = \|\g e\|_2$ on the performance of the proposed methods is evaluated by letting the value of $\varepsilon$ vary between 1 and 10, which  corresponds to a signal-to-noise ratio decreased from 28 dB to 9 dB. For all methods except the $\ell_1$M, the results plotted in Fig.~\ref{fig:noiseVSsnr} indicate that, as expected, e.g., from the bound~\eqref{eq:stabilityxgroupW}, the approximation error directly depends on the noise level. In fact, the mean relative error is almost linear with respect to $\varepsilon$
, which shows that the dependence on $\varepsilon$ in the bound~\eqref{eq:stabilityxgroupW} of Theorem~\ref{thm:stabilitygroup} is of the correct nature. 
But even more interestingly, the noise level does not influence the success rate corresponding to the estimation of the support of $\g x_0$, thus providing evidence that the methods are robust to noise. 
Here again, the $\ell_1$M method, which does not include structural knowledge in its formulation, does not benefit from such satisfactory features.

\begin{figure}
\centering
	\includegraphics[width=.4\linewidth]{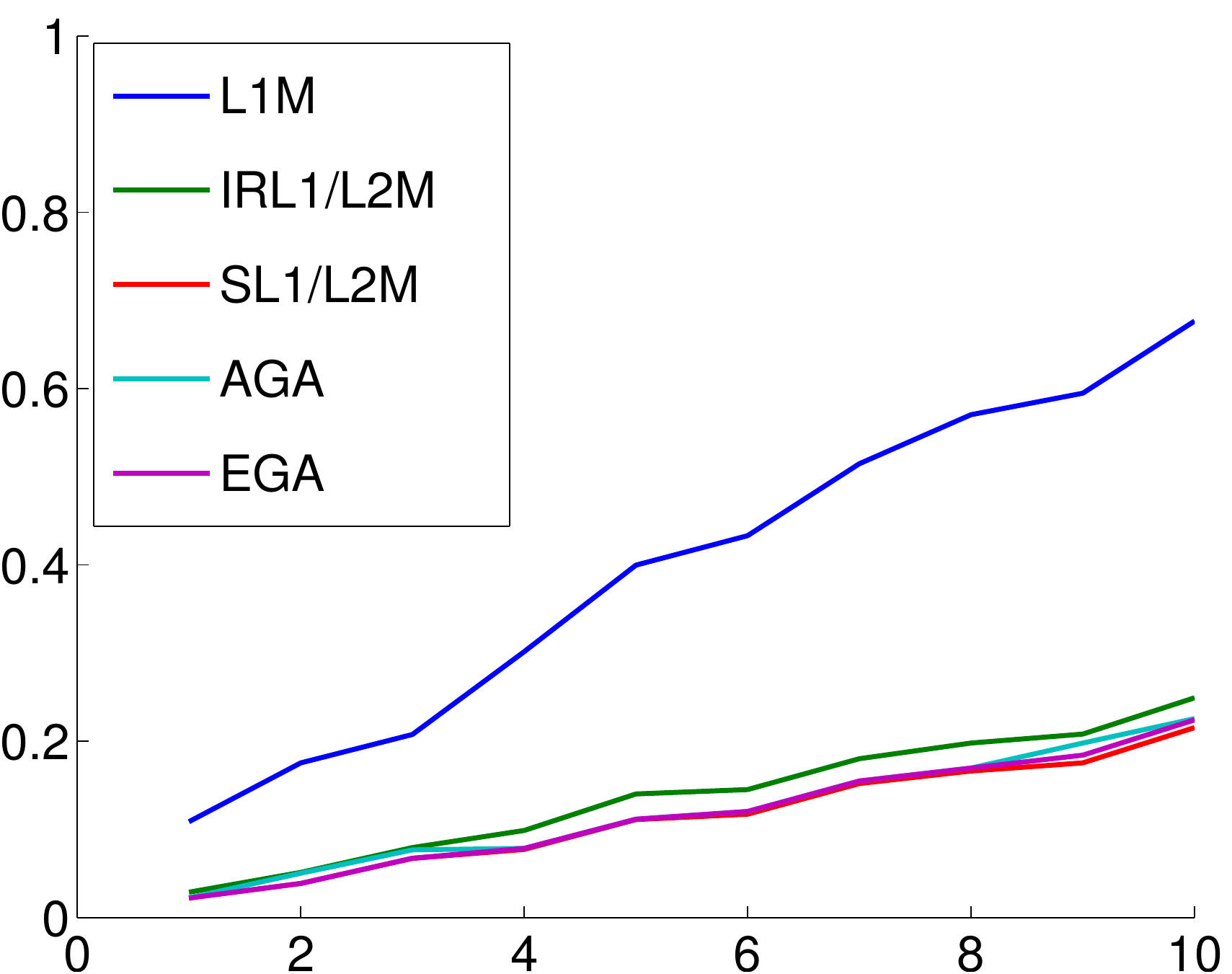}
	\includegraphics[width=.4\linewidth]{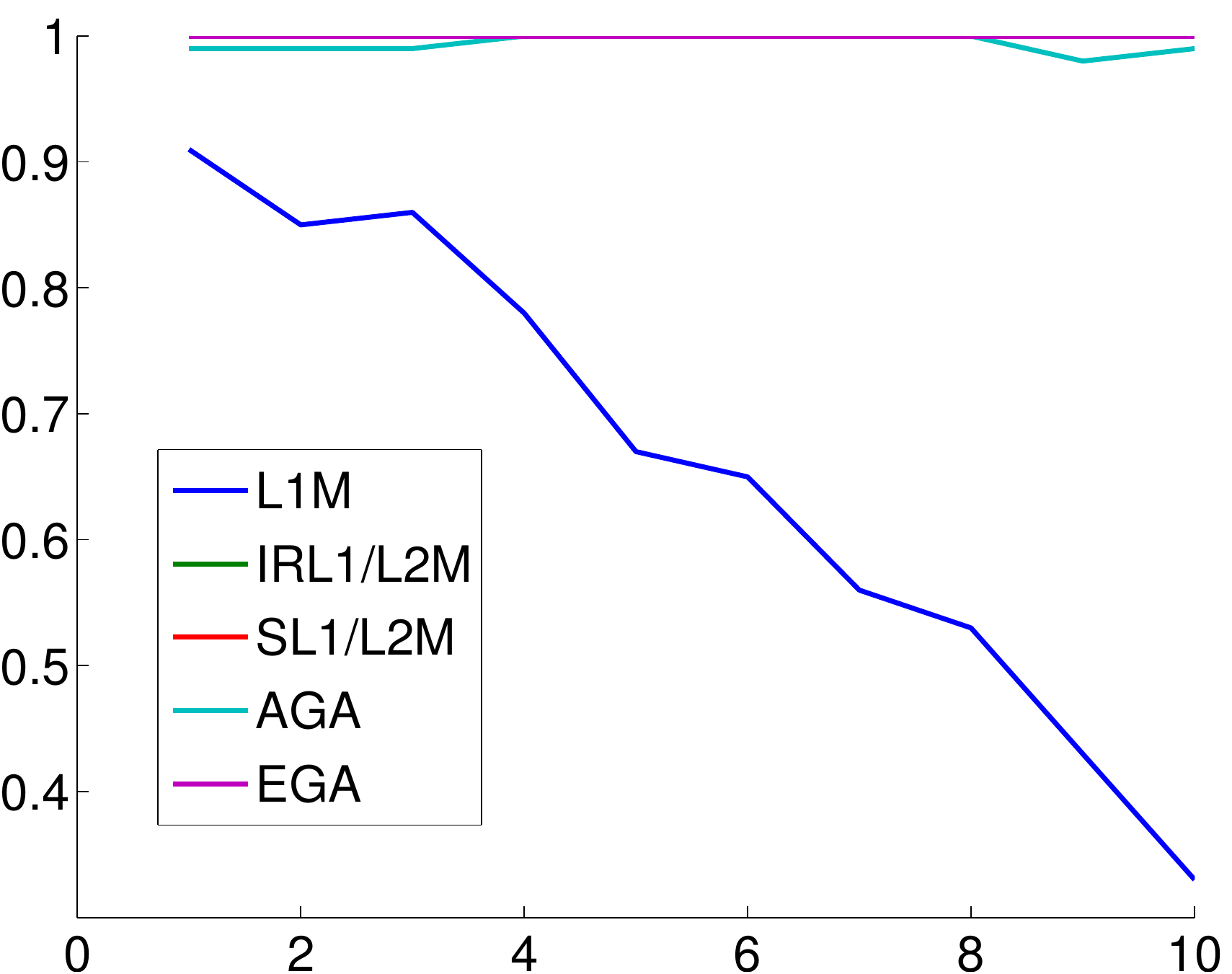}	
	\caption{Mean relative error (left) and success rate (right) versus $\varepsilon=\|\g e\|_2$. Except for the $\ell_1$M, the curves are hardly distinguishable and close to $100 \%$ for the success rate.\label{fig:noiseVSsnr}}
\end{figure}

\paragraph{Influence of $\varepsilon$ as a tuning parameter.}
In practical applications, the noise level might be unknown and $\varepsilon$ becomes a tuning parameter. 
Figure~\ref{fig:noiseVSepsilon} shows the influence of this parameter on the performance of the methods for a fixed noise level $\|\g e\|_2=3$. All methods except the $\ell_1$M perform very well with a slightly overestimated $\varepsilon \in [ 3, 4]$ and, for $\varepsilon$ within a reasonable range around $\|\g e\|_2$ ($\varepsilon\in[2,6]$), the methods still yield rather accurate estimates of $\g x_0$. 
Significantly larger values of $\varepsilon$ lead to an increase of the error for all methods except the greedy algorithms which maintain a mean relative error below $7\%$ for all values of $\varepsilon\geq \|\g e\|_2$. On the other hand, the IR$\ell_1\ell_2$M method is the only one that is not badly affected by the underestimation of $\varepsilon$ in terms of approximation error. 

Regarding the estimation of the support of $\g x_0$, all methods fail with $\varepsilon < \|\g e\|_2$, while overestimating $\varepsilon \geq \|\g e\|_2$ leads to perfect recovery even for much larger values of $\varepsilon$. This is in line with results on linear BP denoising, such as Theorem~4.1 in \cite{Donoho06noise} which guarantees that the estimated support is a subset of the sought one for sufficiently sparse cases when using an overestimated $\varepsilon$.

\begin{figure}
\centering
	\includegraphics[width=.4\linewidth]{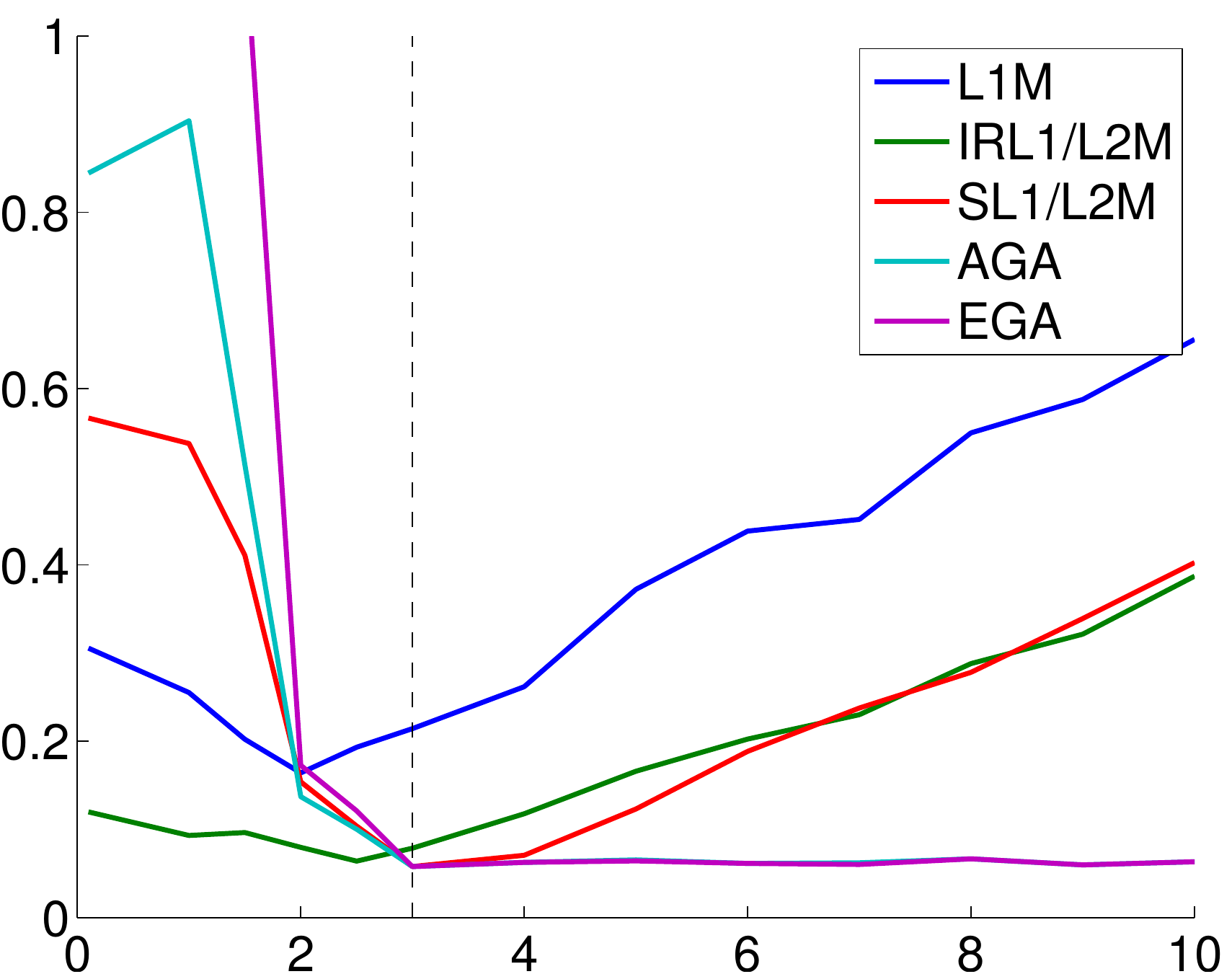}
	\includegraphics[width=.4\linewidth]{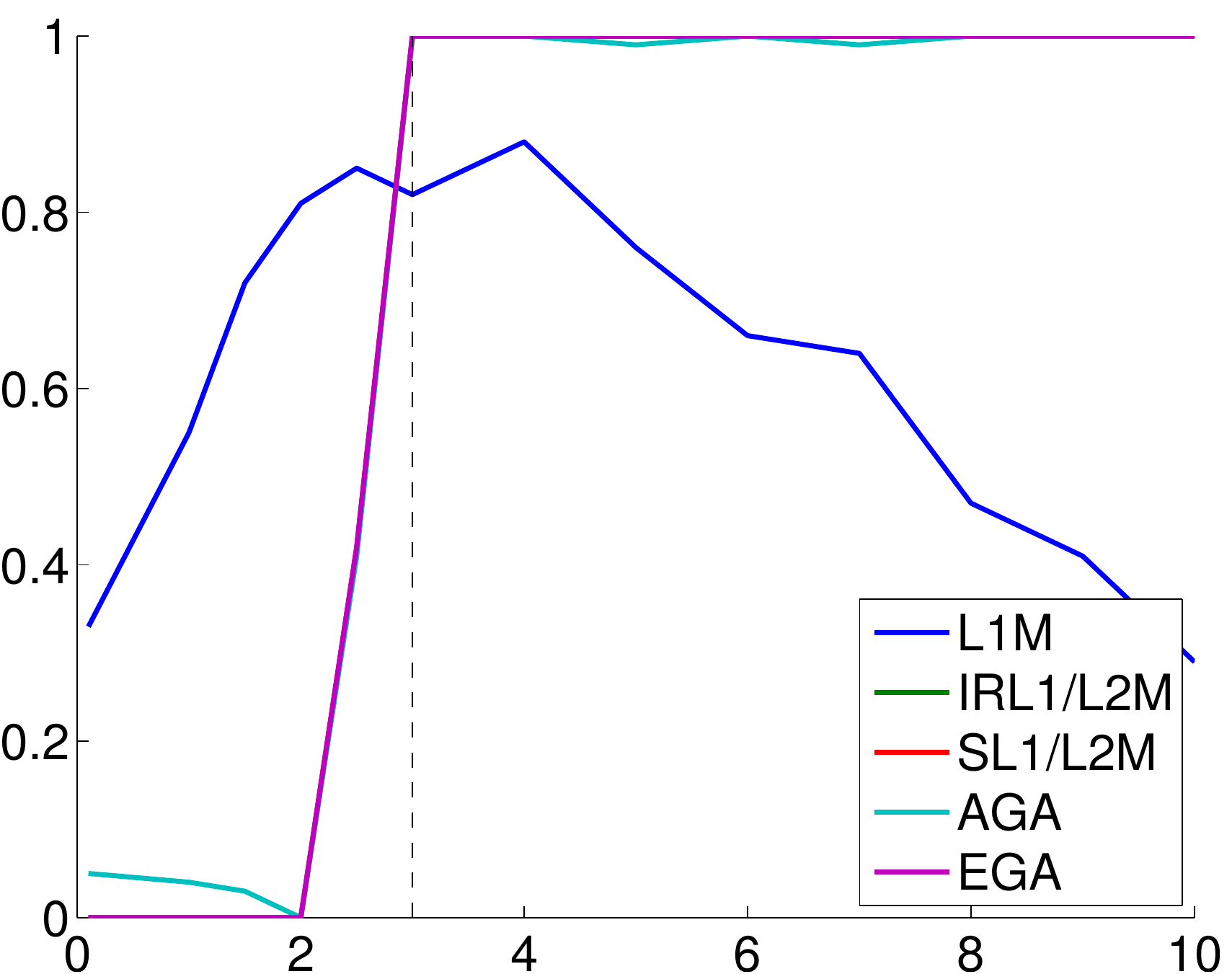}	
	\caption{Mean relative error (left) and success rate (right) versus $\varepsilon$ with data perturbed by a noise of $\ell_2$-norm equal to 3. Except for $\ell_1$M, the curves of the success rate are hardly distinguishable and close to $100 \%$ for $\varepsilon\geq 3$. \label{fig:noiseVSepsilon}}
\end{figure}

\section{Conclusions}

The paper proposed several methods for finding the sparsest solution of a system of polynomial equations. Two generic approaches were considered, one based on convex relaxations and one on a greedy strategy. For the convex relaxations, sufficient conditions of exact recovery of the sparsest solution were derived. 
The methods were also extended to deal with noisy equations, in which case stable recovery bounds for the convex relaxations were obtained. 
Both the computational efficiency and the accuracy of the proposed methods were shown in numerical experiments, which also emphasized the relationship between the probability of success and the sparsity of the solution for each method. In addition, these results indicate that the proposed methods accurately recover the sought solution in many cases where the sufficient conditions do not hold. As in classical BP theory, these conditions suffer from an ``excessive pessimism" and are too restrictive due to their worst-case nature. 

Remaining open issues include the convergence analysis of the greedy approximation towards the sparsest solution and the derivation of sufficient conditions for the solution of the group-sparsity optimization problem \eqref{eq:P0polygroup} to satisfy the polynomial constraints. Another question of interest is whether we can obtain less restrictive conditions for the variants of the convex relaxations with nonnegativity constraints such as~\eqref{eq:P1polygroup+}. 
Future work will also consider applying the proposed methods to more general nonlinear equations, e.g., by using Taylor expansions as in \cite{Ohlsson13nlbp}.

\subsection*{Acknowledgements}
Henrik Ohlsson gratefully acknowledges support from the NSF project FORCES (Foundations Of Resilient CybEr-physical Systems), the Swedish Research Council in the Linnaeus center CADICS, the European Research Council under the advanced grant LEARN, contract 267381, a postdoctoral grant from the Sweden-America Foundation, donated by ASEA's Fellowship Fund, and a postdoctoral grant from the Swedish Research Council.

\bibliographystyle{apalike}

\appendix

\section{Bound on the sparsity level of $\phi$}
\label{app:phisparsity}

\begin{lemma} \label{lem:binoms2}
For all $(a,b,d) \in (\mathbb{N}^*)^3$ such that $a\geq d(b+d)$, the inequality 
$$
	\frac{1}{a}\sum_{q=2}^d \begin{pmatrix}a + q-1\\q\end{pmatrix} \geq \frac{d}{b}  \sum_{q=2}^d\begin{pmatrix}b + q-1\\q\end{pmatrix} 
$$
holds.
\end{lemma}
\begin{proof}
For $q\leq d$, we can bound the terms in the sum as 
\begin{align*}
	\frac{1}{a} \begin{pmatrix}a + q-1\\q\end{pmatrix} &= \frac{(a+q-1)!}{a\, q! (a-1)!} = \frac{1}{a\,q!}\prod_{i=0}^{q-1}(a+i) 
	= \frac{1}{q!}\prod_{i=1}^{q-1}(a+i)\\
	&\geq \frac{1}{q!} a^{q-1} \\
	&\geq \frac{1}{q!} d^{q-1}(b+d)^{q-1} \\
	&\geq \frac{1}{q!} d^{q-1} \prod_{i=1}^{q-1} (b+i) \\
	&\geq \frac{d^{q-1} }{b} \begin{pmatrix}b + q-1\\q\end{pmatrix} 
\end{align*}
where we used $a\geq d(b+d)$ in the second inequality. Then
\begin{align*}
	\frac{1}{a}\sum_{q=2}^d \begin{pmatrix}a + q-1\\q\end{pmatrix} 
		&\geq \frac{1}{b}  \sum_{q=2}^d d^{q-1} \begin{pmatrix}b + q-1\\q\end{pmatrix} \\
		&\geq  \frac{d}{b} \sum_{q=2}^d d^{q-2} \begin{pmatrix}b + q-1\\q\end{pmatrix} \\
		&\geq  \frac{d}{b} \sum_{q=2}^d \begin{pmatrix}b + q-1\\q\end{pmatrix} 		
\end{align*}

\end{proof}

\begin{proposition} \label{prop:phisparsity2}
Let the mapping $\phi : \R^n \rightarrow \R^M$ be defined as above. Then, with $d\geq 3$ and $n\geq d(\|\g x\|_0 + d)$, the vector $\phi(\g x)$ is sparser than the vector $\g x$ in the sense that the inequality
$$
	\frac{\|\phi(\g x)\|_0}{M} \leq \frac{2 \|\g x\|_0 }{d n} 
$$
holds for all $\g x\in\R^n$. 
\end{proposition} 
\begin{proof}
By construction, the number of nonzeros in $\phi(\g x)$ is equal to the sum over $q$, $1\leq q\leq d$, of the number of monomials of degree $q$ in $\|\g x\|_0$ variables:
\begin{align*}
	\frac{\|\phi(\g x)\|_0}{M} &=  \frac{ \sum_{q=1}^d\begin{pmatrix}\|\g x\|_0 + q-1\\q\end{pmatrix}}{ \sum_{q=1}^d\begin{pmatrix}n + q-1\\q\end{pmatrix}} 
	=  \frac{ \|\g x\|_0 }{n} \frac{  \frac{1}{\|\g x\|_0 } \sum_{q=1}^d\begin{pmatrix}\|\g x\|_0 + q-1\\q\end{pmatrix}}{  \frac{1 }{n} \sum_{q=1}^d\begin{pmatrix}n + q-1\\q\end{pmatrix}} \\
	&= \frac{ \|\g x\|_0 }{n} \frac{ 1 +  \frac{1}{\|\g x\|_0 } \sum_{q=2}^d\begin{pmatrix}\|\g x\|_0 + q-1\\q\end{pmatrix}}{1+  \frac{1 }{n} \sum_{q=2}^d\begin{pmatrix}n + q-1\\q\end{pmatrix}} \\
	&= \frac{ \|\g x\|_0 }{n} \left[\frac{ 1}{1+  \frac{1 }{n} \sum_{q=2}^d\begin{pmatrix}n + q-1\\q\end{pmatrix}}   +    \frac{ \frac{1}{\|\g x\|_0 } \sum_{q=2}^d\begin{pmatrix}\|\g x\|_0 + q-1\\q\end{pmatrix}} {1+ \frac{1 }{n}\sum_{q=2}^d\begin{pmatrix}n + q-1\\q\end{pmatrix}}\right]
\end{align*}
The assumption $n\geq d(\|\g x\|_0 + d)$ implies that\footnote{Assume $d> n$, then, the assumption of the Proposition leads to $n> n(\|\g x\|_0 + n)$ and $\|\g x\|_0 + n< 1$ which is impossible since $n\geq 1$.}  $d\leq n$. With $d\leq n$, we have 
$$
	\frac{1 }{n} \begin{pmatrix}n + q-1\\q\end{pmatrix} \geq \frac{1}{q!} n^{q-1} \geq \frac{1}{q!} d^{q-1}
$$
which yields
$$
	\frac{1 }{n} \sum_{q=2}^d\begin{pmatrix}n + q-1\\q\end{pmatrix} \geq \frac{1 }{n}\sum_{q=2}^d \begin{pmatrix}n + 2-1\\2\end{pmatrix} \geq (d-1)\frac{d}{2}
$$
Now, on the one hand we have
\begin{align*}
	d\geq 3 
	\Rightarrow\  & \frac{1}{2}d^2 -\frac{3}{2}d + 1 \geq 0 \\
	\Rightarrow\  & 1 + \frac{1 }{n} \sum_{q=2}^d\begin{pmatrix}n + q-1\\q\end{pmatrix}  \geq d \\	
	\Rightarrow\ & \frac{ 1}{1+  \frac{1 }{n} \sum_{q=2}^d\begin{pmatrix}n + q-1\\q\end{pmatrix}}  \leq \frac{1}{d} 
\end{align*}
and on the other hand, Lemma~\ref{lem:binoms2} yields
$$
	 \frac{ \frac{1}{\|\g x\|_0 } \sum_{q=2}^d\begin{pmatrix}\|\g x\|_0 + q-1\\q\end{pmatrix}} {1+\frac{1 }{n} \sum_{q=2}^d\begin{pmatrix}n + q-1\\q\end{pmatrix}} \leq  \frac{ \frac{1}{\|\g x\|_0 } \sum_{q=2}^d\begin{pmatrix}\|\g x\|_0 + q-1\\q\end{pmatrix}} {\frac{1 }{n}\sum_{q=2}^d\begin{pmatrix}n + q-1\\q\end{pmatrix}} \leq \frac{1}{d}
$$
Thus, 
$$
	\frac{\|\phi(\g x)\|_0}{M} \leq \frac{2 \|\g x\|_0 }{d n} 
$$

\end{proof}

\section{Other conditions for sparse recovery}
\label{app:otherconditions}

The following uses the exact value of $\|\g \phi_0\|_0$.
\begin{theorem}\label{th:2}
	Let $\g x_0$ denote the unique solution to \eqref{eq:P0}--\eqref{eq:polyf}. If the inequality 
	\begin{equation}
		\sum_{q=1}^d\begin{pmatrix}\|\g x_0\|_0 + q-1\\q\end{pmatrix} \leq \frac{1}{2}\left(1+ \frac{1}{\mu(\g A)}\right) 
	\end{equation}
	holds, then the solution $\hat{\g \phi}$ to \eqref{eq:P1poly} is unique and equal to $\phi(\g x_0)$, thus providing $\hat{\g x} = \g x_0$. 
\end{theorem}
Another more compact but slightly less tight result is as follows.
\begin{theorem}
	Let $\g x_0$ denote the unique solution to \eqref{eq:P0}--\eqref{eq:polyf}. If the inequality 
	\begin{equation}
		\begin{pmatrix}\|\g x_0\|_0 + d-1\\d\end{pmatrix} \leq \frac{1}{2d}\left(1+ \frac{1}{\mu(\g A)}\right) 
	\end{equation}
	holds, then the solution $\hat{\g \phi}$ to \eqref{eq:P1poly} is unique and equal to $\phi(\g x_0)$, thus providing $\hat{\g x} = \g x_0$. 
\end{theorem}
\begin{proof}
Since the terms in the sum of Theorem~\ref{th:2} form an increasing sequence, we have
$$
	\sum_{q=1}^d\begin{pmatrix}\|\g x_0\|_0 + q-1\\q\end{pmatrix} \leq d \begin{pmatrix}\|\g x_0\|_0 + d-1\\d\end{pmatrix} 
$$
which yields the sought statement by application of Theorem~\ref{th:2}.

\end{proof}

\section{Value of $m$}
\label{app:m}

Let us define $m$ as the number of monomials involving a base variable $x$ with a degree $\geq 1$. It can be computed as the sum over $q$, $0\leq q\leq d-1$ of the number of monomials of degree $q$ in $n-1$ variables times the remaining degree $d-q$ (since each monomial in $n-1$ variables can be multiplied by $x$ or $x^2$... or $x^{d-q}$ to produce a monomial of degree at most $d$ in $n$ variables): 
$$
	m = \sum_{q=0}^{d-1} (d-q)\begin{pmatrix}(n-1) + q-1\\q\end{pmatrix} = \sum_{q=0}^{d-1} (d-q)\begin{pmatrix}n + q-2\\q\end{pmatrix}  .
$$
Another technique computes $m$ as the total number of all monomials minus the number of monomials not involving $x$ which is the number of monomials in $n-1$ variables:
\begin{align*}
	m & = M -  \sum_{q=1}^{d} \begin{pmatrix}n + q-2\\q\end{pmatrix} \\
	 & =   \sum_{q=1}^{d}\begin{pmatrix}n + q-1\\q\end{pmatrix}  -  \begin{pmatrix}n + q-2\\q\end{pmatrix} 	\\
	& = \sum_{q=1}^{d} \left[\frac{(n + q-1)}{n-1} -1\right] \begin{pmatrix}n + q-2\\q\end{pmatrix}\\
	&=\sum_{q=1}^{d} \frac{q}{n-1} \begin{pmatrix}n + q-2\\q\end{pmatrix} .
\end{align*}


\end{document}